\documentclass[12pt]{article}
\usepackage{setspace}
\usepackage[utf8]{inputenc}
\usepackage[english]{babel}  
\usepackage[colorlinks=true, linkcolor=blue, citecolor = blue]{hyperref} 
\usepackage[ruled,vlined]{algorithm2e}    
\usepackage{algorithmic}
\usepackage{url}  
\usepackage{booktabs}       
\usepackage{amsthm}
\usepackage{enumitem}
\usepackage{ragged2e}
\usepackage{makecell}
\newtheorem{theorem}{Theorem}
\usepackage{amsfonts}       
\usepackage{nicefrac}       
\usepackage{microtype}
\usepackage{authblk}
\usepackage{prodint}
\usepackage[top=1 in,bottom=1in, left=1 in, right=1 in]{geometry}
\usepackage{graphicx}
\usepackage{subfig}
\usepackage{rotating}
\usepackage{float}
\usepackage{wrapfig}
\usepackage{enumitem}
\usepackage{amsmath}
\usepackage{amssymb}
\usepackage[round]{natbib}
\usepackage{tabularx}
\usepackage[export]{adjustbox}
\usepackage{multicol}
\usepackage{color}
\usepackage{multirow}
\usepackage[thinlines]{easytable}

\newtheorem{lemma}{Lemma}

\newtheorem{assumption}{Assumption}

	 \newtheorem{remark}{Remark}

\theoremstyle{remark}

\usepackage{multirow}
\usepackage[table,xcdraw]{xcolor}
\hypersetup{
    colorlinks=lrue,
    linkcolor=blue,
    filecolor=blue,      
    urlcolor=blue,
}
\usepackage{tcolorbox}
\usepackage{tikz}

\usepackage{fancyhdr}
\newcommand\independent{\protect\mathpalette{\protect\independenT}{\perp}}
\def\independenT#1#2{\mathrel{\rlap{$#1#2$}\mkern2mu{#1#2}}}

\usepackage{epsfig}
\usepackage{threeparttable}
\graphicspath{{Images/}}
\title{%
  \rule{\linewidth}{0.8pt}\\[5pt]
  {\Large\textbf{Optimal Treatment Policy Estimation for Recurrent Events with a Competing Terminal Event: An Instrumented Difference-in-Differences Approach}}\\[3pt]
  \rule{\linewidth}{0.8pt}%
}
\author[1]{\small Ritoban Kundu\thanks{\small Corresponding Author. Email: \texttt{ritoban.kundu@pennmedicine.upenn.edu}}}
\author[2]{\small James Flory}
\author[1]{\small Sean Hennessy}
\author[1]{\small Ashkan Ertefaie}
\affil[1]{\small Department of Biostatistics, Epidemiology and Informatics, Perelman School of Medicine, University of Pennsylvania}
\affil[2]{\small Department of Subspecialty Medicine, Memorial Sloan Kettering Cancer Center}

\date{}

\begin{document}
\maketitle
\begin{abstract}
\small
\noindent
    Learning reproducible and generalizable optimal treatment policies for chronic diseases requires large, representative populations with long-term follow-up. Administrative health data provide a natural starting point, but their use is often limited by unmeasured confounding. We address this by proposing a novel framework based on Instrumented Difference-in-Differences (iDID) to estimate optimal policies for recurrent event outcomes subject to a terminating event. The iDID design is particularly useful in this setting because it leverages policy-induced treatment variation while allowing for persistent unmeasured differences across populations, relying on assumptions that are more plausible for administrative health data than those required by conventional IV or DID approaches. A key feature of our approach is that it explicitly addresses the fundamental challenge of avoiding policies that trivially reduce recurrent adverse events by increasing mortality. We derive two distinct Inverse Probability Weighted identifications and develop a multiply robust estimator that achieves consistency if any one of several subsets of nuisance models is correctly specified. We establish the estimator’s consistency and asymptotic normality through large-sample theory and demonstrate its superior finite-sample performance over existing methods via simulation. Finally, we apply this framework to a national Medicare dataset to optimize first-line Type 2 Diabetes strategies, specifically targeting the minimization of disease-related hospitalizations while accounting for survival.
\end{abstract}
\noindent%

\vfill
\doublespacing
\section{Introduction}
The development of data-driven individualized optimal treatment policies has become a cornerstone of modern statistical inquiry. The primary objective is to recover a decision rule that assigns the optimal treatment, from a set of available options, to each patient based on their specific characteristics. These methods have received increasing interest across diverse fields, including precision medicine \citep{laber2024handbook}, econometrics and social sciences \citep{belloni2017program, athey2021policy}, and computer science and operations research \citep{kallus2022efficiently, shi2024off}. Recent literature has expanded these methods across various data structures, ranging from foundational randomized experiments \citep{kosorok2015adaptive, tsiatis2019dynamic} to large-scale observational studies \citep{wager2017efficient, kallus2018balanced} and electronic health records  \citep{wang2016learning, wu2020matched}. This evolution has established prominent methodological paradigms: indirect, model-based approaches like Q-learning \citep{qian2011performance,ertefaie2021robust} and A-learning \citep{murphy2003optimal, shi2018high}, alongside direct policy search methods framed as weighted classification tasks \citep{zhang2012estimating, chernozhukov2019semi}. Learning the optimal treatment policy is closely linked to the estimation of conditional average treatment effects, an area that has recently experienced significant growth \citep{hatt2202combining,demirel2024prediction}. \\

\noindent
While there have been significant advances in developing optimal treatment policies for survival outcomes \citep{goldberg2012q, cui2017tree, jiang2017estimation, zhao2025efficient}, relatively little attention has been given to recurrent event processes. This is an important gap, as accounting for the recurrent nature of adverse events will allow us to better assess the
overall disease burden, and is of critical importance for clinical decision making. This is particularly evident in chronic disease management, such as Type 2 diabetes mellitus (T2DM), where patients experience recurrent complications including hospitalizations, cardiovascular events, and kidney disease. Dilemmas in T2DM management illustrate the need for real-world evidence to guide individualized first-line treatment selection: the choice between metformin and glucagon-like peptide-1 (GLP-1) receptor agonists involves complex tradeoffs between efficacy, safety and comorbidity considerations that are difficult to resolve through clinical trials alone. Recent work on C-learning \citep{zhan2025c} accommodates recurrent events but assumes away administrative censoring and terminal events. In many applications, particularly those involving Medicare claims follow up is truncated by death, which induces informative censoring. There is a well-developed literature on modeling the expected number of recurrent events prior to failure \citep{cook2007statistical, zhao1997consistent, ghosh2000nonparametric, cook2009robust}. Causal work in this area has focused primarily on estimating average treatment effects for recurrent outcomes in the presence of terminal death \citep{schaubel2010estimating, janvin2024causal, su2020doubly, baer2023causal}. However, the problem of optimal policy learning in this setting remains largely unexplored, leaving a critical gap between causal effect estimation and decision-making.\\

\noindent
In observational settings such as large-scale Medicare data, unmeasured confounding remains the primary obstacle to credible policy learning, as it renders the ignorability assumption unverifiable. The foundational works of \citet{angrist1995identification, angrist1996identification} introduced a counterfactual framework for instrumental variables (IVs), identifying the complier average causal effect under monotonicity. Subsequent developments \citep{tan2006regression, ogburn2015doubly} provided semiparametric estimators for these effects. For policy learning, however, the population average treatment effect is typically the relevant target rather than a complier-specific estimand \citep{hernan2006instruments, aronow2013beyond}. Results by \citet{wang2018bounded} on point identification of the population average treatment effect underpin recent IV-based approaches to optimal treatment policy estimation \citep{cui2021semiparametric}. At the same time, the growing availability of longitudinal data, including administrative claims and electronic health records, has increased the appeal of designs that exploit temporal structure. Difference-in-differences (DiD) methods identify treatment effects under a parallel trends assumption \citep{abadie2005semiparametric, sant2020doubly, roth2023s}, but this assumption is often violated due to time-varying unmeasured confounding. This has motivated extensions such as changes-in-changes \citep{athey2006identification}, sensitivity analyses \citep{keele2019patterns}, and negative control approaches \citep{sofer2016negative}. Moreover, DiD typically targets the average treatment effect on the treated, which may not generalize to the full population and thus limits its use for policy learning. Beyond IV and DiD, recent work has explored proximal inference \citep{miao2018identifying} and double negative control methods \citep{miao2024confounding} as alternative strategies for addressing unmeasured confounding. The instrumented difference-in-differences (iDID) framework combines instrumental variables and difference-in-differences paradigms to enable identification when the parallel trends assumption is violated \citep{ye2023instrumented, vo2024structural}. It further permits the instrument to have a direct effect on the outcome, provided it does not modify the outcome trend. Within this setting, \citet{zhao2025semiparametric} developed an iDID-based approach for learning optimal treatment policies with continuous outcomes.\\ 

\noindent
Our work substantially expands the existing literature through four main contributions. First, while existing methods focus on continuous outcomes, we develop a framework for optimal treatment policies with recurrent events in the presence of a terminal event, a setting that is pervasive in chronic disease applications but largely unaddressed. This extension is nontrivial: terminal mortality induces competing risks that can yield degenerate ``optimal” policies that reduce recurrent events only by increasing death. We address this issue by formulating the policy objective as a constrained optimization problem that explicitly penalizes such solutions. Second, we propose a multiply robust estimation procedure for the optimal treatment policy that remains consistent if any one of several subsets of nuisance models is correctly specified, providing protection against model misspecification in this constrained setting. Third, we establish the estimator’s large-sample properties and show through simulations that it achieves improved finite-sample performance relative to existing approaches. Finally, we apply our method to a national Medicare dataset to estimate optimal first-line treatment strategies for Type 2 diabetes, targeting reductions in disease-related hospitalizations while appropriately accounting for competing mortality risk in presence of potential unmeasured confounding. Large-scale Medicare data can address clinically important treatment questions that are too numerous and complex for a feasible trial agenda, enabling us to estimate optimal first-line treatment strategies for Type 2 diabetes targeting reductions in disease-related hospitalizations while appropriately accounting for competing mortality risk.

\section{Notation and Framework}
We first introduce the notation for the proposed framework. Let $A \in \{0, 1\}$ represent a binary treatment variable and $Z \in \{0, 1\}$ a binary instrumental variable. We define $L \in \{0, 1\}$ as a binary period indicator representing whether a unit is from time period $L=1$ or $L=0$. Let $\boldsymbol W \in \mathcal{W}\subset\mathbb{R}^p$ denote a vector of measured pre-IV covariates. The outcome process is characterized by both recurrent events and a terminal event. Let $D$ denote the time to death (the terminating event).  We define $N^*(\cdot)$ as the underlying right-continuous recurrent event process. To account for the terminating nature of $D$, we define the process as $N^*(t) = N^*(t \wedge D)$, ensuring no events are recorded after the terminal event occurs. In practice, these processes are subject to a censoring variable $C$, which represents the time at which observation ceases for reasons independent of the terminal event, such as administrative censoring or loss to follow-up. The observed recurrent event process is $N(t) = N^*(t \wedge C)$, and the observed follow-up time is $X = D \wedge C$. We define $\Delta = I(D \le C)$ as the failure indicator. Moreover, let $\boldsymbol U$ denote a vector of unmeasured confounders that affects the relationship between $A$ and $D$, and between $A$ and $N^*(\cdot)$. Let $0 < t_1 < t_2 < \dots < t_m$ denote the landmark times at which we observe the process $N(\cdot)$ for both time periods $L=0$ and $L=1$. The observed data is thus given by $O = (\boldsymbol W, A, Z, L, X, \{N(t)\}_{t=t_1}^{t_m}, \Delta)$. We assume that $O_1, O_2, \dots, O_n$ are independent and identically distributed (i.i.d.) realizations of $O$. Let $Y^*(t):=\mathbb{I}(D\geq t)$ denote the true at risk process and $Y(t):=\mathbb{I}(X\geq t)$ denote the observed at risk process.\\

\noindent
We define the following potential outcomes. Let $A^{(z)}_l$ denote the potential exposure in period $L=l$ if the instrument were set to $z$. Similarly, let $N^{*(a,z)}_l(t)$ and $D^{(a,z)}_l$ denote the potential recurrent event process at time $t$ and the potential time to the terminal event, respectively, for a unit in period $L=l$ given treatment $a$ and instrument $z$. These potential outcomes account for the terminating nature of $D^{(a,z)}_l$ such that $N^{*(a,z)}_l(t) = N^{*(a,z)}_l(t \wedge D^{(a,z)}_l)$. Furthermore, let $N^{*(a)}_l(t)$ denote the potential recurrent event process at time $t$ if the exposure were set to level $a$, while the instrument $Z$ remains at its observed value for a unit in period $l$. Similarly, let $D^{(a)}_l$ denote the potential time to the terminal event under these same conditions. Finally, we define $Y^{*(a,z)}_l(t) = I(D^{(a,z)}_l \geq t)$ and $Y^{*(a)}_l(t) = I(D^{(a)}_l \geq t)$ as the potential at-risk indicators at time $t$ under the respective treatment and instrument assignments.\\

\noindent
Our primary objective is to identify an optimal policy, $d^t: \mathcal{W} \mapsto \{0, 1\}$, that minimizes the expected cumulative number of recurrent events by time $t$. However, in the presence of a terminal event such as mortality, unconstrained optimization may yield degenerate policies that artificially reduce the recurrent event burden simply by increasing the competing risk of death. To preclude this, we restrict our search to a class of admissible policies satisfying a clinically motivated survival constraint: we require that the marginal survival probability under the optimal policy is bounded below by the survival probability under a designated baseline behavioral policy, denoted $\widetilde{d}(\boldsymbol{W})$. 
Specifically, the optimization problem is formulated as:
\begin{align}
d^t_{\text{opt},l}=\arg\min_{d^t \in \mathcal{D}} \mathbb{E} \left[ N_l^{*(d^t(\boldsymbol W))}(t) \right]\quad \text{subject to } \mathbb{E}[ Y_l^{*(d^t(\boldsymbol W))}(t)-Y_l^{*(\widetilde{d}(\boldsymbol W))}(t)]>0\cdot\label{eq:eq2}
\end{align}
This behavioral policy $\widetilde{d}(\boldsymbol{W})$ represents the observed standard of care and is defined via a threshold on the pseudo-propensity score, $\widetilde{d}(\boldsymbol{W}) = \mathbb{I}\{\widetilde{\pi}(\boldsymbol{W}) > c\}$, where $c$ is a context-dependent threshold and $\widetilde{\pi}(\boldsymbol{W}) \equiv P(A=1 \mid \boldsymbol{W})$ denotes the propensity score computed solely from the observed covariates $\boldsymbol{W}$. This quantity must be carefully distinguished from the true propensity score $P(A=1 \mid \boldsymbol{W}, \boldsymbol{U})$, which conditions additionally on the unmeasured confounders $\boldsymbol{U}$ and is fundamentally unidentifiable from the observed data. The pseudo-propensity score $\widetilde{\pi}(\boldsymbol{W})$ is therefore a misspecified surrogate that ignores the unobserved heterogeneity induced by $\boldsymbol{U}$. Throughout, we treat $\widetilde{\pi}(\boldsymbol{W})$ as pre-specified by the practitioner, which is natural in settings where treatment allocation is governed by an institutional protocol or administrative decision based solely on $\boldsymbol{W}$. Fixing $\widetilde{\pi}(\boldsymbol{W})$ rather than estimating it serves a deliberate purpose: it anchors the survival constraint to a well-defined, externally specified reference policy, thereby ensuring that the constraint is interpretable and remains invariant to the estimation of nuisance parameters. Importantly, all identification and asymptotic results developed in this paper are derived using the dichotomized rule $\widetilde{d}(\boldsymbol{W})$, but apply equally when the constraint is instead expressed directly in terms of the pseudo-propensity score $\widetilde{\pi}(\boldsymbol{W})$ without dichotomization. 
\section{Methodology}
\subsection{Identification Assumptions}
In this section, we specify the key assumptions required to identify the optimal treatment policies within the constrained optimization problem \eqref{eq:eq2} under unmeasured confounding using an instrumented differences-in-differences (iDID) approach. We impose the following identification assumptions,
\begin{assumption}[Consistency]\label{ass:ass1}
    For all $0\leq t<\infty$ and $l=0,1$, $A=A^{(Z)}_l$, $N_l^{*}(t)=N_l^{*(A)}(t)$ and $Y_l^{*}(t)=Y_l^{*(A)}(t)$.
\end{assumption}
\noindent
Assumption \ref{ass:ass1} incorporates the Stable Unit Treatment Value Assumption (SUTVA), meaning an individual's observed outcome is not affected by others' exposure levels (no interference) or by the individual's own exposure level at a different time point.
\begin{assumption}[Positivity]\label{ass:ass2}
    For any $l=0,1$ and $z=0,1$, $0<P(L=l,Z=z|\boldsymbol W)<1$, almost surely.
\end{assumption}
\noindent
Assumption \ref{ass:ass2} postulates that there is a positive probability of receiving each $(l, z)$ combination within each level of $\boldsymbol W$. Equivalently, it ensures the support of $\boldsymbol W$ is the same for each level of $(L, Z)$.
\begin{assumption}[Random Sampling]\label{ass:ass3}
   For all $0\leq t<\infty$ and any $l=0,1$, $z=0,1$, $a=0,1$, $L\independent \{A_l^{(z)},N_l^{*(a)}(t),Y_l^{*(a)}(t)\}|Z,\boldsymbol W$.
\end{assumption}
\noindent
Assumption \ref{ass:ass3} is often assumed for repeated cross-sectional datasets and states that for each level of $(Z, \boldsymbol W)$, the collected data at every time point is a random sample from the underlying population. 
\begin{assumption}[Trend Relevance]\label{ass:ass4}
    $\mathbb{E}[A_1^{(1)}-A_0^{(1)}|Z=1,\boldsymbol W]\neq \mathbb{E}[A_1^{(0)}-A_0^{(0)}|Z=0,\boldsymbol W]$ almost surely.
\end{assumption}

\begin{assumption}[Independence and Exclusion Restriction]\label{ass:ass5}
   For all $0\leq t<\infty$ and any $l=0,1$, $Z\independent \{A_l^{(0)},A_l^{(1)},N_1^{*(0)}(t)-N_0^{*(0)}(t),N_l^{*(1)}(t)-N_l^{*(0)}(t),Y_1^{*(0)}(t)-Y_0^{*(0)}(t),Y_l^{*(1)}(t)-Y_l^{*(0)}(t)\}|\boldsymbol W$.
\end{assumption}
\noindent
Assumptions \ref{ass:ass4} and \ref{ass:ass5} are parallel to the core assumptions required for instrumental variables in a Difference-in-Differences framework. Assumption \ref{ass:ass4} states that the instrument $Z$, acting as an encouragement that disproportionately affects a subpopulation, changes the temporal trend in exposure. Assumption \ref{ass:ass5} posits that the instrument is independent of potential treatment trends and that any direct effect of $Z$ on the outcomes is period-invariant, allowing it to be canceled out through differencing. This highlights the primary advantage of employing $Z$ as an instrument in the iDID framework compared to a standard IV approach. In this setting, the instrument is permitted to have a direct effect on the outcome levels, provided it has no direct effect on the temporal trend of the outcome and does not modify the average treatment effect.

\begin{assumption}[No unmeasured common effect modifier]\label{ass:ass6}
For all $0\leq t<\infty$ and any $l=0,1$, $Cov(N_l^{*(1)}(t)-N_l^{*(0)}(t),A_l^{(1)}-A_l^{(0)}|\boldsymbol W)=0$ and $Cov(Y_l^{*(1)}(t)-Y_l^{*(0)}(t),A_l^{(1)}-A_l^{(0)}|\boldsymbol W)=0$. 
\end{assumption}
\noindent
Assumption \ref{ass:ass6} essentially states that there is no common effect modifier by an unmeasured confounder of the additive effect of treatment on the outcome and the additive effect of the IV on treatment. It allows us to identify the population average treatment effect (ATE) rather than a local average treatment effect (LATE) by ensuring that those who comply with the instrument do not have systematically different treatment effects than those who do not.
\begin{assumption}[Stable Treatment Effect over each period]\label{ass:ass7}
   For all $0\leq t<\infty$, $\mathbb{E}[N_1^{*(1)}(t)-N_1^{*(0)}(t)|\boldsymbol W]=\mathbb{E}[N_0^{*(1)}(t)-N_0^{*(0)}(t)|\boldsymbol W]$ and $\mathbb{E}[Y_1^{*(1)}(t)-Y_1^{*(0)}(t)|\boldsymbol W]=\mathbb{E}[Y_0^{*(1)}(t)-Y_0^{*(0)}(t)|\boldsymbol W]$.
\end{assumption}
\noindent
Assumption \ref{ass:ass7} requires that the Conditional Average Treatment Effects (CATEs) do not vary over periods. By ensuring the effect of the treatment on both the recurrent events and survival remains constant across periods, it guarantees that the optimal policy identified is stable across both the periods.
\begin{assumption}[Non-informative Censoring]\label{ass:ass8}
    For all $0\leq t<\infty$, $C\independent \{N^*(t),Y^*(t)\}|A,L,Z,\boldsymbol W$ almost surely.
\end{assumption}
\noindent
Assumption \ref{ass:ass8} implies that the censoring process is non-informative conditionally on the treatment, period, instrument, and baseline covariates.

\subsection{Inverse Probability Weighted Identification}
We present the identification of the optimal treatment policy through two distinct Inverse Probability Weighted (IPW) estimators. First, we define several key nuisance functions and quantities. Let $\pi(l,z,\boldsymbol w) = P(L=l, Z=z \mid \boldsymbol W=\boldsymbol w)$ represent the joint probability of period and instrument assignment given covariates, and $\mu_A(l,z,\boldsymbol w) = E(A \mid L=l, Z=z, \boldsymbol W = \boldsymbol w)$ denote the conditional mean of the treatment. The treatment trend denominator is given by $\delta_A(\boldsymbol w) = \mu_A(1,1,\boldsymbol w) - \mu_A(1,0,\boldsymbol w) - \mu_A(0,1,\boldsymbol w) + \mu_A(0,0,\boldsymbol w)$. Let $N_C(t)=\mathbb{I}(X\leq t,\Delta=0)$ and $N_D(t)=\mathbb{I}(X\leq t,\Delta=1)$.
To address right-censoring, we define the cumulative hazard of the censoring process as 
$$\Lambda_C(t; a, l,z,\boldsymbol w) = \int_{(0,t]} \frac{dE\{N_C(u) \mid A=a, L=l, Z=z,\boldsymbol W=\boldsymbol w\}}{E\{Y^\dagger(u) \mid A=a, L=l, Z=z,\boldsymbol W=\boldsymbol w\}},$$ 
where $Y^\dagger(u) = I(X > u, \Delta = 1 \text{ or } X \ge u, \Delta = 0)$ is a modified at-risk process \citep{baer2023causal}. This process specifically accounts for the terminating nature of events such that $Y^\dagger(t) = Y(t) - \{N_D(t) - N_D(t-)\}$. Finally, the survival function for the censoring process is defined via the product integral as $K(t; a, l,z,\boldsymbol w) = \prod_{u \in (0, t]} \{1 - d\Lambda_C(u; a, l,z,\boldsymbol w)\}$
where $\prod$ denotes the product integral defined in \citet{gill1990survey}.\\

\noindent
Since $N_l^{*(d^t(\boldsymbol W))}(t)=N_l^{*(1)}(t)d^t(\boldsymbol W)+N_l^{*(0)}(t)(1-d^t(\boldsymbol W))$ and $Y_l^{*(d^t(\boldsymbol W))}(t)=Y_l^{*(1)}(t)d^t(\boldsymbol W)+Y_l^{*(0)}(t)(1-d^t(\boldsymbol W))$, we define the CATEs as $\tau^N(t,\boldsymbol W)=E[N_l^{*(1)}(t)-N_l^{*(0)}(t)|\boldsymbol W]$ and $\tau^Y(t,\boldsymbol W)=E[Y_l^{*(1)}(t)-Y_l^{*(0)}(t)|\boldsymbol W]$. Due to Assumption \ref{ass:ass7}, which posits stable treatment effects over each period, we can omit the period index $l$ from these CATEs. Consequently, the optimization problem in \eqref{eq:eq2} simplifies to finding a policy $d^t(\boldsymbol W)$ that satisfies:
\begin{align}
d^t_{\text{opt}}=\arg\min_{d \in \mathcal{D}} \mathbb{E} \left[\tau^N(t,\boldsymbol W)d^t(\boldsymbol W) \right]\quad \text{subject to } \mathbb{E}[\tau^Y(t,\boldsymbol W)\{d^t(\boldsymbol W)-\widetilde{d}(\boldsymbol W) \}]> 0 \cdot\label{eq:eq3}
\end{align}
\begin{theorem}\label{thm:thm1}
    Under Assumptions \ref{ass:ass1}--\ref{ass:ass8}, the following results hold
    \begin{align*}
         &\mathbb{E}\left[\frac{\Delta\cdot(2Z-1)(2L-1)(2A-1)N(t)\{\mathbb{I}(A=d^t(\boldsymbol W))\}}{K(X-,A,L,Z,\boldsymbol W)\pi(L,Z,\boldsymbol W)\delta_A(\boldsymbol W)}\right]=\mathbb{E} \left[\tau^N(t,\boldsymbol W)d^t(\boldsymbol W) \right]+f_N\\
         &\mathbb{E}\left[\frac{\Delta\cdot(2Z-1)(2L-1)(2A-1)Y(t)\{\mathbb{I}(A=d^t(\boldsymbol W))\}}{K(X-,A,L,Z,\boldsymbol W)\pi(L,Z,\boldsymbol W)\delta_A(\boldsymbol W)}\right]=\mathbb{E} \left[\tau^Y(t,\boldsymbol W)d^t(\boldsymbol W) \right]+f_Y
    \end{align*}
    where $f_N,f_Y$ do not depend on $d^t(\boldsymbol W)$. Moreover, the optimization policy in \eqref{eq:eq3} is identified by 
     \begin{align*}
         d^t_{\text{opt}}&=\arg\min_{d^t\in \mathcal{D}} \mathbb{E}\left[\frac{\Delta\cdot(2Z-1)(2L-1)(2A-1)N(t)\{\mathbb{I}(A=d^t(\boldsymbol W))\}}{K(X-,A,L,Z,\boldsymbol W)\pi(L,Z,\boldsymbol W)\delta_A(\boldsymbol W)}\right]\hspace{0.2cm} \text{subject to } \\
         &\mathbb{E}\left[\frac{\Delta\cdot(2Z-1)(2L-1)(2A-1))Y(t)\{\mathbb{I}(A=d^t(\boldsymbol W))-\mathbb{I}(A=\widetilde{d}(\boldsymbol W))\}}{K(X-,A,L,Z,\boldsymbol W)\pi(L,Z,\boldsymbol W)\delta_A(\boldsymbol W)}\right]>0
    \end{align*}
\end{theorem}
\noindent
The proof of Theorem \ref{thm:thm1} is provided in the Supplementary Section S1. The IPW identification formula simultaneously adjusts for right-censoring using the survival function of the censoring process $K(\cdot)$ effectively reweighting the observed events to account for those lost to follow-up and for unmeasured confounding using the iDID framework. A limitation of the first identification result is its reliance on the uncensored sub-sample $(\Delta=1)$, which effectively discards potentially valuable information from units who are censored before the terminal event at any point during the followup time. In particular, it excludes individuals who remain at risk at time $t$ but are censored at some later time $t' > t$. To overcome this loss of information and improve efficiency, we propose a second IPW-type estimator. This approach is inspired by the framework of \citet{schaubel2010estimating}, but we extend it in two significant directions: first, by adapting it to the iDID framework to account for unmeasured confounding; and second, by allowing the censoring survival process $K(\cdot)$ to depend on the period ($L$), the instrument ($Z$), and baseline covariates ($\boldsymbol W$), in addition to the treatment ($A$). Let $dN^*(t)=N^*(t)-N^*(t-)$, $dN(t)=N(t)-N(t-)$, $dY^*(t)=Y^*(t)-Y^*(t-)$ and $dY(t)=Y(t)-Y(t-)$.
\begin{theorem}\label{thm:thm2}
    Under Assumptions \ref{ass:ass1}--\ref{ass:ass8}, the optimization policy in \eqref{eq:eq3} is identified by
    \begin{align*}
        d^t_{\text{opt}}&=\arg\min_{d^t\in \mathcal{D}}\mathbb{E}\left[\int_{0}^t\frac{(2Z-1)(2L-1)(2A-1)dN(s)\{\mathbb{I}(A=d^t(\boldsymbol W))\}}{K(s,A,L,Z,\boldsymbol W)\pi(L,Z,\boldsymbol W)\delta_A(\boldsymbol W)}\right]\hspace{0.2cm} \text{subject to } \\
         & \mathbb{E}\left[\int_{0}^t\frac{(2Z-1)(2L-1)(2A-1)dY(s)\{\mathbb{I}(A=d^t(\boldsymbol W))-\mathbb{I}(A=\widetilde{d}(\boldsymbol W))\}}{K(s,A,L,Z,\boldsymbol W)\pi(L,Z,\boldsymbol W)\delta_A(\boldsymbol W)}\right]>0\cdot
    \end{align*}
\end{theorem}
\noindent
The proof of Theorem \ref{thm:thm2} is provided in the Supplementary Section S2. Unlike Theorem \ref{thm:thm1}, this formulation uses the cumulative event information over the interval $[0,t]$. In contrast to the IPW mapping in Theorem \ref{thm:thm1}, which relies on the uncensored subsample, the present approach allows individuals who are at risk at time $t$ to contribute to the estimation, even if they are censored at a later time $t' > t$.

\subsection{Multiply Robust Identification}
Both IPW estimators require the correct specification of all nuisance parameter models to ensure consistency. Specifically, these include the joint distribution $\pi(\cdot)$; the treatment trend denominator $\delta_A(\boldsymbol{W})$, which captures the instrument's effect on treatment across time; and the censoring survival function $K(\cdot)$, used to account for administrative censoring. Consequently, methods robust against model mis-specification are highly desired, where consistency is guaranteed if only a subset of the posited models is correctly specified. In this section, we propose a Multiply Robust Identification strategy. In our context, this is particularly complex as the estimator must simultaneously account for recurrent events, the competing risk of death, and censoring, all while operating within the iDID framework to address unmeasured confounding.\\

\noindent
Following the framework of \citet{cui2021semiparametric} and \citet{zhao2025semiparametric}, we first derive a multiply robust identification for $\mathbb{E}[\tau^N(\boldsymbol W)]$ and $\mathbb{E}[\tau^Y(\boldsymbol W)]$. We identify these quantities through a modified Wald-type identification formula.  Let $\widetilde{N}(t)=\frac{\Delta}{K(X-,A,L,Z,\boldsymbol W)}N(t)$ and $\widetilde{Y}(t)=\frac{\Delta}{K(X-,A,L,Z,\boldsymbol W)}Y(t)$. For $Q\in\{\widetilde{N}(t),\widetilde{Y}(t)\}$, we define $\mu_Q(l,z,\boldsymbol w)=\mathbb{E}[Q|L=l,Z=z,\boldsymbol W=\boldsymbol w]$ and $\delta_{Q}(\boldsymbol w)= \mu_{Q}(1,1,\boldsymbol w)-\mu_{Q}(1,0,\boldsymbol w)-\mu_{Q}(0,1,\boldsymbol w)+\mu_{Q}(0,0,\boldsymbol w)$.
\begin{theorem}\label{thm:thm3}
    Under Assumptions \ref{ass:ass1}--\ref{ass:ass8}, the following results hold:
    \begin{align*}
        \beta(t):=\mathbb{E}[\tau^N(t,\boldsymbol W)]=\mathbb{E}\left[\frac{\delta_{\widetilde{N}(t)}(\boldsymbol W)}{\delta_{A}(\boldsymbol W)}\right],\quad \eta(t):=\mathbb{E}[\tau^Y(t,\boldsymbol W)]=\mathbb{E}\left[\frac{\delta_{\widetilde{Y}(t)}(\boldsymbol W)}{\delta_{A}(\boldsymbol W)}\right]\cdot
    \end{align*}
\end{theorem}
\noindent
The proof of Theorem \ref{thm:thm3} is provided in the Supplementary Section S4. 
To establish the multiply robust framework, we utilize the theory of von Mises expansions. Let $\mathbb{P}$ and $\bar{\mathbb{P}}$ denote two observed data distributions, with $\mathbb{E}$ and $\bar{\mathbb{E}}$ representing their respective expectations. For a given parameter $\psi$, the von Mises expansion of $\psi$ at $\bar{\mathbb{P}}$ centered at $\mathbb{P}$ is given by:
$$\psi(\bar{\mathbb{P}}) - \psi(\mathbb{P}) = (\bar{\mathbb{E}} - \mathbb{E}) D(O; \bar{\mathbb{P}}) + R(\bar{\mathbb{P}}, \mathbb{P})\cdot$$
In this expansion, $D(O; \bar{\mathbb{P}})$ captures the first-order behavior of the functional $\psi$, while $R(\bar{\mathbb{P}}, \mathbb{P})$ represents the second-order remainder term. For the statement of the next theorem, we introduce the following notation. 
\begin{align*}
    & F^*(u,t,a,l,z,\boldsymbol w)=E(\mathbb{I}(D>u)N^*(t)|A=a,L=l,Z=z,\boldsymbol W=\boldsymbol w),\\
    & F(u,t,a,l,z,\boldsymbol w)=\mathbb{E}[\mathbb{I}(X>u)\widetilde{N}(t)|A=a,L=l,Z=z,\boldsymbol W=\boldsymbol w],\\
    & \Lambda_D(t; a, l,z,\boldsymbol w) = \int_{(0,t]} \frac{dE\{N_D(u) \mid A=a, L=l, Z=z,\boldsymbol W=\boldsymbol w\}}{E\{Y(u) \mid A=a, L=l, Z=z,\boldsymbol W=\boldsymbol w\}},\\
    & H(t,a, l,z,\boldsymbol w)=\prod_{u \in (0, t]} \{1 - d\Lambda_D(u; a, l,z,\boldsymbol w)\},\\
    &M_C(t,a,l,z,\boldsymbol w)=N_C(t)-\int_{(0,t]}Y^{\dagger}(u)d\Lambda_C(u,a,l,z,\boldsymbol w)\cdot
\end{align*}
\begin{theorem}
     Consider two observed data distributions specified by $\mathbb{P}$ and $\bar{\mathbb{P}}$. Under Assumptions \ref{ass:ass1}--\ref{ass:ass8}, for all $0\leq t<\infty$,\\
     (i) The estimand $\beta(t)$ admits a von Mises expansion with $D_{\beta}(t,O,\mathbb{P})$ given by,
     \begin{align*}
     &\frac{\delta_{\widetilde{N}(t)}(\boldsymbol W)}{\delta_{A}(\boldsymbol W)}(\mathbb{P})-\beta(t)+\frac{(2Z-1)(2L-1)}{\pi(L,Z,\boldsymbol W)\delta_A(\boldsymbol W)}\left[\widetilde{N}(t)-\mu_{\widetilde{N}(t)}(L,Z,\boldsymbol W)-\right.\\
    &\left.\frac{\delta_{\widetilde{N}(t)}(\boldsymbol W)}{\delta_{A}(\boldsymbol W)}\{A-\mu_{A}(L,Z,\boldsymbol W)\}+ \int_{0}^\infty \frac{F(u,t,A,L,Z,\boldsymbol W)}{H(u,A,L,Z,\boldsymbol W)}\frac{dM_C(u,A,L,Z,\boldsymbol W)}{K(u,A,L,Z,\boldsymbol W)}\right]\cdot
\end{align*}
(ii) The estimand $\eta(t)$ admits a von Mises expansion with $D_{\eta}(t,O,\mathbb{P})$ given by,
    \begin{align*}
     &\frac{\delta_{\widetilde{Y}(t)}(\boldsymbol W)}{\delta_{A}(\boldsymbol W)}(\mathbb{P})-\eta(t)+\frac{(2Z-1)(2L-1)}{\pi(L,Z,\boldsymbol W)\delta_A(\boldsymbol W)}\left[\widetilde{Y}(t)-\mu_{\widetilde{Y}(t)}(L,Z,\boldsymbol W)-\right.\\
    &\left.\frac{\delta_{\widetilde{Y}(t)}(\boldsymbol W)}{\delta_{A}(\boldsymbol W)}\{A-\mu_{A}(L,Z,\boldsymbol W)\}+ \int_{0}^\infty \frac{H(u\vee  t,A,L,Z,\boldsymbol W)}{H(u,A,L,Z,\boldsymbol W)}\frac{dM_C(u,A,L,Z,\boldsymbol W)}{K(u,A,L,Z,\boldsymbol W)}\right]\cdot
\end{align*}
The explicit forms of the remainder terms $R_{\beta}(t,O,\bar{\mathbb{P}}, \mathbb{P})$ and $R_{\eta}(t,O,\bar{\mathbb{P}}, \mathbb{P})$ are shown in Supplementary Section S4 where they are shown to be of second order along with the proof.
\end{theorem}
\noindent
The existence of this expansion indicates that the observed data estimand is sufficiently smooth to be pathwise differentiable, guaranteeing the existence of at least one asymptotically linear estimator. 
To construct our multiply robust framework, we proceed with $D_{\beta}(t,O,\mathbb{P})$ and $D_{\eta}(t,O,\mathbb{P})$ as the corresponding influence functions for $\beta(t)$ and $\eta(t)$, respectively, accompanied by second-order remainder terms.\\

\noindent
Let $\mathcal{P}$ denote the class of observed data distributions. We define the following six sub-models $\mathcal{M}_1, \dots, \mathcal{M}_6 \subset \mathcal{P}$. In Theorem \ref{thm:thm5}, we show the multiple robustness of the constrained optimization problem in the union model, $\mathcal{M}_{\text{union}}=\bigcup_{j=1}^6 \mathcal{M}_j$, where:\\

\noindent
$\mathcal{M}_1$: $\pi, \mu_A, F, H$ are correctly specified; $\mathcal{M}_2$: $\pi, \mu_A, K$ are correctly specified;\\
\noindent
$\mathcal{M}_3$: $\pi, \delta_{\widetilde{N}}/\delta_A, \delta_{\widetilde{Y}}/\delta_A, F, H$ are correctly specified; $\mathcal{M}_4$: $\pi, \delta_{\widetilde{N}}/\delta_A, \delta_{\widetilde{Y}}/\delta_A, K$ are correctly specified;\\
\noindent
$\mathcal{M}_5$: $\mu_{\widetilde{N}}, \mu_{\widetilde{Y}}, \mu_A, F, H$ are correctly specified; 
$\mathcal{M}_6$: $\mu_{\widetilde{N}}, \mu_{\widetilde{Y}}, \mu_A, K$ are correctly specified.\\

\noindent
 Let $\mathbb{P}_0$ denote the true data generating mechanism. Let $\Delta^t_N(O):=D_{\beta}(t,O,\mathbb{P}_0)+\beta(t)$, $\Delta^t_Y(O):=D_{\eta}(t,O,\mathbb{P}_0)+\eta(t)$, and define 
\begin{align*}
    W^t_{N}&:=(2A-1)\left[\frac{\delta_{\widetilde{N}(t)}(\boldsymbol W)}{\delta_{A}(\boldsymbol W)}+\frac{(2Z-1)(2L-1)}{\pi(L,Z,\boldsymbol W)\delta_A(\boldsymbol W)}\left[\widetilde{N}(t)-\mu_{\widetilde{N}(t)}(L,Z,\boldsymbol W)-\right.\right.\\
    &\left.\left.\frac{\delta_{\widetilde{N}(t)}(\boldsymbol W)}{\delta_{A}(\boldsymbol W)}\{A-\mu_{A}(L,Z,\boldsymbol W)\}+ \int_{0}^\infty \frac{F(u,t,A,L,Z,\boldsymbol W)}{H(u,A,L,Z,\boldsymbol W)}\frac{dM_C(u,A,L,Z,\boldsymbol W)}{K(u,A,L,Z,\boldsymbol W)}\right]\right],\\
    W^t_{Y}&:=(2A-1)\left[\frac{\delta_{\widetilde{Y}(t)}(\boldsymbol W)}{\delta_{A}(\boldsymbol W)}+\frac{(2Z-1)(2L-1)}{\pi(L,Z,\boldsymbol W)\delta_A(\boldsymbol W)}\left[\widetilde{Y}(t)-\mu_{\widetilde{Y}(t)}(L,Z,\boldsymbol W)-\right.\right.\\
    &\left.\left.\frac{\delta_{\widetilde{Y}(t)}(\boldsymbol W)}{\delta_{A}(\boldsymbol W)}\{A-\mu_{A}(L,Z,\boldsymbol W)\}+ \int_{0}^\infty \frac{H(u\vee  t,A,L,Z,\boldsymbol W)}{H(u,A,L,Z,\boldsymbol W)}\frac{dM_C(u,A,L,Z,\boldsymbol W)}{K(u,A,L,Z,\boldsymbol W)}\right]\right]\cdot
\end{align*}
\begin{theorem}\label{thm:thm5}
Under the union of models $\mathcal{M}=\bigcup_{j=1}^6 \mathcal{M}_j$, the optimal policy in \eqref{eq:eq3} is identified by the following constrained optimization problem:
\vspace{-0.5cm}
\begin{align*}
\arg\min_{d^t\in \mathcal{D}}\mathbb{E}[\Delta^t_N(O)d^t(\boldsymbol W)]\hspace{0.2cm} \text{subject to } \mathbb{E}[\Delta^t_Y(O)\{d^t(\boldsymbol W)-\widetilde{d}(\boldsymbol W)\}]>0
\end{align*}
which is equivalent to:
\vspace{-0.5cm}
\begin{align*}
\arg\min_{d^t\in \mathcal{D}}\mathbb{E}[W^t_{N}\mathbb{I}\{A=d^t(\boldsymbol W)\}]\hspace{0.2cm} \text{subject to } \mathbb{E}[W^t_{Y}(\mathbb{I}\{A=d^t(\boldsymbol W)\}-\mathbb{I}\{A=\widetilde{d}(\boldsymbol W)\})]>0
\end{align*}
\end{theorem}
\noindent
The proof of this theorem is given in the Supplementary Section S5. A summary of the two IPW methods and the proposed multiply robust method is provided in Table \ref{tab:method_Comparison}.\\

\noindent
\begin{remark}
    Due to the iDID setup, the absolute value function $V^t_N(d) = \mathbb{E}[N^{*(d^t(\boldsymbol W))}(t)]$ is not point-identified. However, the contrast relative to the behavioral policy is identified. We therefore evaluate the policy gain, defined as the difference between the value function under the optimal policy and the behavioral policy, mathematically let $\Gamma^t(d^t_{\boldsymbol{\theta}}, \widetilde{d}) = V^t_N(d^t_{\boldsymbol{\theta}}) - V^t_N(\widetilde{d})$, which is identified by $\mathbb{E}[\Delta^t_N(O)\{d^t_{\boldsymbol{\theta}}(\boldsymbol{W}) - \widetilde{d}(\boldsymbol{W})\}]$.
    This quantity can be consistently estimated using the two IPW methods and the proposed multiply robust method.
\end{remark}

\newcolumntype{P}[1]{>{\RaggedRight\arraybackslash}p{#1}}
\begin{table}[ht]
\centering
\caption{Comparison of IPW and AIPW Estimators}
\label{tab:method_Comparison}
\renewcommand{\arraystretch}{1.2}
\begin{tabular}{P{1.5cm}P{3.5cm}P{4.5cm}P{3.5cm}}
\toprule
 & \multicolumn{1}{c}{\textbf{Features}} & \multicolumn{1}{c}{\textbf{Pros}} & \multicolumn{1}{c}{\textbf{Cons}} \\
\midrule
\multirow{4}{*}{\textbf{IPW}}
  & \vspace{-0.5\baselineskip}\begin{itemize}[leftmargin=*, nosep, topsep=0pt, itemsep=2pt]
      \item Instrumented DID
      \item Constrained Optimization
    \end{itemize}
  & \vspace{-0.5\baselineskip}\begin{itemize}[leftmargin=*, nosep, topsep=0pt, itemsep=2pt]
      \item Accounts for Unmeasured Confounding
      \item Faster Computation
    \end{itemize}
  & \vspace{-0.5\baselineskip}\begin{itemize}[leftmargin=*, nosep, topsep=0pt, itemsep=2pt]
      \item Parametric Nuisance Models
      \item Only non-censored units up to $t$
    \end{itemize} \\
\midrule
\multirow{4}{*}{\textbf{AIPW}}
  & \vspace{-0.5\baselineskip}\begin{itemize}[leftmargin=*, nosep, topsep=0pt, itemsep=2pt]
      \item Instrumented DID
      \item Constrained Optimization
    \end{itemize}
  & \vspace{-0.5\baselineskip}\begin{itemize}[leftmargin=*, nosep, topsep=0pt, itemsep=2pt]
      \item Accounts for Unmeasured Confounding
      \item Multiply Robust
      \item Flexible Nuisance Models
    \end{itemize}
  & \vspace{-0.5\baselineskip}\begin{itemize}[leftmargin=*, nosep, topsep=0pt, itemsep=2pt]
      \item Longer Computation
      \item Only non-censored units up to $t$
    \end{itemize} \\
\bottomrule
\end{tabular}
\begin{tablenotes}
\small
\item IPW2 includes all units other than those censored at time zero.
\end{tablenotes}
\end{table}
\section{Estimation and Asymptotics}
\subsection{Linear Class of Decision policies}
We focus on a class of policies parametrized using a finite dimensional vector of parameters. Specifically, we consider the class: $\mathcal{D}_{\boldsymbol \theta} = \{d^t_{\boldsymbol \theta}(\boldsymbol W) = \mathbb{I}(\boldsymbol \theta' \widetilde{\boldsymbol W} > 0) : \boldsymbol \theta \in \mathbb{R}^{p+1}\},$
where $\widetilde{\boldsymbol W} = (1, \boldsymbol W')'$.

\begin{remark}
    We focus on a parametric linear class for both interpretability and practicality. Linear policies yield transparent policies  and enable analysis of the convergence rate of the estimated policy (Theorem \ref{thm:thm6}). However, the identification results developed in  Theorems \ref{thm:thm1}, \ref{thm:thm2}, and \ref{thm:thm5} are not restricted to this class and extend to more general, potentially nonparametric policies. In particular, the framework can accommodate richer function classes such as reproducing kernel Hilbert spaces, as in outcome-weighted learning approaches  \citep{zhao2012estimating} 
\end{remark}
\subsection{Optimization through Lagrange Multiplier}
From here on, for the purpose of estimation and asymptotics, we demonstrate the procedure through the more general multiply robust method. For the IPW methods, the derivations can be performed similarly. To operationalize the constrained optimization problem in Theorem \ref{thm:thm5}, we reformulate it as a Lagrangian multiplier problem. Let $\Phi^t_{1}(\boldsymbol \theta)=-\mathbb{E}[\Delta^t_N(O)d^t_{\boldsymbol \theta}(\boldsymbol W)]$, $\Psi^t_{1}(\boldsymbol \theta)=\mathbb{E}[\Delta^t_Y(O)\{d^t_{\boldsymbol \theta}(\boldsymbol W)-\widetilde{d}(\boldsymbol W)\}]$, $\Phi^t_2(\boldsymbol \theta)=-\mathbb{E}[W^t_{N}\cdot\mathbb{I}\{A=d^t_{\boldsymbol \theta}(\boldsymbol W)\}]$ and $\Psi^t_2(\boldsymbol \theta)=\mathbb{E}[W^t_{Y}(\mathbb{I}\{A=d^t_{\boldsymbol \theta}(\boldsymbol W)\}-\mathbb{I}\{A=\widetilde{d}(\boldsymbol W)\})]$. Moreover let $\xi^t_1(\boldsymbol \theta)=\mathbb{I}(\Psi^t_1(\boldsymbol \theta) < 0)$ and $\xi^t_2(\boldsymbol \theta)=\mathbb{I}(\Psi^t_2(\boldsymbol \theta) < 0)$.
The optimization task is expressed as:
\begin{align*}
\boldsymbol \theta^{*t} &= \arg\max_{\boldsymbol \theta\in \Theta} \left[ \Phi^t_1(\boldsymbol \theta) - \lambda \cdot \xi^t_1(\boldsymbol \theta) \right] = \arg\max_{\boldsymbol \theta\in \Theta} \left[ \Phi^t_2(\boldsymbol \theta) - \lambda \cdot \xi^t_2(\boldsymbol \theta)\right],
\end{align*}
where $\lambda$ is a sufficiently large pre-specified positive constant that serves as a penalty parameter to ensure the safety constraint is satisfied. 
\subsection{Estimation}
For the purpose of estimation, we replace the population quantities $\Phi^t_{1}(\boldsymbol{\theta}), \Phi^t_{2}(\boldsymbol{\theta}), \Psi^t_{1}(\boldsymbol{\theta}),$ and $\Psi^t_{2}(\boldsymbol{\theta})$ with their respective sample estimators. The multiply robust estimators for the objective $\widehat{\Phi}^t_{1}(\boldsymbol{\theta})$ and the constraint $\widehat{\Psi}^t_{1}(\boldsymbol{\theta})$ are given by:
\begin{align*}
\widehat{\Phi}^t_{1}(\boldsymbol \theta) &= -\frac{1}{n}\sum_{i=1}^n\left[\frac{\widehat{\delta}_{\widetilde{N}(t)}(\boldsymbol W_i)}{\widehat{\delta}_{A}(\boldsymbol W_i)}+\frac{(2Z_i-1)(2L_i-1)}{\widehat{\pi}(L_i,Z_i,\boldsymbol W_i)\widehat{\delta}_A(\boldsymbol W_i)}\left\{\frac{\Delta_i\cdot N_i(t)}{\widehat{K}(X_i-,A_i,L_i,Z_i,\boldsymbol W_i)}\right.\right.\\
&\quad\left.\left.-\widehat{\mu}_{\widetilde{N}(t)}(L_i,Z_i,\boldsymbol W_i)-\frac{\widehat{\delta}_{\widetilde{N}(t)}(\boldsymbol W_i)}{\widehat{\delta}_{A}(\boldsymbol W_i)}\{A-\widehat{\mu}_{A}(L_i,Z_i,\boldsymbol W_i)\}\right.\right.\\
&\quad\left.\left.+ \int_{0}^\infty \frac{\widehat{F}(u,t,A_i,L_i,Z_i,\boldsymbol W_i)}{\widehat{H}(u,A_i,L_i,Z_i,\boldsymbol W_i)}\frac{d\widehat{M}_C(u,A_i,L_i,Z_i,\boldsymbol W_i)}{\widehat{K}(u,A_i,L_i,Z_i,\boldsymbol W_i)}\right\}\right]d^t_{\boldsymbol \theta}(\boldsymbol W),
\end{align*}
\begin{align*}
\widehat{\Psi}^t_{1}(\boldsymbol \theta) &= \frac{1}{n}\sum_{i=1}^n\left[\frac{\widehat{\delta}_{\widetilde{Y}(t)}(\boldsymbol W_i)}{\widehat{\delta}_{A}(\boldsymbol W_i)}+\frac{(2Z_i-1)(2L_i-1)}{\widehat{\pi}(L_i,Z_i,\boldsymbol W_i)\widehat{\delta}_A(\boldsymbol W_i)}\left\{\frac{\Delta_i\cdot Y_i(t)}{\widehat{K}(X_i-,A_i,L_i,Z_i,\boldsymbol W_i)}\right.\right.\\
&\quad\left.\left.-\widehat{\mu}_{\widetilde{Y}(t)}(L_i,Z_i,\boldsymbol W_i)-\frac{\widehat{\delta}_{\widetilde{Y}(t)}(\boldsymbol W_i)}{\widehat{\delta}_{A}(\boldsymbol W_i)}\{A-\widehat{\mu}_{A}(L,Z,\boldsymbol W)\}\right.\right.\\
&\quad\left.\left.+ \int_{0}^\infty \frac{\widehat{H}(u \vee t,A_i,L_i,Z_i,\boldsymbol W_i)}{\widehat{H}(u,A_i,L_i,Z_i,\boldsymbol W_i)}\frac{d\widehat{M}_C(u,A_i,L_i,Z_i,\boldsymbol W_i)}{\widehat{K}(u,A_i,L_i,Z_i,\boldsymbol W_i)}\right\}\right][d^t_{\boldsymbol \theta}(\boldsymbol W)-\widetilde{d}(\boldsymbol W_i)].
\end{align*}
Hence $\widehat{\xi}^t_1(\boldsymbol \theta)=\mathbb{I}(\widehat{\Psi}^t_1(\boldsymbol \theta) < 0)$.
Similarly, we can define the sample versions $\widehat{\Phi}^t_{2}(\boldsymbol{\theta}), \widehat{\Psi}^t_{2}(\boldsymbol{\theta})$ and $\widehat{\xi}^t_2(\boldsymbol \theta)$. The estimated optimal treatment policy parameter $\widehat{\boldsymbol{\theta}^t}$ is then given by:
\begin{align*}
\widehat{\boldsymbol{\theta}}^t &= \arg\max_{\boldsymbol{\theta}\in \Theta} \left[ \widehat{\Phi}^t_1(\boldsymbol{\theta}) - \lambda \cdot \widehat{\xi}^t_1(\boldsymbol \theta)\right] = \arg\max_{\boldsymbol{\theta}\in \Theta} \left[\widehat{\Phi}^t_2(\boldsymbol{\theta}) - \lambda \cdot\widehat{\xi}^t_2(\boldsymbol \theta)\right]
\end{align*}
For the two IPW estimators, we need to estimate the nuisance functions $\pi(l,z,\boldsymbol W)$, $\mu_a(l,z,\boldsymbol W)$,
$K(s,a,l,z,\boldsymbol W)$ and $\widetilde{\pi}(\boldsymbol W)$. To ensure the proper convergence and asymptotic normality of the IPW-based policy parameters, these nuisance functions must be estimated at a sufficiently fast rate, typically $O_p(n^{-1/2})$. This requirement restricts choices to parametric or certain semiparametric approaches, such as generalized linear models and Cox proportional hazards models.\\

\noindent
For the multiply robust estimator, we employ flexible machine learning methods to estimate the nuisance parameter models specified by $\pi(l, z, \boldsymbol W)$, $\mu_a(l, z, \boldsymbol W)$, $K(s, a, l, z, \boldsymbol W)$, $\mu_{\widetilde{N}(t)}(l, z, \boldsymbol W)$, $\mu_{\widetilde{Y}(t)}(L, Z, \boldsymbol W)$, $F(u, t, A, L, Z, \boldsymbol W)$, $H(u, A, L, Z, \boldsymbol W)$ and $\widetilde{\pi}(\boldsymbol W)$. To maintain the $\sqrt{n}$-consistency of the policy parameters and protect against overfitting, we utilize the technique of cross-fitting \citep{klaassen1987consistent,zheng2011cross}. To optimize the policy, we employ the genetic algorithm implemented in the R package \texttt{rgenoud} to identify the global optimum $\widehat{\boldsymbol{\theta}}$. This choice is necessitated by the fact that our objective functions are both non-convex and non-smooth, rendering traditional derivative-based optimization algorithms unsuitable. The genetic algorithm provides a robust search mechanism across the parameter space, ensuring the identification of the optimal treatment policy even in the presence of complex, discontinuous policy surfaces.
\subsection{Asymptotic Properties}
In this section, we study the asymptotic properties of the multiply robust estimator. We define the objective function for optimization as $G^t(\boldsymbol{\theta}) = \Phi^t_1(\boldsymbol{\theta}) - \lambda \cdot \xi^t_1(\boldsymbol{\theta})$.

\begin{assumption}\label{ass:ass9}
For all $0\leq t<\infty$, the following holds:
\begin{itemize}
\item[(i)] The supports of $N(t)$ and $\boldsymbol{W}$ are bounded.
\item[(ii)] The functions $\Phi^t_1(\boldsymbol{\theta})$ and $\Psi^t_1(\boldsymbol{\theta})$ are twice differentiable in neighborhoods $\mathcal{N}_1$ and $\mathcal{N}_2$ of $\boldsymbol{\theta}^{*t}$, respectively, such that $\mathcal{N} = \mathcal{N}_1 \cap \mathcal{N}_2 \neq \emptyset$.
\item[(iii)] For $\boldsymbol{\theta} \in \Theta$, $|\Psi^t_1(\boldsymbol{\theta})| > 0$ almost surely.
\item[(iv)] Margin Condition: There exists a constant $\delta_0 > 0$ such that $P(0 < |\boldsymbol{\theta}^{'} \widetilde{\boldsymbol{W}}| < \delta) = O(\delta)$, where the $O(\delta)$ term is uniform in $0 < \delta < \delta_0$.
\end{itemize}
\end{assumption}
\noindent
Assumptions \ref{ass:ass9}(i) and (iii) are standard regularity conditions used to establish uniform convergence of the empirical processes. Assumption \ref{ass:ass9}(iii) ensures that the safety constraint $\mathbb{I}(\Psi^t_1(\boldsymbol{\theta}) < 0)$ is identifiable and converges as $n \to \infty$, preventing the constraint boundary from becoming degenerate. Margin conditions such as Assumption \ref{ass:ass9}(iv) are common in policy learning literature \citep{tsybakov2004optimal, luedtke2016statistical} to control the behavior of the decision boundary and ensure the objective function is well-behaved near the optimal parameter $\boldsymbol{\theta}^{*t}$. 

\begin{assumption}\label{ass:ass10}
 We assume the following rate conditions for nuisance parameter estimation, that for any $l,z=0,1$ and for a $T$ large enough,
 \begin{align*}
   &||\widehat{\mu}_A(l,z,\boldsymbol W)-\mu_A(l,z,\boldsymbol W)||_{L_2}=o_p(n^{-1/4}), \hspace{0.1cm} ||\widehat{\mu}_{\widetilde{N}}(l,z,\boldsymbol W)-\mu_{\widetilde{N}}(l,z,\boldsymbol W)||_{L_2}=o_p(n^{-1/4}),\\
     & ||\widehat{\pi}(l,z,\boldsymbol W)-\pi(l,z,\boldsymbol W)||_{L_2}=o_p(n^{-1/4}),\\
     &\text{sup}_{u\leq L}||\widehat{K}(u,A,L,Z,\boldsymbol W)-K(u,A,L,Z,\boldsymbol W)||_{L_2}=o_p(n^{-1/4}),\\
     & \sup_{u\leq T}\left|\left|d\widehat{\Lambda}_C(u,A,L,Z,\boldsymbol W)d\Lambda_C(u,A,L,Z,\boldsymbol W)\right|\right|_{L_2}=o_p(n^{-1/4}),\\
      & \text{sup}_{u\leq T}||\widehat{H}(u,A,L,Z,\boldsymbol W)-H(u,A,L,Z,\boldsymbol W)||_{L_2}=o_p(n^{-1/4}),\\
      &\text{sup}_{u\leq T}||\widehat{F}(u,A,L,Z,\boldsymbol W)-F(u,A,L,Z,\boldsymbol W)||_{L_2}=o_p(n^{-1/4})
\end{align*}
\end{assumption}
\noindent
The rate conditions in Assumption \ref{ass:ass10} are standard in the semiparametric inference literature. These rates can be achieved by various flexible methods, such as ensemble learners or specific machine learning algorithms, provided the underlying nuisance functions satisfy certain smoothness or structural properties. Crucially, the nuisance parameters do not individually require $n^{-1/4}$ convergence rates for the results in Theorem \ref{thm:thm6} to hold. Due to the second-order nature of the von Mises remainder, it is sufficient that the product of the convergence rates of the relevant nuisance parameter pairs is $o_p(n^{-1/2})$. 
\begin{theorem}\label{thm:thm6}
Under Assumptions \ref{ass:ass1}--\ref{ass:ass10}, as $n \rightarrow \infty$, the following properties hold:
\begin{itemize}
\item[(i)] $\|\widehat{\boldsymbol{\theta}}^t - \boldsymbol{\theta}^{*t}\| = O_p(n^{-1/3})$.
\item[(ii)] $\sqrt{n}(G^t(\widehat{\boldsymbol{\theta}^t}) - G^t(\boldsymbol{\theta}^{*t})) = o_p(1)$.
\item[(iii)] $\sqrt{n}(\widehat{G}^t(\widehat{\boldsymbol{\theta}}^t) - G^t(\boldsymbol{\theta}^{*t})) \xrightarrow{d} \mathcal{N}(0, \sigma^2)$, where $\sigma^2 = \mathbb{E}\big[\{\Delta^t_N(O)d_{\boldsymbol{\theta}^{*t}}(\boldsymbol{W}) - G^t(\boldsymbol{\theta}^{*t})\}^2\big]$.
\item[(iv)] Estimated policy gain defined as $\widehat{\Gamma}^t(\widehat{d}^t_{\boldsymbol \theta}(\boldsymbol W),\widehat{\widetilde{d}}(\boldsymbol W))=\frac{1}{n}\sum_{i=1}^n[\widehat{\Delta}^t_N(O_i)\{\widehat{d}^t_{\boldsymbol{\theta}}(\boldsymbol{W}_i) -\widetilde{d}(\boldsymbol{W}_i)\}]$ satisfies:
\begin{align*}
     &\sqrt{n}(\widehat{\Gamma}^t(\widehat{d}^t_{\boldsymbol \theta}(\boldsymbol W),\widetilde{d}(\boldsymbol W))-\Gamma^t(d^t_{\boldsymbol \theta^*}(\boldsymbol W),\widetilde{d}(\boldsymbol W)))\\
     &\hspace{3cm}\xrightarrow{d}\mathcal{N}(0,\mathbb{E}[\{\Delta^t_N(O)(d^t_{\boldsymbol \theta^*}(\boldsymbol W)-\widetilde{d}(\boldsymbol W))-\Gamma^t(d^t_{\boldsymbol \theta^*}(\boldsymbol W),\widetilde{d}(\boldsymbol W))\}^2])\cdot
\end{align*}
\end{itemize}
\end{theorem}
\noindent
The proof of this theorem is given in Supplementary Section S6. 
\section{Simulations}
\subsection{Data Generation}
We evaluate finite-sample performance across sample sizes $n \in \{1000, 1500, 2500, 5000\}$. Measured confounders $W_1, W_2 \sim \mathcal{N}(0, 1)$, a binary instrument $Z \sim \text{Ber}(0.5)$, and a DiD period indicator $L \sim \text{Ber}(0.5)$ are generated independently. Unmeasured confounding is introduced via latent variables $U_0$ and $U_1$ following a Bridge distribution. Treatment is assigned period-specifically as $A = LA_0 + (1-L)A_1$, where $A_0 \mid L, Z, W_1, W_2, U_0 \sim \text{Ber}(p_0)$ and $A_1 \mid L, Z, W_1, W_2, U_1 \sim \text{Ber}(p_1)$, with $p_0 = \text{expit}(2 - 7Z + 0.2U_0 + 2W_1)$ and $p_1 = \text{expit}(-1.5 + 5Z + 0.15U_1 + 1.5W_2)$.\\

\noindent
Censoring times follow a Weibull distribution with shape parameter 2 and scale determined by $\theta = -3 + 0.3A + 0.1L + 0.1Z + 0.2W_1 + 0.1W_2$. The terminal and recurrent event processes are specified under an additive hazards framework, with all hazard functions constrained to remain non-negative across the covariate support. The death time is constructed as $D = LD_1 + (1-L)D_0$, with period-specific hazards $\lambda^d_0 = 0.2 + [(-0.1 + 0.2W_1 - 0.2W_2)/10] A_0 - 0.01W_1 + 0.02W_2 + 0.03U_0 + 0.01Z$ and $\lambda^d_1 = 0.25 + [(-0.1 + 0.2W_1 - 0.2W_2)/10] A_1 - 0.01W_1 + 0.02W_2 + 0.03U_1 + 0.01Z$.\\

\noindent
\noindent For the recurrent event process, we consider landmark times $t \in \{1, 2, 3, 4, 5\}$. Event increments $dN^*_0(t)$ and $dN^*_1(t)$ are drawn from Poisson distributions with intensities $\lambda^n_1(t) = 0.1t + 0.2 + (-0.1 + 0.2W_1 - 0.2W_2)A_0 + 0.05W_1 + 0.03W_2 + 0.05U_0 + 0.02Z$ and $\lambda^n_0(t) = 0.1t + (-0.1 + 0.2W_1 - 0.2W_2)A_1 + 0.05W_1 + 0.03W_2 + 0.05U_0 + 0.02Z$ for $L=1$ and $L=0$, respectively. From these increments, we construct the true recurrent event process $N^*(t) = N^*(\min(t,D))$, which counts events up to death $D$, and the observed process $N(t)$, which is additionally subject to censoring.

\subsection{Evaluation Metrics}
\noindent For all simulation scenarios, we evaluate the performance of the estimated optimal treatment policies $\widehat{d}^t_{\boldsymbol{\theta}}$ at landmark times $t=1$ and $t=4$, comparing the proposed multiply robust AIPW estimator against the two IPW-based alternatives.\\

\noindent
\textit{Optimal Treatment Policy:} Due to the complex data-generating mechanism, the true optimal treatment policy $d^t_{\text{opt}}$ is not available in closed form. As a gold standard, we simulate an independent dataset of $n=10^6$ observations and solve the constrained optimization problem using the first IPW estimator applied directly to the uncensored process $N^*(t)$, yielding the benchmark policy $d^t_{\text{opt}}$. Each estimated policy $\widehat{d}^t_{\boldsymbol{\theta}}$ is then evaluated on a separate test dataset of size $10^6$ using the Percentage of Correct Decisions (PCD):
$$\text{PCD} = 1 - \frac{1}{n} \sum_{i=1}^n |\widehat{d}^t_{\boldsymbol{\theta}}(\boldsymbol{W}_i) - d^t_{\text{opt}}(\boldsymbol{W}_i)|\cdot$$

\noindent \textit{Policy Gain Estimation:} \noindent \textit{Policy Gain Estimation:} To evaluate the performance of the policy gain estimator $\widehat{\Gamma}(\widehat{d}^t_{\boldsymbol{\theta}}, \widetilde{d})$, we use the multiply robust AIPW estimator with nuisance parameters estimated via 5-fold cross-fitted machine learning models. The true policy gain is obtained from the $n=10^6$ population dataset described above, using the IPW1 estimator applied directly to the uncensored process $N^*(t)$. We consider two evaluation schemes: fixed policy evaluation, where $\widehat{\Gamma}^t(d^t_{\boldsymbol{\theta}}, \widetilde{d})$ is computed using the fixed true optimal policy across all replications to isolate the performance of the gain estimator itself, and estimated policy evaluation, where $\widehat{\Gamma}^t(\widehat{d}^t_{\boldsymbol{\theta}}, \widetilde{d})$ is computed using policies estimated within each replication to reflect the full scope of estimation uncertainty. In both cases, inferential validity is assessed via root-$n$ scaled bias and empirical coverage probability of 95\% confidence intervals derived from the limiting variance estimator in Theorem \ref{thm:thm6}(iv). All results are based on $R = 500$ independent replications per scenario.

\subsection{Results}
Figure \ref{fig:accu} displays the Percentage of Correct Decisions (PCD) for the estimated optimal treatment policies across all methods, sample sizes, and landmark times. The AIPW estimator consistently achieves the highest accuracy across all settings, with PCD increasing steadily with sample size at both $t=1$ and $t=4$. Both IPW estimators exhibit lower and more variable accuracy, particularly at $t=4$, underscoring the benefit of the multiply robust construction in recovering the true optimal policy.\\
\begin{figure}[ht]
    \centering
    \caption{Percentage of Correct Decisions (PCD) for estimated optimal treatment policies across sample sizes ($n = 1000, 1500, 2500, 5000$) and landmark times ($t = 1, 4$).}
    \label{fig:accu}
    \includegraphics[width =\linewidth]{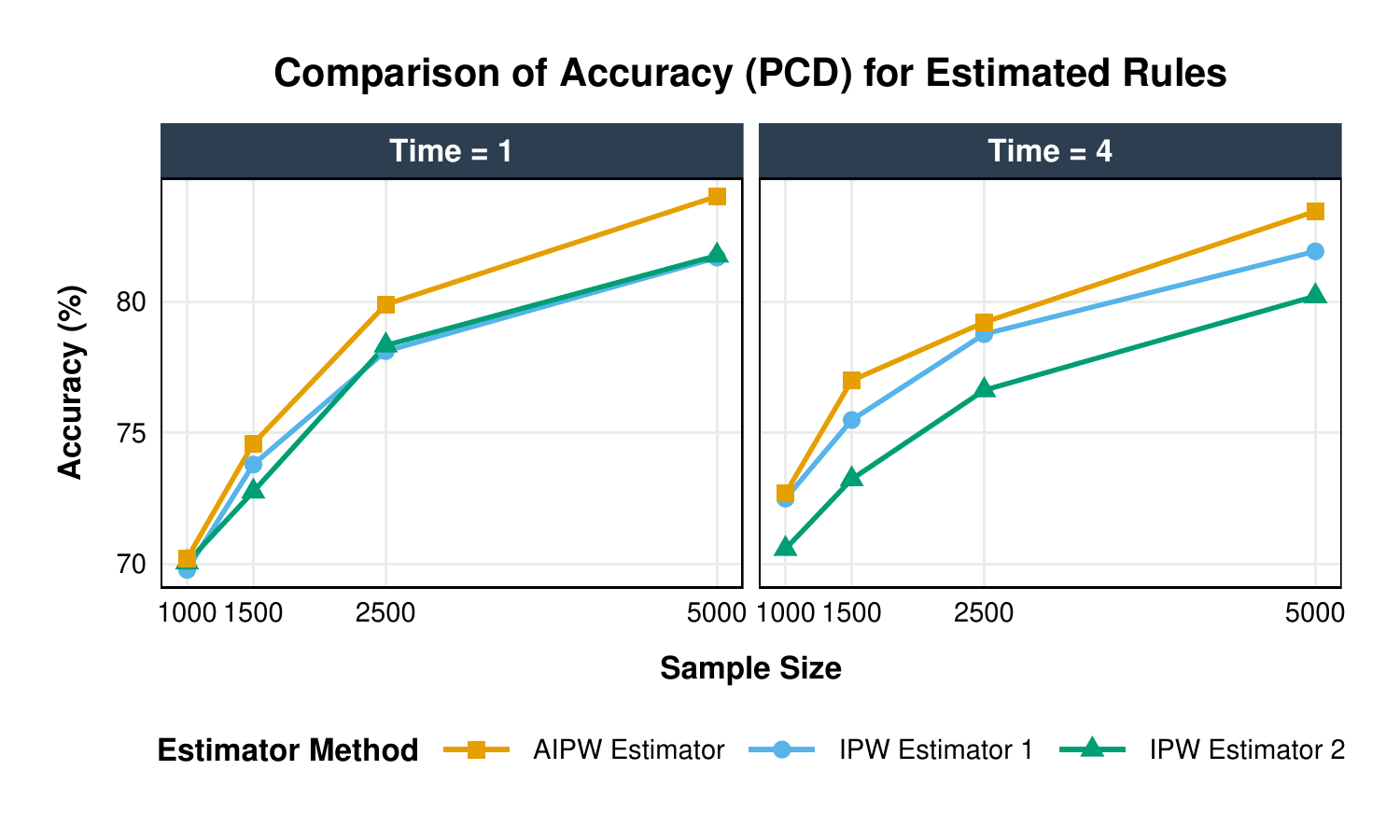}
\end{figure}

\noindent
Table \ref{tab:tab2} reports the inferential performance of the AIPW policy gain estimator under a fixed treatment policy. The estimator performs remarkably well across all sample sizes and both landmark times: AIPW estimates are virtually indistinguishable from the true values, root-$n$ scaled biases are negligible throughout, and empirical coverage probabilities remain at or near the nominal 95\% level across all configurations, ranging narrowly between 0.944 and 0.962. Notably, this strong inferential performance is sustained at both the short ($t=1$) and longer ($t=4$) landmark times, and shows no degradation as the complexity of the recurrent event process increases. These results provide compelling empirical validation of the $\sqrt{n}$-consistency and asymptotic normality established in Theorem~\ref{thm:thm6}.\\
\begin{table}[ht]
\centering
\caption{Performance of the AIPW policy gain estimator under a fixed treatment policy $d(\cdot)$. Columns report the policy value, AIPW estimate, root-$n$ scaled bias, and empirical coverage probability of 95\% confidence intervals across 500 Monte Carlo simulations.}
\label{tab:tab2}
\begin{tabular}{cccccc}
\toprule
\textbf{\makecell{Sample \\ Size}} & \textbf{Time} & \textbf{\makecell{True \\ Value}} & \textbf{\makecell{AIPW \\ Estimator}} & \textbf{\makecell{Root-$n$ \\ Bias}} & \textbf{\makecell{Coverage \\ Probability}} \\ 
\midrule
1000 & 1 & -0.078 & -0.078 &  0.034 & 0.956 \\ 
1500 & 1 & -0.078 & -0.078 & -0.013 & 0.950 \\ 
2500 & 1 & -0.078 & -0.077 &  0.078 & 0.956 \\ 
5000 & 1 & -0.078 & -0.078 &  0.022 & 0.948 \\ 
\midrule
1000 & 4 & -0.163 & -0.157 &  0.199 & 0.962 \\ 
1500 & 4 & -0.163 & -0.159 &  0.216 & 0.952 \\ 
2500 & 4 & -0.163 & -0.171 & -0.423 & 0.956 \\ 
5000 & 4 & -0.163 & -0.161 &  0.143 & 0.944 \\ 
\bottomrule
\end{tabular}
\end{table}

\noindent
When the policy gain is instead evaluated at the estimated optimal policy (Supplementary Table S1), finite-sample bias and undercoverage arise due to the non-smoothness of the estimated decision boundary, particularly at smaller sample sizes. Nonetheless, both root-$n$ bias and coverage probability improve monotonically with $n$ at both landmark times, consistent with the asymptotic theory. The residual undercoverage at moderate sample sizes is expected given the non-regularity induced by the estimated decision boundary, and is a well-documented phenomenon in the policy learning literature.

\section{Data Analysis}
In this section, we apply our methodology for estimating optimal treatment policies for recurrent outcomes to a clinical case study using Medicare service data from 2016–2023. Specifically, we compare two T2DM drug classes, namely metformin $(A=1)$ and glucagon-like peptide-1 receptor agonists (GLP-1 RA) $(A=0)$ as first-line therapies. The primary outcome is the total number (comprising both initial and recurrent episodes) of a composite endpoint, including myocardial infarction, stroke, arterial revascularization, heart failure hospitalization, end-stage kidney disease, kidney transplantation, and mortality.\\

\noindent
We defined the index date as the first instance of metformin or GLP-1 RA (including dual agonists) initiation. To maintain a focus on monotherapy, we excluded any patients who were prescribed multiple drug classes on their index date. To address unmeasured confounding, we utilized provider preference as an instrumental variable, defined by the longitudinal change in a provider's metformin prescription proportion between 2016 and 2023. We restricted the analysis to providers with at least ten records in both years to ensure a stable estimate of prescribing behavior. The resulting binary instrument reflects this temporal shift in preference; consequently, we excluded patients with index dates in 2016 or 2023 to ensure the instrument captures a distinct pre- or post-period preference relative to the treatment decision.\\

\noindent
Supplementary Table S2 summarizes the baseline characteristics for the 219,286 Medicare beneficiaries included in the study, where metformin was the predominant first-line therapy (95\%). On average, patients initiating GLP-1 RA were younger (68 vs.\ 71 years) and more likely to be female (62\% vs.\ 52\%) compared to those initiating metformin. The GLP-1 RA cohort had a slightly higher proportion of Non-Hispanic White (NHW) patients (81\% vs.\ 77\%), and a marginally higher mean Claims-based Frailty Index (CFI) score (0.133 vs.\ 0.131). Regarding clinical outcomes, the raw mortality rate was lower in the GLP-1 RA group (1.93\% vs.\ 3.29\%), while the incidence of recurrent composite events, though numerically similar across both groups (1.00\% vs.\ 1.01\%), represents a clinically meaningful difference at the scale of this Medicare cohort, where even marginal disparities in event rates translate to thousands of hospitalizations. Both groups demonstrated a high and consistent rate of administrative follow-up, with 93\% of patients remaining in the study until the end of the observation period.\\

\noindent
To determine the optimal temporal split for the instrumented difference-in-differences framework, we utilized the F-statistic method proposed by \citet{ye2023instrumented}. This diagnostic ensures that the instrument remains sufficiently strong across the selected time periods to avoid issues related to weak identification. Our analysis yielded an F-statistic well above the conventional threshold of 10, supporting the validity of the instrument. Based on this criterion, the study period was partitioned into a pre-period spanning 2017–2021 and a post-period consisting of 2022.

\subsection{Results}
Over the 2017--2022 study period, observed clinical practice strongly favored metformin as first-line therapy (95\%), reflected by a positive baseline intercept (1.04) and provider preferences toward older and male patients (age: 0.03; male: 0.43). Higher frailty and Non-Hispanic White (NHW) race were associated with reduced metformin initiation (CFI: $-0.63$; NHW: $-0.36$). A distinct temporal shift was evident by 2022: metformin initiation declined to 85.9\% (intercept: $-0.40$), with providers demonstrating a substantially stronger tendency to prescribe GLP-1 RAs to frailer patients (CFI: $-1.44$).\\

\noindent
All evaluated models consistently recommended a dramatic reduction in metformin use relative to the behavioral policy (Table~\ref{tab:policy_results}), reflecting strong algorithmic consensus that GLP-1 RAs are substantially underutilized as first-line therapy. The standard non-IV IPW approach recommends metformin for only 4.78\% of patients, favoring younger (age: $-0.01$), lower-frailty (CFI: 0.26), and non-NHW patients (NHW: $-0.58$). Adjusting for unmeasured confounding via the proposed iDID framework yields more conservative reallocation targets: the IPW1 estimator recommends metformin for 7.41\% of patients, while the multiply robust AIPW estimator recommends it for 21.16\%. The AIPW optimal policy assigns higher frailty as a driver toward GLP-1 RA (CFI: $-0.50$), and similarly discourages metformin for male (male: $-0.46$) and NHW patients (NHW: $-0.18$), yielding a more clinically nuanced reallocation relative to the non-IV approach.\\

\noindent
The estimated benefit of implementing the AIPW optimal policy is a Final Value Gain of $-0.034$ (95\% CI: $[-0.068,\, 0.001]$), indicating an expected reduction of approximately 0.034 adverse composite events per patient relative to the observed 2022 behavioral policy. Crucially, the constrained optimization framework ensures this reduction is not achieved at the expense of increased mortality, guaranteeing survival probability under the optimal policy is no lower than under the behavioral policy. These findings align with emerging clinical guidelines: for patients managing multiple comorbidities, GLP-1 RAs offer targeted cardiovascular protection and slowing of chronic kidney disease progression, benefits that metformin does not inherently provide.
\begin{table}[ht]
\centering
\caption{Estimated optimal treatment policy coefficients and recommended metformin proportion for the Medicare Type 2 diabetes cohort. Behavioral policies reflect observed prescribing patterns over the full 2017--2022 period and in 2022 alone. The Non-IV IPW, iDID IPW, and iDID AIPW columns report the optimal treatment policies derived under each method.}
\label{tab:policy_results}
\resizebox{\textwidth}{!}{%
\begin{tabular}{lcccccc}
\toprule
\textbf{Method} & \textbf{Metformin Proportion} & \textbf{Intercept} & \textbf{Age} & \textbf{CFI Score} & \textbf{Male} & \textbf{Race (NHW)} \\
\midrule
Behavioral (Overall) & 95.0\% & 1.04 & 0.03 & -0.63 & 0.43 & -0.36 \\
Behavioral (2022) & 85.9\% & -0.40 & 0.04 & -1.44 & 0.59 & -0.38 \\
Non-IV IPW & 4.78\% & 0.75 & -0.01 & 0.26 & 0.18 & -0.58 \\
iDID (IPW1) & 7.41\% & 0.64 & -0.01 & 0.21 & -0.67 & -0.31 \\
iDID (AIPW) & 21.16\% & 0.71 & -0.01 & -0.50 & -0.46 & -0.18 \\
\bottomrule
\end{tabular}%
}
\end{table}

\section{Discussion}
Our proposed iDID framework addresses a critical gap in the policy learning literature by enabling estimation of optimal treatment policies for recurrent event outcomes subject to a terminal event under unmeasured confounding. The framework's constrained optimization formulation explicitly guards against degenerate policies that reduce recurrent events only by increasing mortality, a clinically critical but frequently overlooked failure mode. We developed IPW estimators offering computational 
efficiency alongside a multiply robust AIPW estimator that achieves $\sqrt{n}$-consistency and asymptotic normality of the estimated policy gain under flexible machine learning models for nuisance parameter estimation, with simulation studies corroborating these theoretical guarantees. Applied to a national Medicare cohort of Type 2 diabetes patients, the iDID AIPW estimator yields a clinically coherent reallocation (21.16\% metformin) that appropriately directs GLP-1 RAs toward frailer, female, and non-NHW patients, in contrast to a non-IV estimator that produces an extreme and confounded reallocation demonstrating the practical consequences of ignoring unmeasured confounding in policy learning at scale.\\

\noindent
The present work opens several directions for future research. First, the current policy class is restricted to linear decision boundaries, which ensures interpretability but may not capture nonlinear covariate interactions; extending the framework to more flexible policy classes is a natural next step. Second, although the iDID identification assumptions, including period-invariance of the instrument's direct effect are weaker than the standard IV exclusion restriction, they remain untestable, and developing formal sensitivity analyses for these assumptions would strengthen the practical utility of the framework. Perhaps the most consequential generalization concerns the dynamic treatment setting: our framework currently operates with a static instrument and baseline covariates, whereas in many chronic disease applications, treatment decisions, confounders, and instruments evolve over time. 
Extending iDID policy learning to accommodate time-varying treatments and confounders, where the exclusion restriction and parallel trends analogue must hold conditionally at each decision point would naturally connect to the dynamic treatment policy literature and enable sequential treatment 
adaptation as a patient's disease state and frailty evolve over the course of follow-up.

\section*{Acknowledgments}
This work was supported by the Patient-Centered Outcomes Research Institute (PCORI). The authors thank Colleen Brensinger for her assistance with data preprocessing.\\

\noindent
{\it Conflict of Interest}: None declared.
\bibliographystyle{abbrvnat}
\bibliography{references.bib}

\newpage
\newgeometry{top=1in, bottom=1in, left=0.5in, right=0.5in}
\section*{Supplementary Materials}
\pagenumbering{arabic}

\setcounter{figure}{0}
\renewcommand{\thefigure}{S\arabic{figure}}
\setcounter{table}{0}
\renewcommand{\thetable}{S\arabic{table}}
\setcounter{equation}{0}
\renewcommand{\theequation}{S\arabic{equation}}
\setcounter{section}{0}
\renewcommand{\thesection}{S\arabic{section}}
\setcounter{theorem}{0}
\renewcommand{\thetheorem}{S\arabic{theorem}}
\section{Proof of Theorem 1}
We start evaluating the following expectation.
\begin{align*}
    &\mathbb{E}\left[\frac{\Delta\cdot(2Z-1)(2L-1)(2A-1)N(t)\{\mathbb{I}(A=d^t(\boldsymbol W))\}}{K(X-,A,L,Z,\boldsymbol W)\pi(L,Z,\boldsymbol W)\delta_A(\boldsymbol W)}\right]\\
    &=\mathbb{E}\left[\frac{\mathbb{I}(C\geq D)(2Z-1)(2L-1)(2A-1)N(t)\{\mathbb{I}(A=d^t(\boldsymbol W))\}}{K(D-,A,L,Z,\boldsymbol W)\pi(L,Z,\boldsymbol W)\delta_A(\boldsymbol W)}\right]
\end{align*}
Moreover, $\mathbb{I}(C\geq D)N(t)=\mathbb{I}(C\geq D)N^*(t)$, hence the above expression can be written as,
\begin{align*}
    &\mathbb{E}\left[\frac{(2Z-1)(2L-1)(2A-1)\mathbb{I}(C\geq D) N^*(t)\{\mathbb{I}(A=d^t(\boldsymbol W))\}}{K(D-|A,L,Z,\boldsymbol W)\pi(L,Z,\boldsymbol W)\delta_A(\boldsymbol W)}\right]\\
    & =\mathbb{E}\left[\left.\mathbb{E}\left\{\frac{(2Z-1)(2L-1)(2A-1)\mathbb{I}(C\geq D)N^*(t)\{\mathbb{I}(A=d^t(\boldsymbol W))\}}{K(D-|A,L,Z,\boldsymbol W)\pi(L,Z,\boldsymbol W)\delta_A(\boldsymbol W)}\right\}\right|A,L,Z,\boldsymbol W,D, N^*(t)\right]\\
    &  =\mathbb{E}\left[\frac{(2Z-1)(2L-1)(2A-1)\mathbb{E}[\mathbb{I}(C\geq D)|A,L,Z,\boldsymbol W,D,N^*(t)]N^*(t)\{\mathbb{I}(A=d^t(\boldsymbol W))\}}{K(D-|A,L,Z,\boldsymbol W)\pi(L,Z,\boldsymbol W)\delta_A(\boldsymbol W)}\right]
\end{align*}
Using Assumption 8 of non-informative censoring in the main text, the above expression is simplified as,
\begin{align*}
    &\mathbb{E}\left\{\frac{(2Z-1)(2L-1)(2A-1)N^*(t)\mathbb{I}(A=d^t(\boldsymbol W))}{\pi(L,Z,\boldsymbol W)\delta_A(\boldsymbol W)}\right\}\\
    &=\mathbb{E}\left\{\frac{(2Z-1)(2L-1)(2A-1)N^{*(A)}_L(t)\mathbb{I}(A=d^t(\boldsymbol W))}{\pi(L,Z,\boldsymbol W)\delta_A(\boldsymbol W)}\right\}\\
    &= \mathbb{E}\left\{\sum_{a=0,1}\frac{(2Z-1)(2L-1)(2a-1)N^{*(a)}_L(t)\mathbb{I}(A=a)\mathbb{I}(d^t(\boldsymbol W)=a)}{\pi(L,Z,\boldsymbol W)\delta_A(\boldsymbol W)}\right\}\\
    & =\mathbb{E}\left[\mathbb{E}\left\{\left.\sum_{a=0,1}\frac{(2Z-1)(2L-1)(2a-1)N^{*(a)}_L(t)\mathbb{I}(A=a)\mathbb{I}(d^t(\boldsymbol W)=a)}{\pi(L,Z,\boldsymbol W)\delta_A(\boldsymbol W)}\right|Z,L,\boldsymbol W,\boldsymbol U\right\}\right]
\end{align*}
Using the fact that $N^{*(a)}_L(t)\independent A|Z,L,\boldsymbol W,\boldsymbol U$, the above expression can be written as:
\begin{align*}
   & = \mathbb{E}\left\{\sum_{a=0,1}\frac{(2Z-1)(2L-1)(2a-1)\mathbb{E}(N^{*(a)}_L(t)|Z,L,\boldsymbol W,\boldsymbol U)P(A=a|Z,L,\boldsymbol W,\boldsymbol U)\mathbb{I}(d^t(\boldsymbol W)=a)}{\pi(L,Z,\boldsymbol W)\delta_A(\boldsymbol W)}\right\}
\end{align*}
Using Assumption 3 of Random Sampling in the main text, the above term can be written as
\begin{align*}
    &=\mathbb{E}\left\{\frac{\mathbb{E}(N^{*(1)}_1(t)|Z=1,\boldsymbol W,\boldsymbol U)P(A=1|Z=1,L=1,\boldsymbol W,\boldsymbol U)\mathbb{I}(d^t(\boldsymbol W)=1)}{\delta_A(\boldsymbol W)}\right\}\\
    &-\mathbb{E}\left\{\frac{\mathbb{E}(N^{*(1)}_0(t)|Z=1,\boldsymbol W,\boldsymbol U)P(A=1|Z=1,L=0,\boldsymbol W,\boldsymbol U)\mathbb{I}(d^t(\boldsymbol W)=1)}{\delta_A(\boldsymbol W)}\right\}\\
    &-\mathbb{E}\left\{\frac{\mathbb{E}(N^{*(1)}_1(t)|Z=0,\boldsymbol W,\boldsymbol U)P(A=1|Z=0,L=1,\boldsymbol W,\boldsymbol U)\mathbb{I}(d^t(\boldsymbol W)=1)}{\delta_A(\boldsymbol W)}\right\}\\
    & + \mathbb{E}\left\{\frac{\mathbb{E}(N^{*(1)}_0(t)|Z=0,\boldsymbol W,\boldsymbol U)P(A=1|Z=0,L=0,\boldsymbol W,\boldsymbol U)\mathbb{I}(d^t(\boldsymbol W)=1)}{\delta_A(\boldsymbol W)}\right\}\\
    &-\mathbb{E}\left\{\frac{\mathbb{E}(N^{*(0)}_1(t)|Z=1,\boldsymbol W,\boldsymbol U)P(A=0|Z=1,L=1,\boldsymbol W,\boldsymbol U)\mathbb{I}(d^t(\boldsymbol W)=0)}{\delta_A(\boldsymbol W)}\right\}\\
    &+\mathbb{E}\left\{\frac{\mathbb{E}(N^{*(0)}_0(t)|Z=1,\boldsymbol W,\boldsymbol U)P(A=0|Z=1,L=0,\boldsymbol W,\boldsymbol U)\mathbb{I}(d^t(\boldsymbol W)=0)}{\delta_A(\boldsymbol W)}\right\}\\
    &+\mathbb{E}\left\{\frac{\mathbb{E}(N^{*(0)}_1(t)|Z=0,\boldsymbol W,\boldsymbol U)P(A=0|Z=0,L=1,\boldsymbol W,\boldsymbol U)\mathbb{I}(d^t(\boldsymbol W)=0)}{\delta_A(\boldsymbol W)}\right\}\\
    & - \mathbb{E}\left\{\frac{\mathbb{E}(N^{*(0)}_0(t)|Z=0,\boldsymbol W,\boldsymbol U)P(A=0|Z=0,L=0,\boldsymbol W,\boldsymbol U)\mathbb{I}(d^t(\boldsymbol W)=0)}{\delta_A(\boldsymbol W)}\right\}\cdot
\end{align*}
The last four terms in the above expression can be written as: 
\begin{align*}
    &-\mathbb{E}\left\{\frac{\mathbb{E}(N^{*(0)}_1(t)|Z=1,\boldsymbol W,\boldsymbol U)\mathbb{I}(d^t(\boldsymbol W)=0)}{\delta_A(\boldsymbol W)}\right\}+\mathbb{E}\left\{\frac{\mathbb{E}(N^{*(0)}_0(t)|Z=1,\boldsymbol W,\boldsymbol U)\mathbb{I}(d^t(\boldsymbol W)=0)}{\delta_A(\boldsymbol W)}\right\}\\
    &+\mathbb{E}\left\{\frac{\mathbb{E}(N^{*(0)}_1(t)|Z=0,\boldsymbol W,\boldsymbol U)\mathbb{I}(d^t(\boldsymbol W)=0)}{\delta_A(\boldsymbol W)}\right\} -\mathbb{E}\left\{\frac{\mathbb{E}(N^{*(0)}_0(t)|Z=0,\boldsymbol W,\boldsymbol U)\mathbb{I}(d^t(\boldsymbol W)=0)}{\delta_A(\boldsymbol W)}\right\}\\
    &+\mathbb{E}\left\{\frac{\mathbb{E}(N^{*(0)}_1(t)|Z=1,\boldsymbol W,\boldsymbol U)P(A=1|Z=1,L=1,\boldsymbol W,\boldsymbol U)\mathbb{I}(d^t(\boldsymbol W)=0)}{\delta_A(\boldsymbol W)}\right\}\\
    &-\mathbb{E}\left\{\frac{\mathbb{E}(N^{*(0)}_0(t)|Z=1,\boldsymbol W,\boldsymbol U)P(A=1|Z=1,L=0,\boldsymbol W,\boldsymbol U)\mathbb{I}(d^t(\boldsymbol W)=0)}{\delta_A(\boldsymbol W)}\right\}\\
    &-\mathbb{E}\left\{\frac{\mathbb{E}(N^{*(0)}_1(t)|Z=0,\boldsymbol W,\boldsymbol U)P(A=1|Z=0,L=1,\boldsymbol W,\boldsymbol U)\mathbb{I}(d^t(\boldsymbol W)=0)}{\delta_A(\boldsymbol W)}\right\}\\
    & +\mathbb{E}\left\{\frac{\mathbb{E}(N^{*(0)}_0(t)|Z=0,\boldsymbol W,\boldsymbol U)P(A=1|Z=0,L=0,\boldsymbol W,\boldsymbol U)\mathbb{I}(d^t(\boldsymbol W)=0)}{\delta_A(\boldsymbol W)}\right\} \cdot
\end{align*}
The above terms can be rearranged as,
\begin{align*}
    &-\mathbb{E}\left\{\frac{\mathbb{E}(N^{*(0)}_1(t)-N^{*(0)}_0(t)|Z=1,\boldsymbol W,\boldsymbol U)\mathbb{I}(d^t(\boldsymbol W)=0)}{\delta_A(\boldsymbol W)}\right\}\\
    &+\mathbb{E}\left\{\frac{\mathbb{E}(N^{*(0)}_1(t)-N^{*(0)}_0(t)|Z=0,\boldsymbol W,\boldsymbol U)\mathbb{I}(d^t(\boldsymbol W)=0)}{\delta_A(\boldsymbol W)}\right\}\\
    &+\mathbb{E}\left\{\frac{\mathbb{E}(N^{*(0)}_1(t)|Z=1,\boldsymbol W,\boldsymbol U)P(A=1|Z=1,L=1,\boldsymbol W,\boldsymbol U)\mathbb{I}(d^t(\boldsymbol W)=0)}{\delta_A(\boldsymbol W)}\right\}\\
    &-\mathbb{E}\left\{\frac{\mathbb{E}(N^{*(0)}_0(t)|Z=1,\boldsymbol W,\boldsymbol U)P(A=1|Z=1,L=0,\boldsymbol W,\boldsymbol U)\mathbb{I}(d^t(\boldsymbol W)=0)}{\delta_A(\boldsymbol W)}\right\}\\
    &-\mathbb{E}\left\{\frac{\mathbb{E}(N^{*(0)}_1(t)|Z=0,\boldsymbol W,\boldsymbol U)P(A=1|Z=0,L=1,\boldsymbol W,\boldsymbol U)\mathbb{I}(d^t(\boldsymbol W)=0)}{\delta_A(\boldsymbol W)}\right\}\\
    & +\mathbb{E}\left\{\frac{\mathbb{E}(N^{*(0)}_0(t)|Z=0,\boldsymbol W,\boldsymbol U)P(A=1|Z=0,L=0,\boldsymbol W,\boldsymbol U)\mathbb{I}(d^t(\boldsymbol W)=0)}{\delta_A(\boldsymbol W)}\right\} 
\end{align*}
Using Assumption 5, the original expression simplifies to:
\begin{align*}
    &=\mathbb{E}\left\{\frac{\mathbb{E}(N^{*(1)}_1(t)|Z=1,\boldsymbol W,\boldsymbol U)P(A=1|Z=1,L=1,\boldsymbol W,\boldsymbol U)\mathbb{I}(d^t(\boldsymbol W)=1)}{\delta_A(\boldsymbol W)}\right\}\\
    &-\mathbb{E}\left\{\frac{\mathbb{E}(N^{*(1)}_0(t)|Z=1,\boldsymbol W,\boldsymbol U)P(A=1|Z=1,L=0,\boldsymbol W,\boldsymbol U)\mathbb{I}(d^t(\boldsymbol W)=1)}{\delta_A(\boldsymbol W)}\right\}\\
    &-\mathbb{E}\left\{\frac{\mathbb{E}(N^{*(1)}_1(t)|Z=0,\boldsymbol W,\boldsymbol U)P(A=1|Z=0,L=1,\boldsymbol W,\boldsymbol U)\mathbb{I}(d^t(\boldsymbol W)=1)}{\delta_A(\boldsymbol W)}\right\}\\
    & + \mathbb{E}\left\{\frac{\mathbb{E}(N^{*(1)}_0(t)|Z=0,\boldsymbol W,\boldsymbol U)P(A=1|Z=0,L=0,\boldsymbol W,\boldsymbol U)\mathbb{I}(d^t(\boldsymbol W)=1)}{\delta_A(\boldsymbol W)}\right\}\\
    &+\mathbb{E}\left\{\frac{\mathbb{E}(N^{*(0)}_1(t)|Z=1,\boldsymbol W,\boldsymbol U)P(A=1|Z=1,L=1,\boldsymbol W,\boldsymbol U)\mathbb{I}(d^t(\boldsymbol W)=0)}{\delta_A(\boldsymbol W)}\right\}\\
    &-\mathbb{E}\left\{\frac{\mathbb{E}(N^{*(0)}_0(t)|Z=1,\boldsymbol W,\boldsymbol U)P(A=1|Z=1,L=0,\boldsymbol W,\boldsymbol U)\mathbb{I}(d^t(\boldsymbol W)=0)}{\delta_A(\boldsymbol W)}\right\}\\
    &-\mathbb{E}\left\{\frac{\mathbb{E}(N^{*(0)}_1(t)|Z=0,\boldsymbol W,\boldsymbol U)P(A=1|Z=0,L=1,\boldsymbol W,\boldsymbol U)\mathbb{I}(d^t(\boldsymbol W)=0)}{\delta_A(\boldsymbol W)}\right\}\\
    & +\mathbb{E}\left\{\frac{\mathbb{E}(N^{*(0)}_0(t)|Z=0,\boldsymbol W,\boldsymbol U)P(A=1|Z=0,L=0,\boldsymbol W,\boldsymbol U)\mathbb{I}(d^t(\boldsymbol W)=0)}{\delta_A(\boldsymbol W)}\right\} 
\end{align*}
After rearranging we obtain that, 
\begin{align*}
    &=\mathbb{E}\!\left\{\frac{%
\left[
\begin{aligned}
&\mathbb{E}\big(N^{*(1)}_1(t)\mid Z=1,\boldsymbol W,\boldsymbol U\big)\,\mathbb{I}\{d^t(\boldsymbol W)=1\}\\
&\quad+\,\mathbb{E}\big(N^{*(0)}_1(t)\mid Z=1,\boldsymbol W,\boldsymbol U\big)\,\mathbb{I}\{d^t(\boldsymbol W)=0\}
\end{aligned}
\right]\;P\big(A=1\mid Z=1,L=1,\boldsymbol W,\boldsymbol U\big)}{\delta_A(\boldsymbol W)}\right\}\\
    &-\mathbb{E}\!\left\{\frac{%
\left[
\begin{aligned}
&\mathbb{E}\big(N^{*(1)}_0(t)\mid Z=1,\boldsymbol W,\boldsymbol U\big)\,\mathbb{I}\{d^t(\boldsymbol W)=1\}\\
&\quad+\,\mathbb{E}\big(N^{*(0)}_0(t)\mid Z=1,\boldsymbol W,\boldsymbol U\big)\,\mathbb{I}\{d^t(\boldsymbol W)=0\}
\end{aligned}
\right]\;P\big(A=1\mid Z=1,L=0,\boldsymbol W,\boldsymbol U\big)}{\delta_A(\boldsymbol W)}\right\}\\
&-\mathbb{E}\!\left\{\frac{%
\left[
\begin{aligned}
&\mathbb{E}\big(N^{*(1)}_1(t)\mid Z=0,\boldsymbol W,\boldsymbol U\big)\,\mathbb{I}\{d^t(\boldsymbol W)=1\}\\
&\quad+\,\mathbb{E}\big(N^{*(0)}_1(t)\mid Z=0,\boldsymbol W,\boldsymbol U\big)\,\mathbb{I}\{d^t(\boldsymbol W)=0\}
\end{aligned}
\right]\;P\big(A=1\mid Z=0,L=1,\boldsymbol W,\boldsymbol U\big)}{\delta_A(\boldsymbol W)}\right\}\\
& + \mathbb{E}\!\left\{\frac{%
\left[
\begin{aligned}
&\mathbb{E}\big(N^{*(1)}_0(t)\mid Z=0,\boldsymbol W,\boldsymbol U\big)\,\mathbb{I}\{d^t(\boldsymbol W)=1\}\\
&\quad+\,\mathbb{E}\big(N^{*(0)}_0(t)\mid Z=0,\boldsymbol W,\boldsymbol U\big)\,\mathbb{I}\{d^t(\boldsymbol W)=0\}
\end{aligned}
\right]\;P\big(A=1\mid Z=0,L=0,\boldsymbol W,\boldsymbol U\big)}{\delta_A(\boldsymbol W)}\right\}
\end{align*}
For any $t=0,1$ and $z=0,1$, we obtain that
\begin{align*}
    &\mathbb{E}[N^{*(1)}_l(t)|Z=z,\boldsymbol W,\boldsymbol U]\mathbb{I}\{d^t(\boldsymbol W)=1\}+\mathbb{E}[N^{*(0)}_l(t)|Z=z,\boldsymbol W,\boldsymbol U]\mathbb{I}\{d^t(\boldsymbol W)=0\}\\
    &=  \mathbb{E}[N^{*(1)}_l(t)-N^{*(0)}_l(t)|Z=z,\boldsymbol W,\boldsymbol U]\mathbb{I}\{d^t(\boldsymbol W)=1\}+ \mathbb{E}[N^{*(0)}_l(t)|Z=z,\boldsymbol W,\boldsymbol U]
\end{align*}
Using Assumption 5 in the main text, the above expression is simplified as:
\begin{align*}
   \tau^N_l(t,\boldsymbol W,\boldsymbol U)d^t(\boldsymbol W)+ \nu_l(\boldsymbol W,\boldsymbol U,Z)
\end{align*}
where $\tau^N_l(t,\boldsymbol W,\boldsymbol U)=\mathbb{E}[N^{*(1)}_l(t)-N^{*(0)}_l(t)|\boldsymbol W,\boldsymbol U]$. Note that the second term does not depend on $d^t(\boldsymbol W)$. Hence one can write
\begin{align*}
    & \mathbb{E}\left\{\frac{(2Z-1)(2L-1)(2A-1)N^*(t)\mathbb{I}(A=d^t(\boldsymbol W))}{\pi(L,Z,\boldsymbol W)\delta_A(\boldsymbol W)}\right\}\\
    & =\mathbb{E}\left\{\frac{d^t(\boldsymbol W)}{\delta_A(\boldsymbol W)}[\tau_1(\boldsymbol W,\boldsymbol U)\delta_{A,1}(\boldsymbol W,\boldsymbol U)- \tau_0(\boldsymbol W,\boldsymbol U)\delta_{A,0}(\boldsymbol W,\boldsymbol U)]\right\} + f_N
\end{align*}
where, $f_N$ does not depend on $d^t(\boldsymbol W)$ and 
\begin{align*}
    &\delta_{A,l}(\boldsymbol W,\boldsymbol U)=P(A=1|Z=1,L=l,\boldsymbol W,\boldsymbol U)-P(A=1|Z=0,L=l,\boldsymbol W,\boldsymbol U)\\
    &=P(A_l(1)=1|\boldsymbol W,\boldsymbol U)-P(A_l(0)=1|\boldsymbol W,\boldsymbol U)=E[A_l(1)-A_l(0)|\boldsymbol W,\boldsymbol U]
\end{align*}
Using Assumption 6 in the main text of no unmeasured common effect modifier, we obtain 
\begin{align*}
    \mathbb{E}[\tau^N_l(t,\boldsymbol W,\boldsymbol U)\delta_{A,l}(\boldsymbol W,\boldsymbol U)|\boldsymbol W]=\mathbb{E}[\tau^N_l(t,\boldsymbol W,\boldsymbol U)|\boldsymbol W]\mathbb{E}[\delta_{A,l}(\boldsymbol W,\boldsymbol U)|\boldsymbol W]=\tau^N_l(t,\boldsymbol W)\delta_{A,l}(\boldsymbol W)
\end{align*}
Hence we obtain 
\begin{align*}
     &\mathbb{E}_{\boldsymbol W}\left[\frac{d^t(\boldsymbol W)}{\delta_A(\boldsymbol W)}\mathbb{E}\left\{[\tau^N_1(t,\boldsymbol W,\boldsymbol U)\delta_{A,1}(\boldsymbol W,\boldsymbol U)- \tau^N_0(1,\boldsymbol W,\boldsymbol U)\delta_{A,0}(\boldsymbol W,\boldsymbol U)]|\boldsymbol W\right\}\right]+f_N\\
     &=\mathbb{E}_{\boldsymbol W}\left[\frac{d^t(\boldsymbol W)}{\delta_A(\boldsymbol W)}\{\tau^N_1(t,\boldsymbol W)\delta_{A,1}(\boldsymbol W)-\tau^N_0(t,\boldsymbol W)\delta_{A,0}(\boldsymbol W)\}\right]+f_N\\
     &=\mathbb{E}_{\boldsymbol W}\left[\frac{d^t(\boldsymbol W)}{\delta_A(\boldsymbol W)}\tau^N(t,\boldsymbol W)\{\delta_{A,1}(\boldsymbol W)-\delta_{A,0}(\boldsymbol W)\}\right]+_N=\mathbb{E}_{\boldsymbol W}[\tau^N(t,\boldsymbol W)d^t(\boldsymbol W)]+f_N
\end{align*}
The last equality is due to fact, $\delta_{A,1}(\boldsymbol W)-\delta_{A,0}(\boldsymbol W)=\delta_A(\boldsymbol W)$ and Assumption 7 of stable treatment effect over each period in the main text.\\

\noindent
Using the exactly similar steps and replacing $N(t)$ by $Y(t)$, we obtain
\begin{align*}
    \mathbb{E}\left[\frac{\Delta\cdot(2Z-1)(2L-1)(2A-1)Y(t)\{\mathbb{I}(A=d^t(\boldsymbol W))\}}{K(X-,A,L,Z,\boldsymbol W)\pi(L,Z,\boldsymbol W)\delta_A(\boldsymbol W)}\right]=\mathbb{E} \left[\tau^Y(t,\boldsymbol W)d^t(\boldsymbol W) \right]+f_Y
\end{align*}
Since both $f_N$ and $f_Y$ do not depend on $d^t(\boldsymbol W)$, we obtain that the optimization policy in equation (3) of the main text is identified by 
\begin{align*}
         d^t_{\text{opt}}&=\arg\min_{d^t\in \mathcal{D}} \mathbb{E}\left[\frac{\Delta\cdot(2Z-1)(2L-1)(2A-1)N(t)\{\mathbb{I}(A=d^t(\boldsymbol W))\}}{K(X-,A,L,Z,\boldsymbol W)\pi(L,Z,\boldsymbol W)\delta_A(\boldsymbol W)}\right]\hspace{0.2cm} \text{subject to } \\
         &\mathbb{E}\left[\frac{\Delta\cdot(2Z-1)(2L-1)(2A-1))Y(t)\{\mathbb{I}(A=d^t(\boldsymbol W))-\mathbb{I}(A=\widetilde{d}(\boldsymbol W))\}}{K(X-,A,L,Z,\boldsymbol W)\pi(L,Z,\boldsymbol W)\delta_A(\boldsymbol W)}\right]>0 \cdot
    \end{align*}
This completes the proof of Theorem 1.

\section{Proof of Theorem 2}
We start with evaluating the following expression,
\begin{align*}
    &\mathbb{E}\left[\int_{0}^t\frac{(2Z-1)(2L-1)(2A-1)dN(s)\{\mathbb{I}(A=d^t(\boldsymbol W))\}}{K(s,A,L,Z,\boldsymbol W)\pi(L,Z,\boldsymbol W)\delta_A(\boldsymbol W)}\right]\\
    &= \int_{0}^t\mathbb{E}\left[\frac{(2Z-1)(2L-1)(2A-1)dN(s)\{\mathbb{I}(A=d(\boldsymbol W))\}}{K(s,A,L,Z,\boldsymbol W)\pi(L,Z,\boldsymbol W)\delta_A(\boldsymbol W)}\right]\\
    & =  \int_{0}^t\mathbb{E}\left[\frac{(2Z-1)(2L-1)(2A-1)\mathbb{I}(C>s)dN^*(s)\{\mathbb{I}(A=d^t(\boldsymbol W))\}}{K(s,A,L,Z,\boldsymbol W)\pi(L,Z,\boldsymbol W)\delta_A(\boldsymbol W)}\right]\\
   &= \int_{0}^t\mathbb{E}\left[\left.\mathbb{E}\left\{\frac{(2Z-1)(2L-1)(2A-1)\mathbb{I}(C>s)dN^*(s)\{\mathbb{I}(A=d^t(\boldsymbol W))\}}{K(s,A,L,Z,\boldsymbol W)\pi(L,Z,\boldsymbol W)\delta_A(\boldsymbol W)}\right\}\right|A,L,Z,\boldsymbol W,dN^*(s)\right]
\end{align*}
Using Assumption 8 of non-informative censoring in the main text the above term is simplified as,
\begin{align*}
   &\int_{0}^t\mathbb{E}\left\{\frac{(2Z-1)(2L-1)(2A-1)E(\mathbb{I}(C>s)|A,L,Z,\boldsymbol W,dN^*(s))dN^*(s)\{\mathbb{I}(A=d^t(\boldsymbol W))\}}{K(s,A,L,Z,\boldsymbol W)\pi(L,Z,\boldsymbol W)\delta_A(\boldsymbol W)}\right\}\\
   &=\int_{0}^t\mathbb{E}\left\{\frac{(2Z-1)(2L-1)(2A-1)dN^*(s)\{\mathbb{I}(A=d^t(\boldsymbol W))\}}{\pi(L,Z,\boldsymbol W)\delta_A(\boldsymbol W)}\right\}\\
   &=\mathbb{E}\left\{\frac{(2Z-1)(2L-1)(2A-1)N^*(s)\{\mathbb{I}(A=d^t(\boldsymbol W))\}}{\pi(L,Z,\boldsymbol W)\delta_A(\boldsymbol W)}\right\}
\end{align*}
Hence the above expression boils down to IPW identification using Theorem 1. Hence one can proceed the remaining proof using exactly same steps as the proof of Theorem 1.

\section{Proof of Theorem 3}
We start with evaluating the following expectation. For any $l=0,1$, $z=0,1$ and $\boldsymbol w \in \text{Supp}(\boldsymbol W)$,
\begin{align*}
     &\mu_{\widetilde{N}(t)}(l,z,\boldsymbol w)=\mathbb{E}[\widetilde{N}(t)|L=l,Z=z,\boldsymbol W=\boldsymbol w]=\mathbb{E}\left.\left[\frac{\mathbb{I}(C\geq D)}{K(X-,A,L,Z,\boldsymbol W)}N(t)\right|L=l,Z=z,\boldsymbol W=\boldsymbol w\right]\\
    & =\mathbb{E}\left.\left[\frac{\mathbb{I}(C\geq D)}{K(D-,A,L,Z,\boldsymbol W)}N^*(t)\right|L=l,Z=z,\boldsymbol W=\boldsymbol w\right]\\
    &=\mathbb{E}\left[\left.\mathbb{E}\left.\left\{\frac{\mathbb{I}(C\geq D)}{K(D-,A,L,Z,\boldsymbol W)}N^*(t)\right|L,Z,\boldsymbol W,N^*(t),D,A\right\}\right|L=l,Z=z,\boldsymbol W=\boldsymbol w\right]\\
    &=\mathbb{E}\left[\left.\frac{N^*(t)}{K(D-,A,L,Z,\boldsymbol W)}\mathbb{E}\left.\left\{\mathbb{I}(C\geq D)\right|L,Z,\boldsymbol W,N^*(t),D,A\right\}\right|L=l,Z=z,\boldsymbol W=\boldsymbol w\right]\\
    &= \mathbb{E}\left[\left.\frac{N^*(t)}{K(D-,A,L,Z,\boldsymbol W)}\mathbb{E}\left.\left\{\mathbb{I}(C\geq D)\right|A,L,Z,\boldsymbol W,D\right\}\right|L=l,Z=z,\boldsymbol W=\boldsymbol w\right]\\
    &=\mathbb{E}\left[\left.\frac{N^*(t)}{K(D-,A,L,Z,\boldsymbol W)}\mathbb{E}\left.\left\{\mathbb{I}(C\geq D)\right|A,L,Z,\boldsymbol W\right\}\right|L=l,Z=z,\boldsymbol W=\boldsymbol w\right]\\
    &=\mathbb{E}\left[N^*(t)|L=l,Z=z,\boldsymbol W=\boldsymbol w\right]=\mu_{N^*(t)}(l,z,\boldsymbol w)\cdot
\end{align*}
We used Assumption 8 of non-informative censoring from the main text in the above derivation. Hence, we obtain $\delta_{\widetilde{N}(t)}(\boldsymbol w)=\delta_{N^*(t)}(\boldsymbol w)$. Next
\begin{align*}
    \delta_{N^*(t)}(\boldsymbol w)&=\mu_{N^*(t)}(1,1,\boldsymbol w)-\mu_{N^*(t)}(1,0,\boldsymbol w)-\mu_{N^*(t)}(0,1,\boldsymbol w)+\mu_{N^*(t)}(0,0,\boldsymbol w)\\
    &=\sum_{z=0,1}(2z-1)(\mathbb{E}[N^*(t)|L=1,Z=z,\boldsymbol W=\boldsymbol w]-\mathbb{E}[N^*(t)|L=0,Z=z,\boldsymbol W=\boldsymbol w])
\end{align*}
Next using Assumption 1 of consistency in the main text, the above expression becomes
\begin{align*}
    \sum_{z=0,1}(2z-1)(\mathbb{E}[N_1^{*A_1(z)}(t)|L=1,Z=z,\boldsymbol W=\boldsymbol w]-\mathbb{E}[N_0^{*A_0(z)}(t)|L=0,Z=z,\boldsymbol W=\boldsymbol w])
\end{align*}
Next using Assumption 3 of random sampling in the main text, the above expression becomes,
\begin{align*}
    &\sum_{z=0,1}(2z-1)(\mathbb{E}[N_1^{*A_1(z)}(t)-N_0^{*A_0(z)}(t)|Z=z,\boldsymbol W=\boldsymbol w])\\
    &=\sum_{z=0,1}(2z-1)(\mathbb{E}[A_1(z)N_1^{*1}(t)+(1-A_1(z))N_1^{*0}(t)-A_0(z)N_0^{*1}(t)+(1-A_0(z))N_0^{*0}(t)|Z=z,\boldsymbol W=\boldsymbol w])\\
    &=\sum_{z=0,1}(2z-1)(\mathbb{E}[A_1(z)[N_1^{*1}(t)-N_1^{*0}(t)]-A_0(z)[N_0^{*1}(t)-N_0^{*0}(t)]+N_1^{*0}(t)-N_0^{*0}(t)|Z=z,\boldsymbol W=\boldsymbol w])
\end{align*}
Using Assumption 5 of independence and exclusion restriction in the main text, the above expression becomes
\begin{align*}
    &\sum_{z=0,1}(2z-1)(\mathbb{E}[A_1(z)(N_1^{*1}(t)-N_1^{*0}(t))-A_0(z)(N_0^{*1}(t)-N_0^{*0}(t))+N_1^{*0}(t)-N_0^{*0}(t)|\boldsymbol W=\boldsymbol w])\\
    &=\mathbb{E}[(A_1(1)-A_1(0))(N_1^{*1}(t)-N_1^{*0}(t))-(A_0(1)-A_0(0))(N_0^{*1}(t)-N_0^{*0}(t))|\boldsymbol W=\boldsymbol w]
\end{align*}
Using Assumption 6 of no unmeasured common effect modifier in the main text, the above expression becomes
\begin{align*}
    &\mathbb{E}[A_1(1)-A_1(0)|\boldsymbol W=\boldsymbol w] \mathbb{E}[N_1^{*1}(t)-N_1^{*0}(t)|\boldsymbol W=\boldsymbol w]\\
    &\hspace{3cm}-\mathbb{E}[A_0(1)-A_0(0)|\boldsymbol W=\boldsymbol w]\mathbb{E}[N_0^{*1}(t)-N_0^{*0}(t)|\boldsymbol W=\boldsymbol w]
\end{align*}
Using Assumption 7 of stable treatment effect over each period in the main text, the above expression becomes
\begin{align*}
    \mathbb{E}[A_1(1)-A_1(0)-A_0(1)+A_0(0)|\boldsymbol W=\boldsymbol w]\tau^N(t,\boldsymbol w)
\end{align*}
Next we observe that
\begin{align*}
    &\delta_A(\boldsymbol w)=\mu_{A}(1,1,\boldsymbol w)-\mu_{A}(1,0,\boldsymbol w)-\mu_{A}(0,1,\boldsymbol w)+\mu_{A}(0,0,\boldsymbol w)\\
    &=\sum_{z=0,1}(2z-1)(\mathbb{E}[A|L=1,Z=z,\boldsymbol W=\boldsymbol w]-\mathbb{E}[A|L=0,Z=z,\boldsymbol W=\boldsymbol w])\\
    &=\sum_{z=0,1}(2z-1)(\mathbb{E}[A_1(z)|L=1,Z=z,\boldsymbol W=\boldsymbol w]-\mathbb{E}[A_0(z)|L=0,Z=z,\boldsymbol W=\boldsymbol w])\\
    &=\sum_{z=0,1}(2z-1)(\mathbb{E}[A_1(z)|\boldsymbol W=\boldsymbol w]-\mathbb{E}[A_0(z)|\boldsymbol W=\boldsymbol w])\\
    &=\mathbb{E}[A_1(1)-A_1(0)-A_0(1)+A_0(0)|\boldsymbol W=\boldsymbol w]
\end{align*}
Hence we obtain $\frac{\delta_{N^*(t)}(\boldsymbol w)}{\delta_A(\boldsymbol w)}=\tau^N(t,\boldsymbol w)$ and consequently 
$$\mathbb{E}\left[\frac{\delta_{\widetilde{N}(t)}(\boldsymbol W)}{\delta_A(\boldsymbol W)}\right]=\mathbb{E}\left[\frac{\delta_{N^*(t)}(\boldsymbol W)}{\delta_A(\boldsymbol W)}\right]=\mathbb{E}[\tau^N(t,\boldsymbol W)]\cdot$$ Using the exactly similar steps, we also obtain
$$\mathbb{E}\left[\frac{\delta_{\widetilde{Y}(t)}(\boldsymbol W)}{\delta_A(\boldsymbol W)}\right]=\mathbb{E}\left[\frac{\delta_{Y^*(t)}(\boldsymbol W)}{\delta_A(\boldsymbol W)}\right]=\mathbb{E}[\tau^Y(t,\boldsymbol W)]\cdot$$ 
\noindent
This completes the proof of Theorem 3.
\section{Proof of Theorem 4}
We prove that the conjectured gradient satisfies a von Mises expansion by showing that it vanishes in expectation and that the corresponding remainder is of second order. The first order term is given by,
\begin{align*}
    D_{\beta}(t,O,\mathbb{P})=&\frac{\delta_{\widetilde{N}(t)}(\boldsymbol W)}{\delta_{A}(\boldsymbol W)}(\mathbb{P})-\beta(t)+\frac{(2Z-1)(2L-1)}{\pi(L,Z,\boldsymbol W)\delta_A(\boldsymbol W)}\left[\widetilde{N}(t)-\mu_{\widetilde{N}(t)}(L,Z,\boldsymbol W)-\right.\\
    &\left.\frac{\delta_{\widetilde{N}(t)}(\boldsymbol W)}{\delta_{A}(\boldsymbol W)}\{A-\mu_{A}(L,Z,\boldsymbol W)\}+ \int_{0}^\infty \frac{F(u,t,A,L,Z,\boldsymbol W)}{H(u,A,L,Z,\boldsymbol W)}\frac{dM_C(u,A,L,Z,\boldsymbol W)}{K(u,A,L,Z,\boldsymbol W)}\right]\cdot
\end{align*}
We first show that $\mathbb{E}_{\mathbb{P}}[D^t_{\beta}(t,O,\mathbb{P})]=0$. For a known censoring mechanism $K$, it is shown in the regular iDID case \citep{ye2023instrumented} that the quantity
\begin{align*}
     D_{\beta}(t,O,\mathbb{P})-\frac{(2Z-1)(2L-1)}{\pi(L,Z,\boldsymbol W)\delta_A(\boldsymbol W)}\int_{0}^\infty \frac{F(u,t,A,L,Z,\boldsymbol W)}{H(u,A,L,Z,\boldsymbol W)}\frac{dM_C(u,A,L,Z,\boldsymbol W)}{K(u,A,L,Z,\boldsymbol W)}
\end{align*}
vanishes in expectation under ${\mathbb{P}}$. Hence in our case of estimated $K$, we just concentrate on the last term, 
\begin{align*}
    & \mathbb{E}\left[\frac{(2Z-1)(2L-1)}{\pi(L,Z,\boldsymbol W)\delta_A(\boldsymbol W)}\int_{0}^\infty \frac{F(u,t,A,L,Z,\boldsymbol W)}{H(u,A,L,Z,\boldsymbol W)}\frac{dM_C(u,A,L,Z,\boldsymbol W)}{K(u,A,L,Z,\boldsymbol W)}\right]\\
    &=\mathbb{E}_{\mathbb{P}}\left.\left[\mathbb{E}\left\{\frac{(2Z-1)(2L-1)}{\pi(L,Z,\boldsymbol W)\delta_A(\boldsymbol W)}\int_{0}^\infty \frac{F(u,t,A,L,Z,\boldsymbol W)}{H(u,A,L,Z,\boldsymbol W)}\frac{dM_C(u,A,L,Z,\boldsymbol W)}{K(u,A,L,Z,\boldsymbol W)}\right|A,L,Z,\boldsymbol W\right\}\right]\\
    &=\mathbb{E}_{\mathbb{P}}\left[\frac{(2Z-1)(2L-1)}{\pi(L,Z,\boldsymbol W)\delta_A(\boldsymbol W)}\int_{0}^\infty \frac{F(u,t,A,L,Z,\boldsymbol W)}{H(u,A,L,Z,\boldsymbol W)}\frac{d\mathbb{E}\left\{M_C(u,A,L,Z,\boldsymbol W)|A,L,Z,\boldsymbol W\right\}}{K(u,A,L,Z,\boldsymbol W)}\right]\cdot
\end{align*}
We compute the conditional expectation of the martingale,
\begin{align*}
    &\mathbb{E}_{\mathbb{P}}\left\{M_C(u,A,L,Z,\boldsymbol W)|A,L,Z,\boldsymbol W\right\}=\mathbb{E}\left.\left\{N_C(u)-\int_{(0,u]}Y^{\dagger}(s)d\Lambda_C(s,A,L,Z,\boldsymbol W)\right|A,L,Z,\boldsymbol W\right\}\\
    &=\mathbb{E}\left.\left\{N_C(u)-\int_{(0,u]}Y^{\dagger}(s)\frac{d\mathbb{E}\{N_C(s)|A,L,Z,\boldsymbol W\}}{\mathbb{E}\{Y^{\dagger}(s)|A,L,Z,\boldsymbol W\}}\right|A,L,Z,\boldsymbol W\right\}\\
    &=\mathbb{E}\{N_C(u)|A,L,Z,\boldsymbol W\}-\int_{(0,u]}\mathbb{E}\{Y^{\dagger}(s)|A,L,Z,\boldsymbol W\}\frac{d\mathbb{E}\{N_C(s)|A,L,Z,\boldsymbol W\}}{\mathbb{E}\{Y^{\dagger}(s)|A,L,Z,\boldsymbol W\}}\\
    &=\mathbb{E}\{N_C(u)|A,L,Z,\boldsymbol W\}-\int_{(0,u]}d\mathbb{E}\{N_C(s)|A,L,Z,\boldsymbol W\}=0\cdot
\end{align*}
This shows $\mathbb{E}_{\mathbb{P}}[D_{\beta}(t,O,\mathbb{P})]=0$.\\

\noindent
Next, we show that the remainder term in the von Mises expansion is of second order. The von Mises expansion of $\beta(t)$ at $\bar{\mathbb{P}}$ around $\mathbb{P}$ is given by,
\begin{align*}
    \beta(t,\bar{\mathbb{P}})- \beta(t,\mathbb{P})=(\bar{\mathbb{E}}-\mathbb{E}) D^t_{\beta}(t,O,\bar{\mathbb{P}})+R_{\beta}(t,O,\bar{\mathbb{P}}, \mathbb{P})
\end{align*}
Let $\phi(t,O,\mathbb{P})=\frac{\delta_{\widetilde{N}(t)}(\boldsymbol W)}{\delta_{A}(\boldsymbol W)}(\mathbb{P})$, hence $\beta(t,\mathbb{P})=\mathbb{E}(\phi(t,O,\mathbb{P}))$. Let $D_{\text{aug},\beta}(t,O,\mathbb{P})=d^t_{\beta}(t,O,\mathbb{P})-\phi(t,O,\mathbb{P})+\beta(t,\mathbb{P})$. Since, $\bar{\mathbb{E}}[D^t_{\beta}(t,O,\bar{\mathbb{P}})]=0$, we obtain
\begin{align*}
    &R(t,O,\bar{\mathbb{P}},\mathbb{P})=\beta(t,\bar{\mathbb{P}})- \beta(t,\mathbb{P})+\mathbb{E}[D_{\beta}(t,O,\bar{\mathbb{P}})]= \beta(t,\bar{\mathbb{P}})- \beta(t,\mathbb{P})+\mathbb{E}[\phi(t,O,\bar{\mathbb{P}})-\beta(t,\bar{\mathbb{P}})+D_{\text{aug},\beta}(t,O,\bar{\mathbb{P}})]\\
    &=\mathbb{E}(\phi(t,O,\bar{\mathbb{P}}))-\mathbb{E}(\phi(t,O,\mathbb{P}))+\mathbb{E}[D_{\text{aug},\beta}(t,O,\bar{\mathbb{P}})]
\end{align*}
\subsubsection*{Part 1}
We drop the random variables except for $\boldsymbol W$ from each parameters for notational simplicity. Since $\mathbb{E}[\tau^N(t,\boldsymbol W)]=\beta(t,\mathbb{P})$, we obtain
\begin{align*}
    &\mathbb{E}(\phi(t,O,\bar{\mathbb{P}}))-\mathbb{E}(\phi(t,O,\mathbb{P}))=\mathbb{E}\left[\frac{\bar{\delta}_{\widetilde{N}(t)}}{\bar{\delta}_{A}}-\tau^N(t)\right]=\mathbb{E}\left[\frac{1}{\bar{\delta}_{A}}[\bar{\delta}_{\widetilde{N}(t)}-\tau^N(t)\bar{\delta}_{A}]-\frac{1}{\bar{\delta}_{A}}[\delta_{\widetilde{N}(t)}-\tau^N(t)\delta_{A}]\right]\\
    &=\mathbb{E}\left[\frac{1}{\bar{\delta}_{A}}[\bar{\delta}_{\widetilde{N}(t)}-\tau^N(t)\bar{\delta}_{A}-\delta_{\widetilde{N}(t)}-\tau^N(t)\delta_{A}]\right]=\mathbb{E}\left[\frac{1}{\bar{\delta}_{A}}[\bar{\delta}_{\widetilde{N}(t)}-\delta_{\widetilde{N}(t)}-\tau^N(t)[\bar{\delta}_{A}-\delta_{A}]\right]
\end{align*}
For $M\in \{\widetilde{N}(t),A\}$, we know that
\begin{align*}
    \mathbb{E}\left.\left[\frac{(2Z-1)(2L-1)\mu_M}{\pi}\right|\boldsymbol W\right]=\delta_M, \quad \mathbb{E}\left.\left[\frac{(2Z-1)(2L-1)\bar{\mu}_M}{\pi}\right|\boldsymbol W\right]=\bar{\delta}_M,
\end{align*}
Hence the above term can be written as, 
\begin{align*}
    \mathbb{E}\left[\frac{1}{\bar{\delta}_{A}}[\bar{\delta}_{\widetilde{N}(t)}-\delta_{\widetilde{N}(t)}-\tau^N(t)[\bar{\delta}_{A}-\delta_{A}]\right]= \mathbb{E}\left[\frac{(2Z-1)(2L-1)}{\pi \bar{\delta}_{A}}[\bar{\mu}_{\widetilde{N}(t)}-\mu_{\widetilde{N}(t)}-\tau^N(t)[\bar{\mu}_{A}-\mu_{A}]]\right]
\end{align*}
\subsubsection*{Part 2}
We work the following expression,
\begin{align*}
    &\mathbb{E}\left[\frac{(2Z-1)(2L-1)}{\bar{\pi}\bar{\delta_A}}\left[\frac{\mathbb{I}(C\geq D)}{\bar{K}(X-)}N(t)-\bar{\mu}_{\widetilde{N}(t)}-\frac{\bar{\delta}_{\widetilde{N}(t)}}{\bar{\delta}_{A}}\{A-\bar{\mu}_{A}\}\right]\right]=\\
    &\mathbb{E}\left[\frac{(2Z-1)(2L-1)}{\bar{\pi}\bar{\delta_A}}\left[\frac{\mathbb{I}(C\geq D)}{\bar{K}(X-)}N(t)-\frac{\mathbb{I}(C\geq D)}{K(X-)}N(t)+\frac{\mathbb{I}(C\geq D)}{K(X-)}N(t)-\bar{\mu}_{\widetilde{N}(t)}-\frac{\bar{\delta}_{\widetilde{N}(t)}}{\bar{\delta}_{A}}\{A-\bar{\mu}_{A}\}\right]\right]
\end{align*}
\subsubsection*{Part 2.1}
\begin{align*}
    &\mathbb{E}\left[\frac{(2Z-1)(2L-1)}{\bar{\pi}\bar{\delta_A}}\left[\frac{\mathbb{I}(C\geq D)}{K(X-)}N(t)-\bar{\mu}_{\widetilde{N}(t)}-\frac{\bar{\delta}_{\widetilde{N}(t)}}{\bar{\delta}_{A}}\{A-\bar{\mu}_{A}\}\right]\right]\\
    &=  \mathbb{E}\left[\frac{(2Z-1)(2L-1)}{\bar{\pi}\bar{\delta_A}}\left[\widetilde{N}(t)-\bar{\mu}_{\widetilde{N}(t)}-\frac{\bar{\delta}_{\widetilde{N}(t)}}{\bar{\delta}_{A}}\{A-\bar{\mu}_{A}\}\right]\right]\\
    &=\mathbb{E}\left[\frac{(2Z-1)(2L-1)}{\bar{\pi}\bar{\delta_A}}\left[\mu_{\widetilde{N}(t)}-\bar{\mu}_{\widetilde{N}(t)}-\frac{\bar{\delta}_{\widetilde{N}(t)}}{\bar{\delta}_{A}}\{\mu_A-\bar{\mu}_{A}\}\right]\right]\\
    &=\mathbb{E}\left[\frac{(2Z-1)(2L-1)}{\bar{\pi}\bar{\delta_A}}\left[(\mu_{\widetilde{N}(t)}-\bar{\mu}_{\widetilde{N}(t)})-\tau^N(t)\{\mu_A-\bar{\mu}_{A}\}+\tau^N(t)\{\mu_A-\bar{\mu}_{A}\}-\frac{\bar{\delta}_{\widetilde{N}(t)}}{\bar{\delta}_{A}}\{\mu_A-\bar{\mu}_{A}\}\right]\right]
\end{align*}
\subsubsection*{Combination of Part 1 and Part 2.1}
Combining Parts 1 and 2.1, we obtain the following expression,
\begin{align*}
    &\mathbb{E}\left[\frac{(2Z-1)(2L-1)}{\bar{\delta_A}}\left[\left(\frac{1}{\pi}-\frac{1}{\bar{\pi}}\right)[\bar{\mu}_{\widetilde{N}(t)}-\mu_{\widetilde{N}(t)}-\tau^N(t)[\bar{\mu}_{A}-\mu_{A}]]\right]+\frac{1}{\bar{\pi}}\left(\tau^N(t)-\frac{\bar{\delta}_{\widetilde{N}(t)}}{\bar{\delta}_{A}}\right)(\mu_{A}-\bar{\mu}_{A})\right]
\end{align*}
\subsubsection*{Part 2.2}
At first we show that $\mathbb{E}[F(u,t,A,L,Z,\boldsymbol W)]=\mathbb{E}[F^*(u,t,A,L,Z,\boldsymbol W)]$. It is derived as following:
\begin{align*}
    & \mathbb{E}[F(u,t,A,L,Z,\boldsymbol W)]=\mathbb{E}[\mathbb{I}(X>u)\widetilde{N}(t)|A,L,Z,\boldsymbol W]\\
    & =\mathbb{E}\left[\mathbb{I}(X>u)\frac{\mathbb{I}(C\geq D)}{K(X-,A,L,Z,\boldsymbol W)}N(t)|A,L,Z,\boldsymbol W\right]\\
    &=\mathbb{E}\left[\mathbb{I}(D>u)\frac{\mathbb{I}(C\geq D)}{K(D-,A,L,Z,\boldsymbol W)}N^*(t)|A,L,Z,\boldsymbol W=\boldsymbol w\right]\\
    &=\mathbb{E}\left[\mathbb{I}(D>u)N^*(t)|A,L,Z,\boldsymbol W\right]=\mathbb{E}[F^*(u,t,A,L,Z,\boldsymbol W)]
\end{align*}
We proceed with the following expression as,
\begin{align*}
    &\mathbb{E}\left[\frac{(2Z-1)(2L-1)}{\bar{\pi}\bar{\delta_A}}\left[\frac{\mathbb{I}(C\geq D)}{\bar{K}(X-)}N(t)-\frac{\mathbb{I}(C\geq D)}{K(X-)}N(t)\right]\right]\\
    &=\mathbb{E}\left[\frac{(2Z-1)(2L-1)\mathbb{I}(C\geq D)N(t)}{\bar{\pi}\bar{\delta_A}}\left(\frac{1}{\bar{K}(X-)}-\frac{1}{K(X-)}\right)\right]\\
    &=\mathbb{E}\left[\frac{(2Z-1)(2L-1)\mathbb{I}(C\geq D)N^*(t)}{\bar{\pi}\bar{\delta_A}}\left(\frac{1}{\bar{K}(X-)}-\frac{1}{K(X-)}\right)\right]\\
    &=\mathbb{E}\left[\frac{(2Z-1)(2L-1)K(X-) N^*(t)}{\bar{\pi}\bar{\delta_A}}\left(\frac{1}{\bar{K}(X-)}-\frac{1}{K(X-)}\right)\right]\\
    &=\mathbb{E}\left[\frac{(2Z-1)(2L-1)N^*(t)}{\bar{\pi}\bar{\delta_A}}\left(\frac{K(X-)}{\bar{K}(X-)}-1\right)\right]\\
    &=-\mathbb{E}\left[\frac{(2Z-1)(2L-1)}{\bar{\pi}\bar{\delta_A}}\int_{0}^{\infty}\mathbb{I}(D>u)N^*(t)\frac{K(u)}{\bar{K}(u)}(d\Lambda_C(u)-d\bar{\Lambda}_C(u))\right]\\
    &=-\mathbb{E}\left[\frac{(2Z-1)(2L-1)}{\bar{\pi}\bar{\delta_A}}\int_{0}^{\infty}F(u,t)\frac{K(u)}{\bar{K}(u)}(d\Lambda_C(u)-d\bar{\Lambda}_C(u))\right]
\end{align*}
\noindent
In the fifth equality we used the identity $\Lambda_C(t) = - \int_{(0, t]} \frac{d K(u)}{K(u-)}$. Specifically, let
$R(u)=[\bar K(u)]^{-1} K(u).$ If \(K\) and \(\bar K\) were smooth, then
$dK(u)=-K(u)\,d\Lambda_C(u)$, and 
$d\bar K(u)=-\bar K(u)\,d\bar\Lambda_C(u)$. Then the ordinary quotient rule gives
\[
dR(u)
=
d\left\{\frac{K(u)}{\bar K(u)}\right\}
=
\frac{dK(u)}{\bar K(u)}
-
\frac{K(u)}{\bar K(u)^2}\,d\bar K(u).
\]
Substituting the two differential identities yields
\[
dR(u)
=
\frac{-K(u)\,d\Lambda_C(u)}{\bar K(u)}
-
\frac{K(u)}{\bar K(u)^2}
\{-\bar K(u)\,d\bar\Lambda_C(u)\}.
\]
Therefore,
$dR(u)
=
-\frac{K(u)}{\bar K(u)}\,d\Lambda_C(u)
+
\frac{K(u)}{\bar K(u)}\,d\bar\Lambda_C(u)$.
Hence, $dR(u)
=
R(u)\{d\bar\Lambda_C(u)-d\Lambda_C(u)\}$. 

\subsubsection*{Part 3}
We need to calculate the following term given by:
\begin{align*}
    &\mathbb{E}\left[\frac{(2Z-1)(2L-1)}{\bar{\pi}\bar{\delta_A}}\int_{0}^\infty \frac{\bar{F}(u,t)}{\bar{H}(u)}\frac{dM_C(u)}{\bar{K}(u)}\right]=\mathbb{E}\left[\frac{(2Z-1)(2L-1)}{\bar{\pi}\bar{\delta_A}}\int_{0}^\infty \frac{(dN_C(u)-Y^{\dagger}(u)d\bar{\Lambda}_C(u))}{\bar{K}(u)}\frac{\bar{F}(u,t)}{\bar{H}(u)}\right]\\
    &=\mathbb{E}\left[\frac{(2Z-1)(2L-1)}{\bar{\pi}\bar{\delta_A}}\int_{0}^\infty \frac{(d\mathbb{E}[N_C(u)|A,L,Z,\boldsymbol W]-E[Y^{\dagger}(u)|A,L,Z,\boldsymbol W])d\bar{\Lambda}_C(u))}{\bar{K}(u)}\frac{\bar{F}(u,t)}{\bar{H}(u)}\right]\\
    &=\mathbb{E}\left[\frac{(2Z-1)(2L-1)}{\bar{\pi}\bar{\delta_A}}\int_{0}^\infty \frac{(E[Y^{\dagger}(u)|A,L,Z,\boldsymbol W]d\Lambda_C-E[Y^{\dagger}(u)|A,L,Z,\boldsymbol W])d\bar{\Lambda}_C(u))}{\bar{K}(u)}\frac{\bar{F}(u,t)}{\bar{H}(u)}\right]\\
    &=\mathbb{E}\left[\frac{(2Z-1)(2L-1)}{\bar{\pi}\bar{\delta_A}}\int_{0}^\infty E[Y^{\dagger}(u)|A,L,Z,\boldsymbol W]\frac{(d\Lambda_C(u)-d\bar{\Lambda}_C)(u)}{\bar{K}(u)}\frac{\bar{F}(u,t)}{\bar{H}(u)}\right]\\
    &=\mathbb{E}\left[\frac{(2Z-1)(2L-1)}{\bar{\pi}\bar{\delta_A}}\int_{0}^\infty H(u)\frac{K(u)\bar{F}(u,t)}{\bar{K}(u)\bar{H}(u)}(d\Lambda_C(u)-d\bar{\Lambda}_C(u))\right]
\end{align*}
\subsubsection*{Combining Part 2.2 and 3}
Combining these parts we obtain, 
\begin{align*}
    \mathbb{E}\left[\frac{(2Z-1)(2L-1)}{\bar{\pi}\bar{\delta_A}}\int_{0}^\infty H(u)\frac{K(u)}{\bar{K}(u)}\left[\frac{F(u,t)}{H(u)}-\frac{\bar{F}(u,t)}{\bar{H}(u)}\right](d\Lambda_C(u)-d\bar{\Lambda}_C(u))\right]
\end{align*}
\begin{align*}
    &=\mathbb{E}\left[\frac{(2Z-1)(2L-1)}{\bar{\pi}\bar{\delta_A}}\int_{0}^\infty H(u)\left[\frac{F(u,t)}{H(u)}-\frac{\bar{F}(u,t)}{\ H(u)}\right]d\left\{\frac{K(u)-\bar{K}(u)}{\bar{K}(u)}\right\}\right]\\
    &=\mathbb{E}\left[\frac{(2Z-1)(2L-1)}{\bar{\pi}\bar{\delta_A}}\int_{0}^\infty H(u)\left[\frac{F(u,t)}{H(u)}-\frac{\bar{F}(u,t)}{\bar{H}(u)}\right]d\left\{\frac{K(u)}{\bar{K}(u)}\right\}\right]
\end{align*}
\subsection*{Final Remainder}
The final von Mises remainder is obtained as:
\begin{align*}
    &\mathbb{E}\left[\frac{(2Z-1)(2L-1)}{\bar{\delta_A}}\left[\left(\frac{1}{\pi}-\frac{1}{\bar{\pi}}\right)[\bar{\mu}_{\widetilde{N}(t)}-\mu_{\widetilde{N}(t)}-\tau^N(t)[\bar{\mu}_{A}-\mu_{A}]]\right]+\frac{1}{\bar{\pi}}\left(\tau^N(t)-\frac{\bar{\delta}_{\widetilde{N}(t)}}{\bar{\delta}_{A}}\right)(\mu_{A}-\bar{\mu}_{A})\right.\\
    &\left.\hspace{4cm}+\frac{1}{\bar{\pi}}\int_{0}^\infty H(u)\left[\frac{F(u,t)}{H(u)}-\frac{\bar{F}(u,t)}{\bar{H}(u)}\right]d\left\{\frac{K(u)}{\bar{K}(u)}\right\}\right]
\end{align*}
Hence the remainder term $R_{\beta}(t,O,\bar{\mathbb{P}}, \mathbb{P})$ depends on the pairwise products and hence it is of the second order. Using the exact same steps, we can prove the above results for $\eta(t)$ as well.

\section{Proof of Theorem 5}
Considering the second-order remainder terms from von-Mises expansions in Theorem 4, it is easy to see under $\mathcal{M}=\bigcup_{j=1}^6 \mathcal{M}_j$, 
$\mathbb{E}[\Delta^t_N(O)|\boldsymbol W]=\tau^N(t,\boldsymbol W)$ and $\mathbb{E}[\Delta^t_Y(O)|\boldsymbol W]=\tau^Y(t,\boldsymbol W)\cdot$
Hence we obtain 
    $\mathbb{E}[\Delta^t_N(O)d^t(\boldsymbol W)]=\mathbb{E}[\tau^N(t,\boldsymbol W)d^t(\boldsymbol W)] \quad \text{and}\quad \mathbb{E}[\Delta^t_Y(O)d^t(\boldsymbol W)]=\mathbb{E}[\tau^Y(t,\boldsymbol W)d^t(\boldsymbol W)]\cdot$
This proves the first part of Theorem 5. Next we evaluate the following expectation,
\begin{align*}
    \mathbb{E}[W^t_{N}\mathbb{I}\{A=d^t(\boldsymbol W)\}]&=\frac{1}{2} \mathbb{E}[W^t_{N}(2\mathbb{I}\{A=d^t(\boldsymbol W)\}-1)]+\frac{1}{2}\mathbb{E}(W^t_{N})\\
    &=\frac{1}{2} \mathbb{E}[W^t_{N}(2A-1)(2d^t(\boldsymbol W)-1)]+\frac{1}{2}\mathbb{E}(W^t_{N})\\
    &=\frac{1}{2} \mathbb{E}[\Delta^t_N(O)(2d^t(\boldsymbol W)-1)]+\frac{1}{2}\mathbb{E}(W^t_{N})\\
    &=\mathbb{E}[\Delta^t_N(O)d^t(\boldsymbol W)]-\frac{1}{2} \mathbb{E}[\Delta^t_N(O)-W^t_{N}]
\end{align*}
Hence we obtain that, 
    $\arg\min_{d^t\in \mathcal{D}}\mathbb{E}[W^t_{N}\mathbb{I}\{A=d^t(\boldsymbol W)\}]=\arg\min_{d^t\in \mathcal{D}}\mathbb{E}[\Delta^t_N(O)d^t(\boldsymbol W)].$
Similarly we can show that,
    $\mathbb{E}[W^t_{Y}\mathbb{I}\{A=d^t(\boldsymbol W)\}]=\mathbb{E}[\Delta^t_Y(O)d^t(\boldsymbol W)]-\frac{1}{2} \mathbb{E}[\Delta^t_Y(O)-W^t_{Y}]$. 
Hence we obtain that, 
\begin{align*}
    \mathbb{E}[W^t_{Y}(\mathbb{I}\{A=d^t(\boldsymbol W)\}-\mathbb{I}\{A=\widetilde{d}(\boldsymbol W)\})]=\mathbb{E}[\Delta^t_Y(O)\{d^t(\boldsymbol W)-\widetilde{d}(\boldsymbol W)\}]
\end{align*}
This proves Theorem 5.

\section{Proof of Theorem 6}
\textbf{Proof of (i).} To implement the cross-fitting procedure, we randomly partition the sample indices into $K$ disjoint folds. Without loss of generality, we focus our proof on the $K=2$ case, denoting the index sets of the two partitions as $\mathcal{I}_1$ and $\mathcal{I}_2$. Let $n_1 = |\mathcal{I}_1|$ and $n_2 = |\mathcal{I}_2|$ represent the respective cardinalities of these folds. We first define the following fold-specific quantities:
\begin{align*}
    &\Phi^{t\mathcal{I}_I}_{1N}(\boldsymbol \theta)=-\frac{1}{n_1}\sum_{i \in \mathcal{I}_1}\left[\frac{\delta_{\widetilde{N}(t)}(\boldsymbol W_i)}{\delta_{A}(\boldsymbol W_i)}+\frac{(2Z_i-1)(2L_i-1)}{\pi(L_i,Z_i,\boldsymbol W_i)\delta_A(\boldsymbol W_i)}\left[\frac{\Delta_iN_i(t)}{K(X_i-,A_i,L_i,Z_i,\boldsymbol W_i)}-\mu_{\widetilde{N}(t)}(L_i,Z_i,\boldsymbol W_i)\right.\right.\\
    &\left.\left.-\frac{\delta_{\widetilde{N}(t)}(\boldsymbol W_i)}{\delta_{A}(\boldsymbol W_i)}\{A-\mu_{A}(L_i,Z_i,\boldsymbol W_i)\}+ \int_{0}^\infty \frac{F(u,t,A_i,L_i,Z_i,\boldsymbol W_i)}{H(u,A_i,L_i,Z_i,\boldsymbol W_i)}\frac{dM_C(u,A_i,L_i,Z_i,\boldsymbol W_i)}{K(u,A_i,L_i,Z_i,\boldsymbol W_i)}\right]d^t_{\boldsymbol \theta}(\boldsymbol W_i)\right]
\end{align*}
\begin{align*}
    &\Psi^{t\mathcal{I}_I}_{1N}(\boldsymbol \theta)=\frac{1}{n_1}\sum_{i \in \mathcal{I}_1}\left[\frac{\delta_{Y(t)}(\boldsymbol W_i)}{\delta_{A}(\boldsymbol W_i)}+\frac{(2Z_i-1)(2L_i-1)}{\pi(L_i,Z_i,\boldsymbol W_i)\delta_A(\boldsymbol W_i)}\left[\frac{\Delta_iY_i(t)}{K(X_i-,A_i,L_i,Z_i,\boldsymbol W_i)}-\mu_{Y(t)}(L_i,Z_i,\boldsymbol W_i)\right.\right.\\
    &\left.\left.-\frac{\delta_{Y(t)}(\boldsymbol W_i)}{\delta_{A}(\boldsymbol W_i)}\{A-\mu_{A}(L_i,Z_i,\boldsymbol W_i)\}+ \int_{0}^\infty \frac{H(u\vee t,A_i,L_i,Z_i,\boldsymbol W_i)}{H(u,A_i,L_i,Z_i,\boldsymbol W_i)}\frac{dM_C(u,A_i,L_i,Z_i,\boldsymbol W_i)}{K(u,A_i,L_i,Z_i,\boldsymbol W_i)}\right][d^t_{\boldsymbol \theta}(\boldsymbol W_i)-\widetilde{d}(\boldsymbol W_i)]\right]
\end{align*}

\begin{align*}
    & \widehat{\Phi}^{t\mathcal{I}_1}_1(\boldsymbol \theta)=-\frac{1}{n_1}\sum_{i \in \mathcal{I}_1}\left[\frac{\widehat{\delta}_{\widetilde{N}(t)}(\boldsymbol W_i)}{\widehat{\delta}_{A}(\boldsymbol W_i)}+\frac{(2Z_i-1)(2L_i-1)}{\widehat{\pi}(L_i,Z_i,\boldsymbol W_i)\widehat{\delta}_A(\boldsymbol W_i)}\left[\frac{\Delta_iN_i(t)}{\widehat{K}(X_i-,A_i,L_i,Z_i,\boldsymbol W_i)}-\widehat{\mu}_{\widetilde{N}(t)}(L_i,Z_i,\boldsymbol W_i)\right.\right.\\
    &\left.\left.-\frac{\widehat{\delta}_{\widetilde{N}(t)}(\boldsymbol W_i)}{\widehat{\delta}_{A}(\boldsymbol W_i)}\{A-\widehat{\mu}_{A}(L_i,Z_i,\boldsymbol W_i)\}+ \int_{0}^\infty \frac{\widehat{F}(u,t,A_i,L_i,Z_i,\boldsymbol W_i)}{\widehat{H}(u,A_i,L_i,Z_i,\boldsymbol W_i)}\frac{d\widehat{M_C}(u,A_i,L_i,Z_i,\boldsymbol W_i)}{\widehat{K}(u,A_i,L_i,Z_i,\boldsymbol W_i)}\right]d^t_{\boldsymbol \theta}(\boldsymbol W_i)\right]
\end{align*}
\begin{align*}
    &\widehat{\Psi}_1^{t\mathcal{I}_1}(\boldsymbol \theta)=\frac{1}{n_1}\sum_{i \in \mathcal{I}_1}\left[\frac{\widehat{\delta}_{Y(t)}(\boldsymbol W_i)}{\widehat{\delta}_{A}(\boldsymbol W_i)}+\frac{(2Z_i-1)(2L_i-1)}{\widehat{\pi}(L_i,Z_i,\boldsymbol W_i)\widehat{\delta}_A(\boldsymbol W_i)}\left[\frac{\Delta_iY_i(t)}{\widehat{K}(X_i-,A_i,L_i,Z_i,\boldsymbol W_i)}-\widehat{\mu}_{Y(t)}(L_i,Z_i,\boldsymbol W_i)\right.\right.\\
    &\left.\left.-\frac{\widehat{\delta}_{Y(t)}(\boldsymbol W_i)}{\widehat{\delta}_{A}(\boldsymbol W_i)}\{A-\widehat{\mu}_{A}(L_i,Z_i,\boldsymbol W_i)\}+ \int_{0}^\infty \frac{\widehat{H}(u\vee t,A_i,L_i,Z_i,\boldsymbol W_i)}{\widehat{H}(u,A_i,L_i,Z_i,\boldsymbol W_i)}\frac{d\widehat{M_C}(u,A_i,L_i,Z_i,\boldsymbol W_i)}{\widehat{K}(u,A_i,L_i,Z_i,\boldsymbol W_i)}\right][d^t_{\boldsymbol \theta}(\boldsymbol W_i)-\widetilde{d}(\boldsymbol W_i)]\right]
\end{align*}
The nuisance functions are estimated from $\mathcal{I}_2$. Hence the cross fitted estimator $\widehat{\Phi}_1^t(\boldsymbol \theta)$ and $\widehat{\Psi}_1^t(\boldsymbol \theta)$ can be written as 
\begin{align*}
    \widehat{\Phi}_1^t(\boldsymbol \theta)=\frac{n_1}{n}\widehat{\Phi}^{t\mathcal{I}_1}_1(\boldsymbol \theta)+\frac{n_2}{n}\widehat{\Phi}_1^{t\mathcal{I}_2}(\boldsymbol \theta), \quad \widehat{\Psi}_1^t(\boldsymbol \theta)=\frac{n_1}{n}\widehat{\Psi}_1^{t\mathcal{I}_1}(\boldsymbol \theta)+\frac{n_2}{n}\widehat{\Psi}_1^{t\mathcal{I}_2}(\boldsymbol \theta) \cdot
\end{align*}
\noindent
To show convergences of $\widehat{\Phi}^{t\mathcal{I}_1}_1(\boldsymbol \theta)$ and $\widehat{\Psi}_1^{t\mathcal{I}_1}(\boldsymbol \theta)$ are sufficient to show convergences of $\widehat{\Phi}_1^t(\boldsymbol \theta)$ and $\widehat{\Psi}_1^t(\boldsymbol \theta)$. At first we show the following,
\begin{align*}
    \widehat{\Phi}^{t\mathcal{I}_1}_1(\boldsymbol \theta)-\Phi^{t\mathcal{I}_I}_{1N}(\boldsymbol \theta)=o_p(n^{-1/2})
\end{align*}
We can decompose the above term as 
\begin{align*}
    & \widehat{\Phi}^{t\mathcal{I}_1}_1(\boldsymbol \theta)-\Phi^{t\mathcal{I}_I}_{1N}(\boldsymbol \theta)=\\
    & -\frac{1}{n_1}\sum_{i \in \mathcal{I}_1}\left[\frac{\widehat{\delta}_{\widetilde{N}(t)}(\boldsymbol W_i)}{\widehat{\delta}_{A}(\boldsymbol W_i)}+\frac{(2Z_i-1)(2L_i-1)}{\widehat{\pi}(L_i,Z_i,\boldsymbol W_i)\widehat{\delta}_A(\boldsymbol W_i)}\left[\frac{\Delta_iN_i(t)}{\widehat{K}(X_i-,A_i,L_i,Z_i,\boldsymbol W_i)}-\widehat{\mu}_{\widetilde{N}(t)}(L_i,Z_i,\boldsymbol W_i)\right.\right.\\
    &\left.\left.-\frac{\widehat{\delta}_{\widetilde{N}(t)}(\boldsymbol W_i)}{\widehat{\delta}_{A}(\boldsymbol W_i)}\{A-\widehat{\mu}_{A}(L_i,Z_i,\boldsymbol W_i)\}+ \int_{0}^\infty \frac{\widehat{F}(u,t,A_i,L_i,Z_i,\boldsymbol W_i)}{\widehat{H}(u,A_i,L_i,Z_i,\boldsymbol W_i)}\frac{d\widehat{M_C}(u,A_i,L_i,Z_i,\boldsymbol W_i)}{\widehat{K}(u,A_i,L_i,Z_i,\boldsymbol W_i)}\right]d^t_{\boldsymbol \theta}(\boldsymbol W_i)\right]\\
    &-\frac{1}{n_1}\sum_{i \in \mathcal{I}_1}\left[\frac{\delta_{\widetilde{N}(t)}(\boldsymbol W_i)}{\delta_{A}(\boldsymbol W_i)}+\frac{(2Z_i-1)(2L_i-1)}{\pi(L_i,Z_i,\boldsymbol W_i)\delta_A(\boldsymbol W_i)}\left[\frac{\Delta_iN_i(t)}{K(X_i-,A_i,L_i,Z_i,\boldsymbol W_i)}-\mu_{\widetilde{N}(t)}(L_i,Z_i,\boldsymbol W_i)\right.\right.\\
    &\left.\left.-\frac{\delta_{\widetilde{N}(t)}(\boldsymbol W_i)}{\delta_{A}(\boldsymbol W_i)}\{A-\mu_{A}(L_i,Z_i,\boldsymbol W_i)\}+ \int_{0}^\infty \frac{F(u,t,A_i,L_i,Z_i,\boldsymbol W_i)}{H(u,A_i,L_i,Z_i,\boldsymbol W_i)}\frac{dM_C(u,A_i,L_i,Z_i,\boldsymbol W_i)}{K(u,A_i,L_i,Z_i,\boldsymbol W_i)}\right]d^t_{\boldsymbol \theta}(\boldsymbol W_i)\right]
\end{align*}
For notational simplicity, we suppress the arguments of the nuisance functions and the subscript $i$ in the calculations below and divide it into different parts. 

\subsubsection*{PART 1}
$$\left(\frac{\widehat{\delta}_{\widetilde{N}(t)}}{\widehat{\delta}_{A}}-\frac{\delta_{\widetilde{N}(t)}}{\delta_{A}}\right)d^t_{\boldsymbol \theta}\cdot$$
\subsubsection*{PART 2.1}
\begin{align*}
    &\frac{(2Z-1)(2L-1)}{\widehat{\pi}\widehat{\delta_A}}\mathbb{I}(C\geq D)N(t)\left(\frac{1}{\widehat{K}(X-)}-\frac{1}{K(X-)}\right)d^t_{\boldsymbol \theta}\\
    &=(2Z-1)(2L-1)\mathbb{I}(C\geq D)N(t)\left(\frac{1}{\widehat{K}(X-)}-\frac{1}{K(X-)}\right)\left[\left(\frac{1}{\widehat {\pi}}-\frac{1}{\pi}\right)\left(\frac{1}{\widehat {\delta_A}}-\frac{1}{\delta_A}\right)\right.\\
    &\left.\hspace{5cm}+\frac{1}{\delta_A}\left(\frac{1}{\widehat {\pi}}-\frac{1}{\pi}\right)+\frac{1}{\pi}\left(\frac{1}{\widehat {\delta_A}}-\frac{1}{\delta_A}\right)+\frac{1}{\pi\delta_A}\right]d^t_{\boldsymbol \theta}
\end{align*}
\subsubsection*{Part 2.2}
\begin{align*}
    &\frac{(2Z-1)(2L-1)}{\widehat{\pi}\widehat{\delta_A}}\left[\widetilde{N}(t)-\widehat{\mu}_{\widetilde{N}(t)}-\frac{\widehat{\delta}_{\widetilde{N}(t)}}{\widehat{\delta}_{A}}(A-\widehat{\mu}_{A})\right]- \frac{(2Z-1)(2L-1)}{\pi\delta_A}\left[\widetilde{N}(t)-\mu_{\widetilde{N}(t)}-\frac{\delta_{\widetilde{N}(t)}}{\delta_{A}}(A-\mu_{A})\right]d^t_{\boldsymbol \theta}\\
    &=(2Z-1)(2L-1)\left[\left(\frac{1}{\widehat {\pi}}-\frac{1}{\pi}\right)\left(\frac{1}{\widehat {\delta_A}}-\frac{1}{\delta_A}\right)\left(\widetilde{N}(t)-\widehat{\mu}_{\widetilde{N}(t)}-\frac{\widehat{\delta}_{\widetilde{N}(t)}}{\widehat{\delta}_{A}}(A-\widehat{\mu}_{A})\right)+\frac{1}{\delta_A}\left(\frac{1}{\widehat {\pi}}-\frac{1}{\pi}\right)\mathcal{T}_1\right.\\
    & \left.+\frac{1}{\pi}\left(\frac{1}{\widehat {\delta_A}}-\frac{1}{\delta_A}\right)\mathcal{T}_1+\frac{1}{\pi\delta_A}\mathcal{T}_2+\left(\widetilde{N}(t)-\mu_{\widetilde{N}(t)}-\frac{\delta_{\widetilde{N}(t)}}{\delta_{A}}(A-\mu_{A})\right)\left\{\frac{1}{\pi}\left(\frac{1}{\widehat {\delta_A}}-\frac{1}{\delta_A}\right)+\frac{1}{\delta_A}\left(\frac{1}{\widehat {\pi}}-\frac{1}{\pi}\right)\right\}\right.\\
    &\left.+\frac{1}{\pi\delta_A}\left(\mu_{\widetilde{N}(t)}-\widehat{\mu}_{\widetilde{N}(t)}-\frac{1}{\delta_A}(A-\mu_A)(\widehat{\delta}_{\widetilde{N}(t)}-\delta_{\widetilde{N}(t)})-\delta_{\widetilde{N}(t)}\left(\frac{1}{\widehat{\delta}_A}-\frac{1}{\delta_A}\right)(A-\mu_A)+\frac{\delta_{\widetilde{N}(t)}}{\delta_A}(\widehat{\mu}_A-\mu_A)\right)\right]d^t_{\boldsymbol \theta}
\end{align*}
In the above expression,
\begin{align*}
    \mathcal{T}_1&=\mu_{\widetilde{N}(t)}-\widehat{\mu}_{\widetilde{N}(t)}-(\widehat{\delta}_{\widetilde{N}(t)}-\delta_{\widetilde{N}(t)})\left(\frac{1}{\widehat{\delta}_A}-\frac{1}{\delta_A}\right)(A-\mu_A)-\frac{1}{\delta_A}(\widehat{\delta}_{\widetilde{N}(t)}-\delta_{\widetilde{N}(t)})(A-\mu_A)\\
    & + \frac{1}{\delta_A}(\widehat{\delta}_{\widetilde{N}(t)}-\delta_{\widetilde{N}(t)})(\widehat{\mu}_A-\mu_A)-\delta_{\widetilde{N}(t)}\left(\frac{1}{\widehat{\delta}_A}-\frac{1}{\delta_A}\right)(A-\mu_A)+\delta_{\widetilde{N}(t)}\left(\frac{1}{\widehat{\delta}_A}-\frac{1}{\delta_A}\right)(\widehat{\mu}_A-\mu_A)+\frac{\delta_{\widetilde{N}(t)}}{\delta_A}(\widehat{\mu}_A-\mu_A)\\
    \mathcal{T}_2 &=\frac{(\delta_{\widetilde{N}(t)}-\widehat{\delta}_{\widetilde{N}(t)})}{\delta_A}(\widehat{\mu}_A-\mu_A)+\left(\frac{1}{\widehat{\delta}_A}-\frac{1}{\delta_A}\right)(\widehat{\mu}_A-\mu_A)-(\delta_{\widetilde{N}(t)}-\widehat{\delta}_{\widetilde{N}(t)})\left(\frac{1}{\widehat{\delta}_A}-\frac{1}{\delta_A}\right)(A-\mu_A)
\end{align*}
\subsubsection*{Part 3.11}
\begin{align*}
    &\frac{(2Z-1)(2L-1)}{\widehat{\pi}{\widehat{\delta}_A}}\int_{0}^{\infty}dN_C(u)\left(\frac{\widehat{F}(u,t)}{\widehat{H}(u)\widehat{K}(u)}-\frac{F(u,t)}{H(u)K(u)}\right)d^t_{\boldsymbol \theta}\\
    &=(2Z-1)(2L-1)\left[\int_{0}^{\infty}dN_C(u)\left\{\left(\frac{\widehat{F}(u,t)}{\widehat{H}(u)}-\frac{F(u,t)}{H(u)}\right)\left(\frac{1}{\widehat{K}(u)}-\frac{1}{K(u)}\right)+\frac{F(u,t)}{H(u)}\left(\frac{1}{\widehat{K}(u)}-\frac{1}{K(u)}\right)\right.\right.\\
    &+\left.\left.\frac{1}{K(u)}\left(\frac{\widehat{F}(u,t)}{\widehat{H}(u)}-\frac{F(u,t)}{H(u)}\right)\right\}\right]\left[\left(\frac{1}{\widehat {\pi}}-\frac{1}{\pi}\right)\left(\frac{1}{\widehat {\delta_A}}-\frac{1}{\delta_A}\right)+\frac{1}{\delta_A}\left(\frac{1}{\widehat {\pi}}-\frac{1}{\pi}\right)+\frac{1}{\pi}\left(\frac{1}{\widehat {\delta_A}}-\frac{1}{\delta_A}\right)+\frac{1}{\pi\delta_A}\right]d^t_{\boldsymbol \theta}
\end{align*}
\subsubsection*{Part 3.12}
\begin{align*}
    &\frac{(2Z-1)(2L-1)}{\widehat{\pi}{\widehat{\delta}_A}}\int_{0}^{\infty}Y^{\dagger}(u)\left(\frac{\widehat{F}(u,t)d\widehat{\Lambda_C}(u)}{\widehat{H}\widehat{K}}-\frac{F(u,t)d\Lambda_C(u)}{H(u)K(u)}\right)d^t_{\boldsymbol \theta}\\
    &=(2Z-1)(2L-1)\left[\int_{0}^{\infty}Y^{\dagger}(u)\left\{\left(\frac{\widehat{F}(u,t)}{\widehat{H}(u)}-\frac{F(u,t)}{H(u)}\right)\left(\frac{d\widehat{\Lambda_C}(u)}{\widehat{K}(u)}-\frac{d\Lambda_C(u)}{K(u)}\right)+\frac{F(u,t)}{H(u)}\left(\frac{d\widehat{\Lambda_C}(u)}{\widehat{K}(u)}-\frac{d\Lambda_C(u)}{K(u)}\right)\right.\right.\\
    &+\left.\left.\frac{d\Lambda_C(u)}{K(u)}\left(\frac{\widehat{F}(u,t)}{\widehat{H}(u)}-\frac{F(u,t)}{H(u)}\right)\right\}\right]\left[\left(\frac{1}{\widehat {\pi}}-\frac{1}{\pi}\right)\left(\frac{1}{\widehat {\delta_A}}-\frac{1}{\delta_A}\right)+\frac{1}{\delta_A}\left(\frac{1}{\widehat {\pi}}-\frac{1}{\pi}\right)+\frac{1}{\pi}\left(\frac{1}{\widehat {\delta_A}}-\frac{1}{\delta_A}\right)+\frac{1}{\pi\delta_A}\right]d^t_{\boldsymbol \theta}
\end{align*}
\subsubsection*{Part 3.2}
\begin{align*}
    &(2Z-1)(2L-1)\left(\int_{0}^{\infty}\frac{F(u,t)}{H(u)}\frac{dM_C(u)}{K(u)}\right)\left(\frac{1}{\widehat{\pi}{\widehat{\delta}_A}}-\frac{1}{\pi\delta_A}\right)d^t_{\boldsymbol \theta}\\
    &=(2Z-1)(2L-1)\left(\int_{0}^{\infty}\frac{F(u,t)}{H(u)}\frac{dM_C(u)}{K(u)}\right)\left[\left(\frac{1}{\widehat {\pi}}-\frac{1}{\pi}\right)\left(\frac{1}{\widehat {\delta_A}}-\frac{1}{\delta_A}\right)+\frac{1}{\delta_A}\left(\frac{1}{\widehat {\pi}}-\frac{1}{\pi}\right)+\frac{1}{\pi}\left(\frac{1}{\widehat {\delta_A}}-\frac{1}{\delta_A}\right)\right]d^t_{\boldsymbol \theta}
\end{align*}
The final term is the sum of the all the above parts. Hence this sum consists of mean zero terms and product terms except the following terms.\\

\noindent
At first we start with Part I
\begin{align*}
\mathbb{E}\left[\left(\frac{\widehat{\delta}_{\widetilde{N}(t)}}{\widehat{\delta}_{A}}-\frac{\delta_{\widetilde{N}(t)}}{\delta_{A}}\right)d^t_{\boldsymbol \theta}\right]=\mathbb{E}\left[\left(\frac{(2Z-1)(2L-1)}{\pi \hat{\delta}_{A}}[\hat{\mu}_{\widetilde{N}(t)}-\mu_{\widetilde{N}(t)}-\tau(\boldsymbol W)[\hat{\mu}_{A}-\mu_{A}]]\right)d^t_{\boldsymbol \theta}\right]
\end{align*}
Then we consider $\mathbb{E}\left(\frac{(2Z-1)(2L-1)}{\pi \delta_{A}}\left[\mu_{\widetilde{N}(t)}-\widehat{\mu}_{\widetilde{N}(t)}+\frac{\delta_{\widetilde{N}(t)}}{\delta_A}(\widehat{\mu}_A-\mu_A)\right]d^t_{\boldsymbol \theta}\right)$ from Part 2.2. Then we add Part 1 and terms from Part 2.2 together and obtain, 
\begin{align*}
    \mathbb{E}\left[\left(\frac{(2Z-1)(2L-1)}{\pi}\left(\frac{1}{\hat{\delta}_{A}}-\frac{1}{\delta_{A}}\right)[\hat{\mu}_{\widetilde{N}(t)}-\mu_{\widetilde{N}(t)}-\tau(\boldsymbol W)[\hat{\mu}_{A}-\mu_{A}]]\right)d^t_{\boldsymbol \theta}\right]
\end{align*}
Next we consider the term from Part 2.1
\begin{align*}
&=\mathbb{E}\left[\frac{(2Z-1)(2L-1)\mathbb{I}(C\geq D)N(t)}{\pi\delta_A}\left(\frac{1}{\widehat{K}(X-)}-\frac{1}{K(X-)}\right)d^t_{\theta}\right]\\
&=-\mathbb{E}\left[\frac{(2Z-1)(2L-1)}{\pi\delta_A}\int_{0}^{\infty}\frac{K(u)}{\widehat{K}(u)}F(u,t)(d\Lambda_C(u)-d\widehat{\Lambda}_C(u))\right]
\end{align*}
Also we consider these two terms from 3.11 given by
\begin{align*}
    &\mathbb{E}\left[\frac{(2Z-1)(2L-1)}{\pi\delta_A}\int_{0}^{\infty}dN_C(u)\left\{\frac{F(u,t)}{H(u)}\left(\frac{1}{\widehat{K}(u)}-\frac{1}{K(u)}\right)+\frac{1}{K(u)}\left(\frac{\widehat{F}(u,t)}{\widehat{H}(u)}-\frac{F(u,t)}{H(u)}\right)\right\}\right]\\
    &= \mathbb{E}\left[\frac{(2Z-1)(2L-1)}{\pi\delta_A}\int_{0}^{\infty}\mathbb{E}[Y^{\dagger}(u)|A,L,Z,\boldsymbol W]d\Lambda_C(u)\left\{\frac{F(u,t)}{H(u)}\left(\frac{1}{\widehat{K}(u)}-\frac{1}{K(u)}\right)\right.\right.\\
    &\left.\left.\hspace{4cm}+\frac{1}{K(u)}\left(\frac{\widehat{F}(u,t)}{\widehat{H}(u)}-\frac{F(u,t)}{H(u)}\right)\right\}d^t_{\theta}\right]
\end{align*}
Also we consider these two terms from 3.12 given by
\begin{align*}
    \mathbb{E}\left[\frac{(2Z-1)(2L-1)}{\pi\delta_A}\int_{0}^{\infty}Y^{\dagger}\left\{\frac{F(u,t)}{H(u)}\left(\frac{d\widehat{\Lambda_C}(u)}{\widehat{K}(u)}-\frac{d\Lambda_C(u)}{K(u)}\right)+\frac{d\Lambda_C(u)}{K(u)}\left(\frac{\widehat{F}(u,t)}{\widehat{H}(u)}-\frac{F(u,t)}{H(u)}\right)\right\}d^t_{\theta}\right]
\end{align*}
Adding the above terms we obtain that 
\begin{align*}
   \mathbb{E}\left[\frac{(2Z-1)(2L-1)}{\pi\delta_A}\int_{0}^{\infty}\frac{K(u)}{\widehat{K}(u)}F(u,t)(d\Lambda_C(u)-d\widehat{\Lambda}_C(u))\right]
\end{align*}
Hence all these terms from 2.1, 3.11 and 3.12 cancel out totally. For the product terms, we use Cauchy Schwarz inequality and the rate conditions in Assumption 10 in the main text to show they are $o_p(n^{-1/2})$. For the mean zero terms we use the following lemma from  \citet{kennedy2020sharp} to prove the rate of $o_p(n^{-1/2})$.
\begin{lemma}
    Consider two independent samples $O_1 = (O_1, \dots, O_n)$ and $O_2 = (O_{n+1}, \dots, O_{\tilde{n}})$, let $\hat{f}(o)$ be a function estimated from $O_2$ and $\mathbb{P}_n$ the empirical measure over $O_1$, then we have$$(\mathbb{P}_n - P)(\hat{f} - f) = O_P \left( \frac{\|\hat{f} - f\|}{\sqrt{n}} \right)$$
\end{lemma}
\noindent
Hence it is proved that, $ \widehat{\Phi}^{t\mathcal{I}_1}_1(\boldsymbol \theta)-\Phi^{t\mathcal{I}_I}_{1N}(\boldsymbol \theta)=o_p(n^{-1/2})$. Next, we can write 
\begin{align*}
    \widehat{\Psi}_1^{t\mathcal{I}_1}(\boldsymbol \theta)-\Psi^{\mathcal{I}_I}_{1N}(\boldsymbol \theta)=
    & \frac{1}{n_1}\sum_{i \in \mathcal{I}_1}[\widehat{\Delta}^t_Y(O_i)-\Delta^t_Y(O_i)](d^t_{\boldsymbol \theta}(\boldsymbol W_i)-\widetilde{d}(\boldsymbol W_i))
\end{align*}
where 
\begin{align*}
&\widehat{\Delta}^t_Y(O_i)=\frac{\widehat{\delta}_{Y(t)}(\boldsymbol W_i)}{\widehat{\delta}_{A}(\boldsymbol W_i)}+\frac{(2Z_i-1)(2L_i-1)}{\widehat{\pi}(L_i,Z_i,\boldsymbol W_i)\widehat{\delta}_A(\boldsymbol W_i)}\left[\frac{\Delta_iY_i(t)}{\widehat{K}(X_i-,A_i,L_i,Z_i,\boldsymbol W_i)}-\widehat{\mu}_{Y(t)}(L_i,Z_i,\boldsymbol W_i)\right.\\
    &\left.-\frac{\widehat{\delta}_{Y(t)}(\boldsymbol W_i)}{\widehat{\delta}_{A}(\boldsymbol W_i)}\{A-\widehat{\mu}_{A}(L_i,Z_i,\boldsymbol W_i)\}+ \int_{0}^\infty \frac{\widehat{H}(u\vee t,A_i,L_i,Z_i,\boldsymbol W_i)}{\widehat{H}(u,A_i,L_i,Z_i,\boldsymbol W_i)}\frac{d\widehat{M_C}(u,A_i,L_i,Z_i,\boldsymbol W_i)}{\widehat{K}(u,A_i,L_i,Z_i,\boldsymbol W_i)}\right]\\
&\Delta^t_Y(O_i)=\frac{\delta_{Y(t)}(\boldsymbol W_i)}{\delta_{A}(\boldsymbol W_i)}+\frac{(2Z_i-1)(2L_i-1)}{\pi(L_i,Z_i,\boldsymbol W_i)\delta_A(\boldsymbol W_i)}\left[\frac{\Delta_iY_i(t)}{K(X_i-,A_i,L_i,Z_i,\boldsymbol W_i)}-\mu_{Y(t)}(L_i,Z_i,\boldsymbol W_i)\right.\\
    &\left.-\frac{\delta_{Y(t)}(\boldsymbol W_i)}{\delta_{A}(\boldsymbol W_i)}\{A-\mu_{A}(L_i,Z_i,\boldsymbol W_i)\}+ \int_{0}^\infty \frac{H(u\vee t,A_i,L_i,Z_i,\boldsymbol W_i)}{H(u,A_i,L_i,Z_i,\boldsymbol W_i)}\frac{dM_C(u,A_i,L_i,Z_i,\boldsymbol W_i)}{K(u,A_i,L_i,Z_i,\boldsymbol W_i)}\right]
\end{align*}
Using the same technique as before, we can show $\widehat{\Psi}_1^{t\mathcal{I}_1}(\boldsymbol \theta)-\Psi^{\mathcal{I}_I}_{1n}(\boldsymbol \theta)$ is $o_p(n^{-1/2})$.
Next by WLLN, we can show that 
\begin{align*}
   \Phi^t_{1n}(\boldsymbol \theta)-\Phi^t_1(\boldsymbol\theta)=o_p(1), \quad \Psi^t_{1n}(\boldsymbol \theta)-\Psi^t_1(\boldsymbol\theta)=o_p(1)
\end{align*}
Hence we obtain that 
\begin{align*}
     \widehat{\Phi}_1^t(\boldsymbol \theta)-\Phi^t_1(\boldsymbol\theta)=o_p(1), \quad  \widehat{\Psi}_1^t(\boldsymbol \theta)-\Psi^t_1(\boldsymbol\theta)=o_p(1)
\end{align*}
Next we show that under Assumptions 9 and 10 in the main text, for any $\alpha>0$,
\begin{align*}
     \widehat{\xi}^t_1(\boldsymbol \theta)-\xi^t_1(\boldsymbol\theta)=o_p(n^{-\alpha})
\end{align*}
Fix any $\epsilon >0$, 
\begin{align*}
    & P(n^{\alpha}|\widehat{\xi}^t_1(\boldsymbol \theta)-\xi^t_1\boldsymbol\theta)|>\epsilon)= P(n^{\alpha}|\mathbb{I}(\widehat{\Psi}_1^t(\boldsymbol \theta)<0)-\mathbb{I}(\Psi^t_1(\boldsymbol \theta)<0)|>\epsilon)\\
    & P(\mathbb{I}(\widehat{\Psi}_1^t(\boldsymbol \theta)<0)\neq\mathbb{I}(\Psi^t_1(\boldsymbol \theta)<0))\leq P(|\widehat{\Psi}_1^t(\boldsymbol \theta)-\Psi_1(\boldsymbol\theta)|>|\Psi_1(\boldsymbol\theta)|)\rightarrow 0.
\end{align*}
The last holds since $|\Psi^t_1(\boldsymbol \theta)|>0$ under Assumption 9(iii) in the main text. Since $G^t(\boldsymbol{\theta}) = \Phi^t_1(\boldsymbol{\theta}) - \lambda \cdot \xi^t_1(\boldsymbol{\theta})$, hence we obtain that 
\begin{align*}
    \widehat{G}^t(\boldsymbol \theta)-G^t(\boldsymbol\theta)=o_p(1)
\end{align*}
Since under Assumption 9(ii) in the main text, both $\Phi^t_1(\boldsymbol \theta)$ and $\Psi^t_1(\boldsymbol \theta)$ is twice differentiable in $\mathcal{N}$ and since $|\Psi^t_1(\boldsymbol \theta)|>0$, either $\Psi^t_1(\boldsymbol \theta)>0$ or $\Psi^t_1(\boldsymbol \theta)<0$ for all $\theta \in \mathcal{N}$, then $\xi^t_1\boldsymbol \theta)$ is also twice differentiable in $\mathcal{N}$. Hence $G^t(\boldsymbol \theta)$ is also twice differentiable in $\mathcal{N}$. Next we established that $\widehat{G}^t(\boldsymbol \theta)-G^t(\boldsymbol\theta)=o_p(1)$ and since $\widehat{\boldsymbol \theta}$ maximizes $\widehat{G}^t(\boldsymbol \theta)$, hence we have that $\widehat{G}^t(\widehat{\boldsymbol \theta})\geq \sup_{\boldsymbol \theta}\widehat{G}^t(\boldsymbol \theta)$. Hence by the Argmax theorem, we obtain that $\widehat{\boldsymbol \theta} \xrightarrow{p} \boldsymbol \theta^*$. Next we establish the $n^{-1/3}$ convergence rate for $\widehat{\boldsymbol \theta}$. We wish Theorem 14.4 from \citet{kosorok2008introduction} to establish this. For that we need to satisfy three conditions.\\

\noindent
\textbf{Condition 1:} We need to show that $G^t(\boldsymbol \theta)-G^t(\boldsymbol \theta^*)\leq -c_1||\boldsymbol \theta-\boldsymbol \theta^*||^2$ for some $c_1>0$ and for every $\boldsymbol \theta \in  \mathcal{N}$.\\

\noindent
Since $G^t(\boldsymbol \theta)$ is twice continuously differentiable in $\mathcal{N}$, we obtain by Taylor Series expansion,
\begin{align*}
    G^t(\boldsymbol \theta)-G^t(\boldsymbol \theta^*)&=G^{t'}(\boldsymbol \theta^*)||\boldsymbol \theta-\boldsymbol \theta^*||+\frac{1}{2}G^{t''}(\boldsymbol \theta^*)||\boldsymbol \theta-\boldsymbol \theta^*||^2 +o_p(||\boldsymbol \theta-\boldsymbol \theta^*||^2)\\
    &=\frac{1}{2}G^{t''}(\boldsymbol \theta^*)||\boldsymbol \theta-\boldsymbol \theta^*||^2 +o_p(||\boldsymbol \theta-\boldsymbol \theta^*||^2)
\end{align*}
Since $\boldsymbol \theta^*$ maximizes $G^t(\boldsymbol \theta)$, hence $G^{t'}(\boldsymbol \theta^*)=0$ and $G^{t''}(\boldsymbol \theta^*)<0$. Let $c_1=-\frac{1}{2}G^{t''}(\boldsymbol \theta^*)>0$. Hence we obtain that $G^t(\boldsymbol \theta)-G^t(\boldsymbol \theta^*)\leq -c_1||\boldsymbol \theta-\boldsymbol \theta^*||^2$ for some $c_1>0$ and for every $\boldsymbol \theta \in  \mathcal{N}$.\\

\noindent
\textbf{Condition 2:} For all $N$ large enough and for all small $\delta>0$, there exist a $c_2>0$ such that
\begin{align*}
    \mathbb{E}[\sqrt{n}\sup_{||\boldsymbol \theta-\boldsymbol \theta^*||_2<\delta}| \widehat{G}^t(\boldsymbol \theta)-G^t(\boldsymbol \theta)-[\widehat{G}^t(\boldsymbol \theta^*)-G^t(\boldsymbol \theta^*)])|]\leq c_2\delta^{1/2}.
\end{align*}
and $\phi_n$ such that $\frac{\phi_n(\delta)}{\delta^{\alpha}}$ is decreasing for some $\alpha<2$ not depending on $n$.
\begin{align*}
    & \mathbb{E}[\sqrt{n}\sup_{||\boldsymbol \theta-\boldsymbol \theta^*||_2<\delta}| \widehat{G}^t(\boldsymbol \theta)-G^t(\boldsymbol \theta)-[\widehat{G}^t(\boldsymbol \theta^*)-G^t(\boldsymbol \theta^*)]||\\
    & \leq  \mathbb{E}[\sqrt{n}\sup_{||\boldsymbol \theta-\boldsymbol \theta^*||_2<\delta}| \widehat{\Phi}_1^t(\boldsymbol \theta)-\Phi^t_1(\boldsymbol \theta)-[\widehat{\Phi}_1^t(\boldsymbol \theta^*)-G^t(\boldsymbol \theta^*)]|+ \lambda\mathbb{E}[\sqrt{n}\sup_{||\boldsymbol \theta-\boldsymbol \theta^*||_2<\delta}| \widehat{\xi}^t_1(\boldsymbol \theta)-\xi^t_1\boldsymbol \theta)-[\widehat{\xi}^t_1(\boldsymbol \theta^*)-\xi^t_1\boldsymbol \theta^*)]|
\end{align*}
We first deal with the second term. We already established in the previous part of the proof that for every $\boldsymbol \theta$, $\widehat{\xi}^t_1(\boldsymbol \theta)-\xi^t_1\boldsymbol \theta)=o_p(n^{-\alpha})$. Hence the second term is negligible and we only concentrate on the first term.
\begin{align*}
     &\mathbb{E}[\sqrt{n}\sup_{||\boldsymbol \theta-\boldsymbol \theta^*||_2<\delta}| \widehat{\Phi}_1^t(\boldsymbol \theta)-\Phi^t_1(\boldsymbol \theta)-[\widehat{\Phi}_1^t(\boldsymbol \theta^*)-G^t(\boldsymbol \theta^*)]|\\
     & \leq \mathbb{E}[\sqrt{n}\sup_{||\boldsymbol \theta-\boldsymbol \theta^*||_2<\delta}| \widehat{\Phi}_1^t(\boldsymbol \theta)-\Phi^t_{1n}(\boldsymbol \theta)-[\widehat{\Phi}^t_1(\boldsymbol \theta^*)-\Phi^t_{1n}(\boldsymbol \theta^*)]|\\
     &+\mathbb{E}[\sqrt{n}\sup_{||\boldsymbol \theta-\boldsymbol \theta^*||_2<\delta}| \Phi^t_{1n}(\boldsymbol \theta)-\Phi^t_1(\boldsymbol \theta)-[\Phi^t_{1n}(\boldsymbol \theta^*)-G^t(\boldsymbol \theta^*)]|
\end{align*}
The first part is $o_p(1)$ which has been shown earlier. Hence we deal with the second part only.
\begin{align*}
    &\Phi^t_{1n}(\boldsymbol \theta)-\Phi^t_{1n}(\boldsymbol \theta^*)=-\frac{1}{n}\sum_{i =1}^n\left[\frac{\delta_{\widetilde{N}(t)}(\boldsymbol W_i)}{\delta_{A}(\boldsymbol W_i)}+\frac{(2Z_i-1)(2L_i-1)}{\pi(L_i,Z_i,\boldsymbol W_i)\delta_A(\boldsymbol W_i)}\left[\frac{\Delta_iN_i(t)}{K(X_i-,A_i,L_i,Z_i,\boldsymbol W_i)}\right.\right.\\
    &\left.\left.\hspace{4cm}-\mu_{\widetilde{N}(t)}(L_i,Z_i,\boldsymbol W_i)-\frac{\delta_{\widetilde{N}(t)}(\boldsymbol W_i)}{\delta_{A}(\boldsymbol W_i)}\{A-\mu_{A}(L_i,Z_i,\boldsymbol W_i)\}\right.\right.\\
    &\left.\left.\hspace{4cm}+ \int_{0}^\infty \frac{F(u,t,A_i,L_i,Z_i,\boldsymbol W_i)}{H(u,A_i,L_i,Z_i,\boldsymbol W_i)}\frac{dM_C(u,A_i,L_i,Z_i,\boldsymbol W_i)}{K(u,A_i,L_i,Z_i,\boldsymbol W_i)}\right][d^t_{\boldsymbol \theta}(\boldsymbol W_i)-d^t_{\boldsymbol \theta^*}(\boldsymbol W_i)]\right]
\end{align*}
At first we define
\begin{align*}
    \Delta^t_N(O)&=\frac{\delta_{\widetilde{N}(t)}(\boldsymbol W)}{\delta_{A}(\boldsymbol W)}+\frac{(2Z-1)(2L-1)}{\pi(L,Z,\boldsymbol W)\delta_A(\boldsymbol W)}\left[\frac{\Delta N(t)}{K(X-,A,L,Z,\boldsymbol W)}-\mu_{\widetilde{N}(t)}(L,Z,\boldsymbol W)\right.\\
    &\left.-\frac{\delta_{\widetilde{N}(t)}(\boldsymbol W)}{\delta_{A}(\boldsymbol W)}\{A-\mu_{A}(L,Z,\boldsymbol W)\}+ \int_{0}^\infty \frac{F(u,t,A,L,Z,\boldsymbol W)}{H(u,A,L,Z,\boldsymbol W)}\frac{dM_C(u,A,L,Z,\boldsymbol W)}{K(u,A,L,Z,\boldsymbol W)}\right]
\end{align*}
We define a class of functions given by 
\begin{align*}
    &\mathcal{F}_{\boldsymbol\theta}=\left\{||\boldsymbol \theta-\boldsymbol \theta^*||_2<\delta:\Delta^t_N(O)[d^t_{\boldsymbol \theta}(\boldsymbol W)-d^t_{\boldsymbol \theta^*}(\boldsymbol W)]\right\}
\end{align*}
Let $M_1=\sup|\Delta^t_N(O)|$ and $M_1<\infty$ using Assumptions 9 and 10 in the main text. Next we show if $||\boldsymbol \theta-\boldsymbol \theta^*||_2<\delta$, there exists a $0<k_0<\infty$ such that
\begin{align*}
    d^t_{\boldsymbol \theta}(\boldsymbol W)-d^t_{\boldsymbol \theta^*}(\boldsymbol W)=\mathbb{I}\{-k_0\delta\leq(1,\boldsymbol W')\boldsymbol \theta^*\leq k_0\delta\}
\end{align*}
Using Assumption 9(i) in the main text, we obtain that $(1,\boldsymbol W')\boldsymbol (\boldsymbol\theta-\boldsymbol\theta^*)<k_0\delta$ for some $0<k_0<\infty$. Next when $-k_0\delta\leq(1,\boldsymbol W')\boldsymbol \theta^*\leq k_0\delta$, $|d^t_{\boldsymbol \theta}(\boldsymbol W)-d^t_{\boldsymbol \theta^*}(\boldsymbol W)|\geq 1=\mathbb{I}\{-k_0\delta\leq(1,\boldsymbol W')\boldsymbol \theta^*\leq k_0\delta\}$.\\

\noindent
When $(1,\boldsymbol W')\boldsymbol \theta^*>k_0\delta>0$, we have $(1,\boldsymbol W')\boldsymbol \theta^*=(1,\boldsymbol W')(\boldsymbol\theta-\boldsymbol \theta^*)+(1,\boldsymbol W')\boldsymbol \theta^*>0$, hence $|d^t_{\boldsymbol \theta}(\boldsymbol W)-d^t_{\boldsymbol \theta^*}(\boldsymbol W)|\geq 0=\mathbb{I}\{-k_0\delta\leq(1,\boldsymbol W')\boldsymbol \theta^*\leq k_0\delta\}.$\\

\noindent
When $(1,\boldsymbol W')\boldsymbol \theta^*<-k_0\delta<0$, we have $(1,\boldsymbol W')\boldsymbol \theta^*=(1,\boldsymbol W')(\boldsymbol\theta-\boldsymbol \theta^*)+(1,\boldsymbol W')\boldsymbol \theta^*<0$, hence $|d^t_{\boldsymbol \theta}(\boldsymbol W)-d^t_{\boldsymbol \theta^*}(\boldsymbol W)|\geq 0=\mathbb{I}\{-k_0\delta\leq(1,\boldsymbol W')\boldsymbol \theta^*\leq k_0\delta\}.$\\

\noindent
Hence we define the envelope of $\mathcal{F}_{\boldsymbol\theta}$ as $F=M_1\mathbb{I}\{-k_0\delta\leq(1,\boldsymbol W')\boldsymbol \theta^*\leq k_0\delta\}$. Using Assumption 9(iv) in the main text, we obtain that 
\begin{align*}
    ||F||_{p,2}=\sqrt{\mathbb{E}[F^2]}=M_1\sqrt{P(-k_0\delta\leq(1,\boldsymbol W')\boldsymbol \theta^*\leq k_0\delta)}\leq M_1\sqrt{k_0k_1}\delta^{1/2}
\end{align*}
Next using Lemma 9.6 and Lemma 9.9 of \citet{kosorok2008introduction}, $\mathcal{F}_{\boldsymbol\theta}$ is a class indicator functions is a Vapnik-Cervonenkis (VC) class with bounded bracketing entropy $J^*_{[]}(1,\mathcal{F}_{\boldsymbol\theta})$.
\begin{align*}
    \mathbb{G}_n\mathcal{F}_{\boldsymbol\theta}=\frac{1}{\sqrt{n}}\sum_{i=1}^n[\mathcal{F}_{\boldsymbol\theta}-\mathbb{E}(\mathcal{F}_{\boldsymbol\theta})]=\sqrt{n}[\Phi^t_{1n}(\boldsymbol \theta)-\Phi^t_{1n}(\boldsymbol \theta^*)-[\Phi^t_1(\boldsymbol \theta)-G^t(\boldsymbol \theta^*)]]
\end{align*}
Using theorem 11.2 of \citet{kosorok2008introduction}, there exist a constant $0<k_2<\infty$ such that 
\begin{align*}
    \mathbb{E}(\sup_{||\boldsymbol \theta-\boldsymbol \theta^*||_2<\delta}|\mathbb{G}_n\mathcal{F}_{\boldsymbol\theta}|)\leq k_2J^*_{[]}(1,\mathcal{F}_{\boldsymbol\theta})M_1\sqrt{k_0k_1}\delta^{1/2}=c\delta^{1/2}
\end{align*}
where $0<c<\infty$. Hence we obtain that
\begin{align*}
    &\mathbb{E}[\sqrt{n}\sup_{||\boldsymbol \theta-\boldsymbol \theta^*||_2<\delta}| \widehat{\Phi}_1^t(\boldsymbol \theta)-\Phi^t_1(\boldsymbol \theta)-[\widehat{\Phi}_1^t(\boldsymbol \theta^*)-G^t(\boldsymbol \theta^*)]|\leq c\delta^{1/2}\\
    & \implies \mathbb{E}[\sqrt{n}\sup_{||\boldsymbol \theta-\boldsymbol \theta^*||_2<\delta}| \widehat{G}^t(\boldsymbol \theta)-G^t(\boldsymbol \theta)-[\widehat{G}^t(\boldsymbol \theta^*)-G^t(\boldsymbol \theta^*)]|\leq c\delta^{1/2}
\end{align*}
Using notations of Theorem 14.4 of \citet{kosorok2008introduction}, let $\phi_n(\delta)=\delta^{1/2}$. Let $\alpha=1<2$ and $\frac{\phi_n(\delta)}{\delta^{\alpha}}=\delta^{-1/2}$ is decreasing and $\alpha$ do not depend on $n$.\\

\noindent
\textbf{Condition 3:} $r_n^2\phi_n(r_n^{-1})\leq c_3\sqrt{n}$ for every $n$ and some $c_3<\infty$. $\widehat{\boldsymbol \theta}\xrightarrow{p}\boldsymbol \theta^*$ and $\widehat{G}^t(\widehat{\boldsymbol \theta})\geq \sup_{\boldsymbol \theta}\widehat{G}^t(\boldsymbol \theta)$. We choose $r_n=n^{1/3}$, then $r_n^2\phi_n(r_n^{-1})=n^{1/2}$. Hence condition 3 is satisfied. Hence $n^{1/3}||\widehat{\boldsymbol \theta}-\boldsymbol \theta^*||=O_p(1)$. This proves the first part of Theorem 1.\\

\noindent
\textbf{Proof of (ii).} We start with the Taylor series expansion,
\begin{align*}
    \sqrt{n}(G^t(\widehat{\boldsymbol \theta})-G^t(\boldsymbol \theta^*))&=\sqrt{n}\left\{G^{t'}(\boldsymbol \theta^*)||\widehat{\boldsymbol \theta}-\boldsymbol \theta^*||+\frac{1}{2}G^{t''}(\boldsymbol \theta^*)||\widehat{\boldsymbol \theta}-\boldsymbol \theta^*||^2 +o_p(||\widehat{\boldsymbol \theta}-\boldsymbol \theta^*||^2)\right\}\\
    &=\sqrt{n}\left\{\frac{1}{2}G^{t''}(\boldsymbol \theta^*)||\widehat{\boldsymbol \theta}-\boldsymbol \theta^*||^2 +o_p(||\widehat{\boldsymbol \theta}-\boldsymbol \theta^*||^2)\right\}\\
    &=\sqrt{n}\left\{\frac{1}{2}G^{t''}(\boldsymbol \theta^*)O_p(n^{-2/3}) +o_p(n^{-2/3})\right\}=o_p(1)\cdot
\end{align*}
\textbf{Proof of (iii).} We start with 
\begin{align*}
    \sqrt{n}(\widehat{G}^t(\widehat{\boldsymbol \theta})-G^t(\boldsymbol \theta^*))=  \sqrt{n}(\widehat{G}^t(\widehat{\boldsymbol \theta})-\widehat{G}^t(\boldsymbol \theta^*))+  \sqrt{n}(\widehat{G}^t(\boldsymbol \theta^*)-G^t(\boldsymbol \theta^*))
\end{align*}
The first part can be written as
\begin{align*}
    &\sqrt{n}(\widehat{G}^t(\widehat{\boldsymbol \theta})-\widehat{G}^t(\boldsymbol \theta^*))=  \sqrt{n}(G^t(\widehat{\boldsymbol \theta})-G^t(\boldsymbol \theta^*))+  \sqrt{n}(\widehat{G}^t(\widehat{\boldsymbol \theta})-\widehat{G}^t(\boldsymbol \theta^*))-\{G^t(\widehat{\boldsymbol \theta})-G^t(\boldsymbol \theta^*)\})\\
    &=o_p(1)+\sqrt{n}(\widehat{G}^t(\widehat{\boldsymbol \theta})-\widehat{G}^t(\boldsymbol \theta^*))-\{G^t(\widehat{\boldsymbol \theta})-G^t(\boldsymbol \theta^*)\})
\end{align*}
Since $n^{1/3}||\widehat{\boldsymbol \theta}-\boldsymbol \theta^*||=O_p(1)$, we have a $\delta_1=k_4n^{-1/3}$ where $k_4<\infty$. Therefore,
\begin{align*}
    & \sqrt{n}(\widehat{G}^t(\widehat{\boldsymbol \theta})-\widehat{G}^t(\boldsymbol \theta^*))-\{G^t(\widehat{\boldsymbol \theta})-G^t(\boldsymbol \theta^*)\})\leq \mathbb{E}[\sqrt{n}\sup_{||\widehat{\boldsymbol \theta}-\boldsymbol \theta^*||_2<\delta_1}| \widehat{G}^t(\widehat{\boldsymbol \theta})-G^t(\widehat{\boldsymbol \theta})-\{\widehat{G}^t(\boldsymbol \theta^*))-G^t(\boldsymbol \theta^*)\}|]\\
    & \quad\leq c_3\delta_1^{\frac{1}{2}}\leq c_3\sqrt{k_4}N^{-1/6}=o_p(1)\cdot
\end{align*}
Hence to obtain the asymptotic distribution of $\sqrt{n}(\widehat{G}^t(\widehat{\boldsymbol \theta})-G^t(\boldsymbol \theta^*))$, we need to find the same for $\sqrt{n}(\widehat{G}^t(\boldsymbol \theta^*)-G^t(\boldsymbol \theta^*))$.
\begin{align*}
    \sqrt{n}(\widehat{G}^t(\boldsymbol \theta^*)-G^t(\boldsymbol \theta^*))&=\sqrt{n}(\widehat{G}^t(\boldsymbol \theta^*)-G^t_N(\boldsymbol \theta^*))+\sqrt{n}(G^t_N(\boldsymbol \theta^*)-G^t(\boldsymbol \theta^*))\\
    &=o_p(1)+\sqrt{n}(G^t_N(\boldsymbol \theta^*)-G^t(\boldsymbol \theta^*))\\
    &=\sqrt{n}(\Phi^t_{1n}(\boldsymbol \theta^*)-G^t(\boldsymbol \theta^*))\xrightarrow{d}\mathcal{N}(0,\mathbb{E}[\{\Delta^t_N(O)d^t_{\boldsymbol \theta^*}(\boldsymbol W)-G^t(\boldsymbol \theta^*)\}^2])\cdot
\end{align*}
Hence we obtain that,
\begin{align*}
    \sqrt{n}(\widehat{G}^t(\widehat{\boldsymbol \theta})-G^t(\boldsymbol \theta^*))\xrightarrow{d}\mathcal{N}(0,\mathbb{E}[\{\Delta^t_N(O)d^t_{\boldsymbol \theta^*}(\boldsymbol W)-G^t(\boldsymbol \theta^*)\}^2])
\end{align*}
\textbf{Proof of (iv).} Let $\Gamma^t(d^t_{\boldsymbol \theta}(\boldsymbol W),\widetilde{d}(\boldsymbol W))=V^t_N(d^t_{\boldsymbol \theta})-V^t_N(\widetilde{d})$, identified by $\mathbb{E}[\Delta^t_N(O)\{d^t_{\boldsymbol \theta}(\boldsymbol W)-\widetilde{d}(\boldsymbol W)\}]$. The estimated policy gain is
\begin{align*}
\widehat{\Gamma}^t(\widehat{d}^t_{\boldsymbol \theta}(\boldsymbol W),\widetilde{d}(\boldsymbol W))=\frac{1}{n}\sum_{i=1}^n\widehat{\Delta}^t_N(O_i)\{\widehat{d}^t_{\boldsymbol \theta}(\boldsymbol W_i)-\widetilde{d}(\boldsymbol W_i)\}.
\end{align*}
We decompose:
\begin{align*}
&\sqrt{n}(\widehat{\Gamma}^t(\widehat{d}^t_{\boldsymbol \theta}(\boldsymbol W),\widetilde{d}(\boldsymbol W))-\Gamma^t(d^t_{\boldsymbol \theta^*}(\boldsymbol W),\widetilde{d}(\boldsymbol W)))\\
=&\sqrt{n}(\widehat{\Gamma}^t(\widehat{d}^t_{\boldsymbol \theta}(\boldsymbol W),\widetilde{d}(\boldsymbol W))-\widehat{\Gamma}^t(d^t_{\boldsymbol \theta^*}(\boldsymbol W),\widetilde{d}(\boldsymbol W)))+\sqrt{n}(\widehat{\Gamma}^t(d^t_{\boldsymbol \theta^*}(\boldsymbol W),\widetilde{d}(\boldsymbol W))-\Gamma^t(d^t_{\boldsymbol \theta^*}(\boldsymbol W),\widetilde{d}(\boldsymbol W))).
\end{align*}
The first term:
\begin{align*}
\widehat{\Gamma}^t(\widehat{d}^t_{\boldsymbol \theta}(\boldsymbol W),\widetilde{d}(\boldsymbol W))-\widehat{\Gamma}^t(d^t_{\boldsymbol \theta^*},\widetilde{d})&=\frac{1}{n}\sum_{i=1}^n\widehat{\Delta}^t_N(O_i)(\widehat{d}^t_{\boldsymbol \theta}(\boldsymbol W_i)-d^t_{\boldsymbol \theta^*}(\boldsymbol W_i))
\end{align*}
is $o_p(n^{-1/2})$ using the same steps as the proof of part (iii). Hence we obtain $\sqrt{n}(\widehat{\Gamma}^t(\widehat{d}^t_{\boldsymbol \theta}(\boldsymbol W),\widetilde{d}(\boldsymbol W))-\widehat{\Gamma}^t(d^t_{\boldsymbol \theta^*}(\boldsymbol W),\widetilde{d}(\boldsymbol W)))=o_p(1)$. The second term is given by
\begin{align*}
   \frac{1}{n}\sum_{i=1}^n\widehat{\Delta}^t_N(O_i)\{d^t_{\boldsymbol \theta}(\boldsymbol W_i)-\widetilde{d}(\boldsymbol W_i)\}-\mathbb{E}[\{\Delta^t_{N}(O)(d^t_{\boldsymbol \theta^*}(\boldsymbol W)-\widetilde{d}(\boldsymbol W))]
\end{align*}
Hence by the central limit theorem we obtain,
\begin{align*}
&\sqrt{n}(\widehat{\Gamma}^t(d^t_{\boldsymbol \theta^*}(\boldsymbol W),\widetilde{d}(\boldsymbol W))-\Gamma^t(d^t_{\boldsymbol \theta^*}(\boldsymbol W),\widetilde{d}(\boldsymbol W)))\\
&\hspace{3cm}\xrightarrow{d}\mathcal{N}(0,\mathbb{E}[\{\Delta^t_{N}(O)(d^t_{\boldsymbol \theta^*}(\boldsymbol W)-\widetilde{d}(\boldsymbol W))-\Gamma^t(d^t_{\boldsymbol \theta^*}(\boldsymbol W),\widetilde{d}(\boldsymbol W))\}^2])\cdot
\end{align*}
Combining both terms, we conclude:
\begin{align*}
&\sqrt{n}(\widehat{\Gamma}^t(\widehat{d}^t_{\boldsymbol \theta}(\boldsymbol W),\widetilde{d}(\boldsymbol W))-\Gamma^t(d^t_{\boldsymbol \theta^*}(\boldsymbol W),\widetilde{d}(\boldsymbol W)))\\
&\hspace{3cm}\xrightarrow{d}\mathcal{N}(0,\mathbb{E}[\{\Delta^t_{N}(O)(d^t_{\boldsymbol \theta^*}(\boldsymbol W)-\widetilde{d}(\boldsymbol W))-\Gamma^t(d^t_{\boldsymbol \theta^*}(\boldsymbol W),\widetilde{d}(\boldsymbol W))\}^2])\cdot
\end{align*}
\section{Supplementary Tables}
\begin{table}[H]
\centering
\caption{Comparison of performance of the AIPW Policy Gain Estimator in which the optimal policy are estimated at different simulation runs.}
\label{tab:tab3}
\vspace{0.3cm}
\begin{tabular}{cccccc}
\toprule
\textbf{\makecell{Sample \\ Size}} & \textbf{Time} & \textbf{\makecell{True \\ Value}} & \textbf{\makecell{AIPW \\ Estimator}} & \textbf{\makecell{Root-$n$ \\ Bias}} & \textbf{\makecell{Coverage \\ Probability}} \\ 
\midrule
1000 & 1 & -0.085 & -0.044 & 1.303 & 0.860 \\ 
1500 & 1 & -0.085 & -0.055 & 1.150 & 0.874 \\ 
2500 & 1 & -0.085 & -0.063 & 1.101 & 0.882 \\ 
5000 & 1 & -0.085 & -0.072 & 0.904 & 0.904 \\ 
\midrule
1000 & 4 & -0.196 & -0.116 & 2.545 & 0.888 \\ 
1500 & 4 & -0.196 & -0.120 & 2.937 & 0.880 \\ 
2500 & 4 & -0.196 & -0.152 & 2.229 & 0.904 \\ 
5000 & 4 & -0.196 & -0.168 & 2.004 & 0.920 \\ 
\bottomrule
\end{tabular}
\end{table}
\begin{table}[htbp]
\centering
\caption{Summary of baseline demographic and clinical characteristics for the Medicare study cohort 2017-2022 ($N = 219,286$), stratified by first-line treatment initiation (Metformin vs. GLP-1 receptor agonists). Continuous variables are presented as mean (standard deviation), and categorical variables as percentages. NHW: Non-Hispanic White; CFI: Claims-based Frailty Index.}
\label{tab:baseline}
\vspace{0.3cm}
\begin{tabular}{lcc}
\toprule
\textbf{Characteristic} & \textbf{Metformin (95.03\%)} & \textbf{GLP-1 (4.97\%)} \\ 
\midrule
Count                   & 208,383                     & 10,903                 \\
Age, Mean (SD)          & 71 (10)                     & 68 (10)                \\
Male                    & 48\%                        & 38\%                   \\
NHW                     & 77\%                        & 81\%                   \\
CFI Score, Mean (SD)    & 0.131 (0.04)                & 0.133 (0.04)           \\
Delta                   & 93\%                        & 93\%                   \\
Death                   & 3.29\%                       & 1.93\%                  \\
Recurrent Event         & 1.01\%                       & 1.00\%                  \\ 
\bottomrule
\end{tabular}
\end{table}

\end{document}